\documentclass[11pt]{article}
\usepackage[dvipsnames]{xcolor}
\definecolor{Gred}{RGB}{219, 50, 54}
\definecolor{ToCgreen}{RGB}{0, 128, 0}

\usepackage[margin=1.0in]{geometry}
\usepackage[T1]{fontenc}
\usepackage[scale=0.97]{XCharter}

\usepackage{babel}
\usepackage[spacing=true,kerning=true,babel=true,tracking=true]{microtype}

\usepackage[normalem]{ulem}
\usepackage{qcircuit}
\usepackage{graphicx}
\usepackage{amsmath,amssymb,amsthm,mathrsfs,amsfonts,dsfont}
\usepackage{mathtools}
\usepackage{subfigure, epsfig}
\usepackage{braket}
\usepackage{bm}
\usepackage{enumerate}
\usepackage{xcolor}
\usepackage{comment}
\usepackage{scalefnt}
\usepackage{authblk}
\usepackage{nicematrix}
\usepackage{booktabs}
\definecolor{Gred}{RGB}{219, 50, 54}
\definecolor{ToCgreen}{RGB}{0, 128, 0}
\usepackage[colorlinks]{hyperref}
\usepackage{cleveref}
\hypersetup{
  colorlinks = true,
  citecolor  = ToCgreen,
  linkcolor  = Sepia,
  filecolor  = Gred,
  urlcolor   = Gred
  }
\usepackage{makecell,multirow,tabularx,tabulary}

\usepackage{pagecolor}
\usepackage{lmodern}
\usepackage[utf8]{inputenc}

\usepackage{tikzscale}
\usepackage{pgfplots}
\pgfplotsset{compat=newest}
\usetikzlibrary{plotmarks}
\usetikzlibrary{arrows.meta}
\usetikzlibrary{external}
\usepgfplotslibrary{patchplots}
\usepackage{grffile}
\usepackage{enumitem}

\DeclareMathOperator{\tr}{Tr}
\DeclareMathOperator{\Tr}{Tr}

\DeclarePairedDelimiter\ceil{\lceil}{\rceil}

\newcommand{\coleq}{\mathrel{\mathop:}\nobreak\mkern-1.2mu=}

\newcommand{\mc}{\mathcal}

\newcommand{\mr}{\mathrm}

\newcommand{\mbb}{\mathbb}

\newcommand{\ketbra}[2]{{\vert #1 \rangle\! \langle #2 \vert}}

\newcommand{\E}{\mathop{\mbb{E}}}

\newcommand{\bb}{\begin{equation}\begin{aligned}\hspace{0pt}}
\newcommand{\bbb}{\begin{equation*}\begin{aligned}}
\newcommand{\ee}{\end{aligned}\end{equation}}
\newcommand{\eee}{\end{aligned}\end{equation*}}
\newcommand{\eqt}[1]{\stackrel{\mathclap{\text{\scriptsize \mbox{#1}}}}{=}}
\newcommand{\leqt}[1]{\stackrel{\mathclap{\text{\scriptsize \mbox{#1}}}}{\leq}}
\newcommand{\geqt}[1]{\stackrel{\mathclap{\text{\scriptsize \mbox{#1}}}}{\geq}}

\makeatletter
\newcommand{\pushright}[1]{\ifmeasuring@#1\else\omit\hfill$\displaystyle#1$\fi\ignorespaces}
\makeatother

\newtheorem{lemma}{Lemma}

\newtheorem{proposition}{Proposition}

\newtheorem{theorem}{Theorem}

\definecolor{applegreen}{rgb}{0.55, 0.71, 0.0}

\newcommand{\comments}[1]{}
\usepackage{algorithm}  
\usepackage[noend]{algpseudocode}  
\newcommand{\algorithmfootnote}[2][\footnotesize]{%
  \let\old@algocf@finish\@algocf@finish%
  \def\@algocf@finish{\old@algocf@finish%
    \leavevmode\rlap{\begin{minipage}{\linewidth}
    #1#2
    \end{minipage}}%
  }%
}
\usepackage{xparse}
\NewDocumentCommand{\LeftComment}{s m}{%
  \Statex \IfBooleanF{#1}{\hspace*{\ALG@thistlm}}\(\triangleright\) #2}
\algnewcommand{\LineComment}[1]{\Statex // #1}

\begin{document}

\title{What resources are needed for\\ optimal learning of bosonic Gaussian states?}
\author[1]{Senrui Chen\thanks{\texttt{csenrui@gmail.com}}}
\author[2]{Francesco Anna Mele\thanks{\texttt{francesco.mele@sns.it}}}
\author[3]{Marco Fanizza}
\author[1]{Alfred Li}
\author[1]{Zachary Mann}
\author[1]{Hsin-Yuan~Huang}
\author[1]{Yanbei Chen}
\author[1]{John Preskill\thanks{\texttt{preskill@caltech.edu}}}
\affil[1]{\small \textit{Institute for Quantum Information and Matter, Caltech, Pasadena, CA 91125, USA}}
\affil[2]{\small \textit{Scuola Normale Superiore, Piazza dei Cavalieri 7, 56126 Pisa, Italy}}
\affil[3]{\small \textit{Inria, T\'el\'ecom Paris -- LTCI, Institut Polytechnique de Paris, France}}
\date{\today}
\maketitle

\begin{abstract}
Continuous-variable systems enable key quantum technologies in computation, communication, and sensing.
Bosonic Gaussian states emerge naturally in various such applications, including gravitational-wave and dark-matter detection. A fundamental question is how to characterize an unknown bosonic Gaussian state from as few samples as possible. Despite decades-long exploration, the ultimate efficiency limit remains unclear.
In this work, we study the necessary and sufficient number of copies to learn an $n$-mode Gaussian state, with energy less than $E$, to $\varepsilon$ trace distance closeness with high probability. A particular focus is what quantum resources and algorithmic capacity are needed to achieve the optimal sample complexity.

We prove a lower bound of $\Omega(n^3/\varepsilon^2)$ for Gaussian measurements, matching the best known upper bound up to doubly-log energy dependence, and ${\Omega}(n^2/\varepsilon^2)$ for arbitrary measurements.
We further show an upper bound of $\widetilde{O}(n^2/\varepsilon^2)$ given that the Gaussian state is promised to be either pure or passive. 
Interestingly, while Gaussian measurements suffice for nearly optimal learning of pure Gaussian states, non-Gaussian measurements are provably required for optimal learning of passive Gaussian states.
Finally, focusing on learning single-mode Gaussian states via non-entangling Gaussian measurements, we provide a nearly tight bound of $\widetilde\Theta(E/\varepsilon^2)$ for any non-adaptive schemes, showing adaptivity is indispensable for nearly energy-independent scaling.
As a key technical tool of independent interest, we establish sharp bounds on the trace distance between Gaussian states in terms of the total variation distance between their Wigner functions. In particular, this yields a nearly tight sample complexity of $\widetilde{\Theta}(n^{2}/\varepsilon^{2})$ for learning the Wigner distribution of any Gaussian state to $\varepsilon$ total variation distance, achievable with Gaussian measurements.
Our results significantly advance quantum learning theory in the bosonic regimes and have practical impact in quantum sensing and benchmarking applications.
\end{abstract}

\clearpage

\newpage

\tableofcontents

\newpage
\section{Introduction}\label{sec:intro}

\noindent Bosonic (continuous-variable) quantum information studies information processing in infinite-dimensional quantum systems, which describe a wide spectrum of physical systems including optical and microwave photons, vibrational phonons, motional modes of trapped atoms and ions, and axions. 
The field is as old as the notion of entanglement, tracing back to the EPR paradox~\cite{einstein1935can}.
Today, bosonic systems continue to be a competitive platform for cutting-edge quantum technologies, including fault-tolerant quantum computation based on bosonic codes~\cite{gottesman2000encoding,sivak2023real,campagne2020quantum}, demonstration of quantum advantage~\cite{aaronson2011computational,hamilton2017gaussian,zhong2020quantum,oh2024entanglement,Liu_2025}, continuous-variable quantum key distribution (CV-QKD)~\cite{grosshans2003quantum,jouguet2013experimental,Pirandola_2017,Mele_2025_natphot}, and interferometer-based quantum metrology for gravitational wave detections~\cite{aasi2015advanced,tse2019quantum}.

An important class of bosonic quantum states are known as \emph{Gaussian states}, defined as Gibbs states of quadratic Hamiltonians in the bosonic quadrature operators, or equivalently, states with Gaussian Wigner functions (to be defined later). 
The role of Gaussian states in bosonic quantum information is as fundamental as the role of Gaussian distributions in classical statistics and probability: both arise universally through the central limit theorem~\cite{beigi2024optimal,beigi2025monotonicity} and admit a rich analytic toolkit~\cite{weedbrook2012gaussian}. 
As a concrete example, in interferometer-based sensing setups such as LIGO, the signal-carrying output light is well approximated by a Gaussian state when the signal is modeled as a well-behaved classical stochastic process~\cite{gardner2025stochastic}. Gaussian states also play a crucial role in Gaussian Boson Sampling~\cite{hamilton2017gaussian} and Gaussian-modulation CV-QKD proposals~\cite{grosshans2003quantum}.

One fundamental question in quantum information is \emph{quantum state tomography}. That is, to learn a complete description of an unknown state by measuring many copies of it. The minimal number of copies to complete learning up to certain figures of merit (usually the trace distance) is called the \emph{sample complexity} of this task.
Experimentally speaking, a smaller sample complexity means one can measure a signal to a given precision using fewer experimental runs or shorter sensing time, depending on the concrete setup.
While tomography of bosonic states has been explored both in theory and experiments~\cite{serafini2023quantum}, the ultimate sample complexity is not well-understood, in contrast to its discrete-variable counterpart~\cite{ODonnell2016,Haah2017,chen2023does}. 
Recent work~\cite{mele2025learning} proves tomography of general bosonic states is extremely hard even with energy constraints, while tomography of Gaussian states can be done using a far more manageable number of samples~\cite{mele2025learning,Bittel_2025,fanizza2025,bittel2025energy}. {However, it remains open whether the current state-of-the-art protocols of~\cite{bittel2025energy} achieve the optimal sample complexity for Gaussian state tomography, since no meaningful lower bounds are currently known for this task. }%

In this work, we make significant progress towards filling this gap. We study the following questions:
\begin{center}
    {\em
        (1) What is the sample complexity for learning an $n$-mode Gaussian state with energy no more than $E$\\ to trace distance $\varepsilon$ with probability at least $2/3$?\\
        (2) What quantum resources and algorithmic capacity are needed to achieve sample-optimal learning?
    }
\end{center}
Many bosonic applications involve $n$-mode Gaussian states rather than single-mode. For example, to describe broadband searches for signals at unknown frequencies (e.g., in axion dark-matter searches~\cite{backes2021quantum, salemi2021search}), a large number of frequency-bin bosonic modes will be needed. Understanding how the complexity of learning scales with the number of modes $n$ is crucial for determining the ultimate efficiency of probing such elusive signals. 
The metric of trace distance is standard in the literature of quantum state tomography, as it precisely captures the maximum success probability of distinguishing two quantum states using any quantum measurements~\cite{helstrom1969quantum}. The success probability of $2/3$ is arbitrary and can be amplified straightforwardly.

We answer these questions under various different assumptions on the unknown Gaussian states and the measurement strategies. Our results can be summarized as follows (see also Table.~\ref{tab:bounds_summary}):
\begin{enumerate}
    \item We show a lower bound of $\Omega(n^3/\varepsilon^2)$ for classical (including Gaussian) measurements and ${\Omega}(n^2/\varepsilon^2)$ for arbitrary measurements. (See Sec.~\ref{sec:main_pre} for definitions of classical and Gaussian measurements.) The former establishes the protocols in~\cite{bittel2025energy} as nearly-optimal among all Gaussian measurement schemes\sloppy.
    We also show a tight lower bound of $\Omega(\bar E^2n^3/\varepsilon^2)$ for heterodyne measurements, matching the upper bound obtained in~\cite{bittel2025energy}. Here $\bar E$ is an upper bound on the per-mode energy of the unknown state. 
    \item With the additional assumptions that the unknown Gaussian state is either pure or passive, we show a nearly-tight bound of $\widetilde{\Theta}(n^2/\varepsilon^2)$.\footnote{Throughout this paper, tildes mean we hide log factors in $n$, $\varepsilon$, and polyloglog factors in $\bar E$. Similarly when we use the terms \emph{nearly-optimal} and \emph{nearly-tight}.} While for pure Gaussian states this can be achieved by Gaussian measurements, for passive Gaussian states this can only be achieved by non-classical (and thus non-Gaussian) measurements, as we show a tight bound of $\Theta(n^3/\varepsilon^2)$ for classical measurements in this task, manifesting a \emph{non-Gaussian advantage} in Gaussian state tomography. 
    \item For learning one-mode Gaussian states using non-entangling Gaussian measurements, we show a nearly tight bound of $\widetilde\Theta(E/\varepsilon^2)$ for non-adaptive schemes (which means the measurement setting cannot be chosen based on previous rounds of outcomes). This contrasts the doubly-logarithmic energy dependence from~\cite{bittel2025energy}, showing that adaptivity is necessary for nearly energy-independent sample complexity.
    \item In establishing these results, a crucial ingredient is to understand the connection and difference between learning a Gaussian state to $\varepsilon$ trace distance and learning its Wigner distribution to $\varepsilon$ total variational (TV) distance. The latter closely resembles the classical task of learning Gaussian distributions from samples.  
    As a byproduct of our main results, we obtain a nearly-tight sample complexity bound of $\widetilde{\Theta}(n^2/\varepsilon^2)$ for learning the Wigner distributions of any $n$-mode Gaussian state to $\varepsilon$ TV distance, which can be achieved by Gaussian measurements (see Theorem~\ref{th:main_wigner}).
    We also obtain a sharp bound between Gaussian states' trace distance and their Wigner TV distance (see Lemma~\ref{le:main_wigner_relation}): For any two $n$-mode Gaussian states $\rho_1,\rho_2$ with respective Wigner distributions $W_1,W_2$, it holds that $\Omega(\mr{TV}(W_1,W_2))\le D_\mr{tr}(\rho_1,\rho_2)\le \sqrt n\cdot O(\mr{TV}(W_1,W_2))$ whenever $D_\mr{tr}(\rho_1,\rho_2)$ is smaller than an absolute constant. 
    Moreover, there exist instances for any $n$ such that $D_\mr{tr}(\rho_1,\rho_2) \ge\frac14\sqrt n\,\mr{TV}(W_1,W_2)$. 
    We also show this $\sqrt n$ separation vanishes for pure Gaussian states.
    We believe these results are of independent interest.
\end{enumerate}

\paragraph{Notes added.} A few months after the first version of this manuscript was posted to arXiv, two preprints appear that further improved upon our results: First, \cite{chen2026optimal_merged} obtains that $\Theta(n^2/\varepsilon^2)$ samples are necessary and sufficient to learn any bosonic Gaussian state, closing a gap from our work. 
Combined with out results, this shows a \emph{non-Gaussian advantage} for learning general bosonic Gaussian states.
Second, \cite{chen2026log} shows that (among other results) there exists a \emph{non-Gaussian} non-adaptive single-copy measurement that learns any zero-mean one-mode Gaussian state using $\Theta(1/\varepsilon^2)$ samples. Combined with our Theorem~\ref{th:main_nonada}, this shows another form of non-Gaussian advantage.

\begin{table*}[!tp]\label{tab:summary}
    \centering
    \renewcommand{\arraystretch}{1.25}
    \begin{tabulary}{\textwidth}{@{} L C C @{}}
        \toprule
        Class of states & Class of measurements & Bounds \\
        \midrule

        \multirow{3}{*}{All Gaussian}
        & Classical (including Gaussian)
          & $\widetilde{O}(n^3/\varepsilon^2)$\,\cite{bittel2025energy},~ ${\Omega}(n^3/\varepsilon^2)$\,Thm.~\ref{th:main_lo_gaussian} \\
          \cmidrule(lr){2-3}
        & Any 
          & $\widetilde{O}(n^3/\varepsilon^2)$\,\cite{bittel2025energy},~ ${\Omega}(n^2/\varepsilon^2)$\,Thm.~\ref{th:main_lo_any} \\
          \cmidrule(lr){2-3}
        & Heterodyne 
          & ${O}(\bar E^2n^3/\varepsilon^2)$\,\cite{bittel2025energy},~ ${\Omega}(\bar E^2n^3/\varepsilon^2)$\,Thm.~\ref{th:main_lo_heterodyne} \\
          
        \midrule
        Pure Gaussian
          & Any or Gaussian
          & $\widetilde{\Theta}(n^2/\varepsilon^2)$~Thm.~\ref{th:main_lo_any},\,\ref{th:main_up_pure}\\

        \midrule
        \multirow{2}{*}{Passive Gaussian}
          & Any 
          & $\widetilde{\Theta}(n^2/\varepsilon^2)$~Thm.~\ref{th:main_up_passive}\\
          \cmidrule(lr){2-3}
          & Classical (including Gaussian)
          & {${\Theta}(n^3/\varepsilon^2)$~Thm.~\ref{th:main_lo_gaussian},\,\ref{th:main_up_passive}}\\

        \midrule
        1-mode Gaussian
          & Non-adaptive 1-copy Gaussian 
          & $\widetilde{\Theta}(\bar{E}/\varepsilon^2)$~Thm.~\ref{th:main_nonada}\\

        \bottomrule
    \end{tabulary}

    \caption{ Summary of results. The task is learning an $n$-mode Gaussian states $\rho(\mu,\Sigma)$ to $\varepsilon$ trace distance closeness with probability at least $2/3$, given the promise that $\|\Sigma\|_{\mr{op}}\le \bar E$. 
    Here $\widetilde O$, $\widetilde \Omega$, $\widetilde \Theta$, suppress $\log n$, $\log \varepsilon^{-1}$, and $\mr{poly}\log\log {\bar E}$ factors.}
    \label{tab:bounds_summary}
\end{table*}

\section{Preliminaries}\label{sec:main_pre}

Here, we introduce some background of bosonic quantum information. More details can be found in App.~\ref{sec:pre}. 
An $n$-mode bosonic system has Hilbert space $\mc H_{\mathrm{op}}^n\cong L^2(\mbb R^n)$, with $2n$ quadrature operators $\hat R\coleq (\hat x_1,\cdots,\hat x_n,\hat p_1,\cdots,\hat p_n)$ that satisfy the canonical commutation relation $[\hat R_k,\hat R_l]=i\Omega_{k,l}$ where $\Omega=\begin{pmatrix}
    0_n & I_n\\-I_n& 0_n
\end{pmatrix}$. The mean $\mu$ and covariance $\Sigma$ of a quantum state $\rho$ is given by 
\bb
    \mu_i\coleq \Tr(\hat R_i\rho),\quad \Sigma_{k,l}\coleq\frac12\Tr\left(\{\hat R_k,\hat R_l\}\rho\right) - \mu_k\mu_l.
\ee
Define the displacement operator $\hat D(\xi)\coleq\exp(i\xi^T\Omega\hat R)$ for all $\xi\in\mbb R^{2n}$. 
For any operator $\hat O$, the \emph{Wigner function} of it is defined by $W_{\hat O}(r)\coleq\frac{1}{(2\pi)^{2n}}\int\!\mr{d}^{2n}\xi\,\exp(i\xi^T\Omega r)\Tr[\hat O\hat D(\xi)]$. 
For a density operator $\rho$, its Wigner function satisfies the normalization condition $\int\!\mr{d}^{2n}r\,W_\rho(r)=1$. Any observable expectation value can be computed using Wigner functions according to $\Tr(\rho\hat O)= (2\pi)^n\int\!\mr{d}^{2n}r\,W_\rho(r) W_{\hat O}(r)$.

A \emph{Gaussian state} $\rho$ is a quantum state whose Wigner function is a Gaussian distribution, i.e., $W_\rho(r)=\mc N(\mu,\Sigma)$. They are completely determined by their mean and covariance, usually denoted by $\rho(\mu,\Sigma)$. Given a $2n$ by $2n$ real positive definite matrix $\Sigma$, the necessary and sufficient conditions for it to be a covariance matrix of some Gaussian state is $\Sigma+ \frac12i\Omega\ge 0$. Furthermore, it corresponds to a pure Gaussian state if $\det\Sigma=4^{-n}$. 
A Gaussian state is called \emph{passive} if it can be obtained by applying passive linear optical operations on a thermal states (see App.~\ref{sec:pre} for more details). 
Familiar examples of Gaussian states include coherent states, squeezed states, and thermal states.

Any measurement on a quantum system can be mathematically described by a POVM, i.e., a set of positive operators that sum or integrate to identity. We refer to a POVM as a \emph{classical measurement} if every element of it has a non-negative Wigner function. 
Here “classical” refers to the existence of a classical phase-space description of the measurement statistics, not to the physical implementation of the measurement device.
The most well-known example of classical measurements are Gaussian measurements (or, general-dyne measurements), which is of the form $\{\frac1{(2\pi)^n}D(\xi)\rho(0,V)D(\xi)^\dagger\}_{\xi\in\mbb R^{2n}}$ where $\rho(0,V)$ is a Gaussian state (thus $V+\frac i2\Omega\ge0$) called the seed of this measurement. 
The outcome distribution obtained by applying a Gaussian measurement with seed covariance matrix $V$ on a Gaussian state $\rho(\mu,\Sigma)$ is given by $\mc N(\mu,\Sigma + V)$.
Two examples of Gaussian measurements that are commonly used in quantum optics experiments are heterodyne measurements ($V=\frac12I_{2n}$) and homodyne measurements ($V=K\frac12\mr{diag}(a,\cdots,a,a^{-1},\cdots,a^{-1})K^T$ where $a\rightarrow\infty$ and $K$ is an orthogonal symplectic matrix determining the measured quadratures).
As for non-classical measurements, one typical example is \emph{photon number counting}, i.e., projecting to the eigen basis of the number operator $\hat n_i\coleq\frac12(\hat x_i^2+\hat p_i^2-1)$ at each mode. 
We will soon see that these two classes of measurements have sharply different power for Gaussian state learning.

For two quantum states $\rho$ and $\sigma$, their trace distance is defined as $D_\mr{tr}(\rho,\sigma)=\frac12\tr|\rho-\sigma|$. For two classical probability distributions $p$ and $q$ over $\mc X$, their total variation distance is defined as $\mr{TV}(p,q)=\frac12\sum_{x\in\mc X}|p(x)-q(x)|$. Both distances operationally characterize the distinguishability between two quantum states and classical distributions, respectively.

\section{Results}\label{sec:results}

\subsection{Warm-up: Learning Gaussian Wigner distributions}\label{sec:main_wigner}

Since any Gaussian state is characterized by a Wigner function that is a Gaussian distribution (i.e., a \emph{Gaussian Wigner distribution}), it is natural to compare the state learning problem to the classical task of learning a $2n$-dimensional classical Gaussian distribution to $\varepsilon$ TV distance from samples, the sample complexity of which is known to be $\Theta(n^2/\varepsilon^2)$~\cite{ashtiani2020near,devroye2020minimax}.
Here, we ask the following question: How many copies of an $n$-mode Gaussian state are needed to learn its Wigner distribution to $\varepsilon$ TV distance with high probability? 
The answer is not a priori clear from the classical results, as accessing quantum states can be very different than accessing classical samples.
Nevertheless, we show that a nearly-tight sample complexity bound for this task is given by $\widetilde{\Theta}(n^2/\varepsilon^2)$, as formalized by the following theorem.
\begin{theorem}[Nearly-optimal learning of Gaussian Wigner distributions]\label{th:main_wigner}
    $N=\Omega(n^2/\varepsilon^2)$ copies of an $n$-mode Gaussian state $\rho(\mu,\Sigma)$ are necessary to learn its Wigner distribution to $\varepsilon$ TV distance with probability $2/3$ using any measurements. Furthermore, there exists a Gaussian measurement scheme that achieves this using $N=\widetilde{O}(n^2/\varepsilon^2)=O\!\left(n^2/\varepsilon^2 + (n +\log\log\log \bar E )\log\log \bar E\right)$ copies given the promise that $\|\Sigma\|_\mr{op}\le \bar E$.
\end{theorem}
\begin{proof}[Proof sketch]
The upper bound is essentially established using the nearly energy-independent adaptive learning algorithm proposed in~\cite{bittel2025energy}. At a high level, the algorithm starts with a heterodyne measurement and recursively adjusts the Gaussian measurement seed to be more aligned with the unknown state, based on the obtained measurement outcomes. After $\Theta(\log\log{\bar E})$ steps, each measuring $\Theta(n+\log\log\log\bar E)$ copies, a sufficiently good seed is found with high probability.
Using this seed, one can effectively sample from the Wigner distribution of $\rho$ with a negligible
amount of quantum noise, thus retrieve the $O(n^2/\varepsilon^2)$ sample complexity for learning classical Gaussians.
A more detailed explanation and careful analysis on the role of adaptivity in this algorithm will be presented in Sec.~\ref{sec:energy_adaptivity}.

As for the lower bound, we use a proof technique that will repeatedly appear in the rest of this paper, which is based on the Fano's method~\cite{yu1997assouad}. We first construct a large ensemble of Gaussian states of size $2^{\Omega(n^2)}$ such that the distance separating each pair of states is at least $\Omega(\varepsilon)$. For this proof, the distance refers to the TV distance between Wigner distributions, but we will make use of trace distance in other proofs. The ensemble construction, inspired by the classical literature of~\cite{ashtiani2020near}, is intuitively depicted in Fig.~\ref{fig:construction} and detailed in App.~\ref{sec:wigner_learning}: Basically, it is obtained by applying a large number of beam splitter networks to some fixed initial Gaussian product states, with each network labeling an output state of the ensemble.
Next, we define a communication task between two players, Alice and Bob: Alice selects a state $\rho$ from the ensemble uniformly at random and sends $N$ copies of $\rho$ to Bob. Bob then tries to guess Alice's choice of $\rho$ by measuring these $N$ copies using any measurement. If a learning algorithm that satisfies the assumption of Theorem~\ref{th:main_wigner} exists, Bob can use that to guess correctly with probability at least $2/3$. Fano's inequality then implies Alice and Bob share at least $\Omega(n^2)$ bits of mutual information. On the other hand, we develop a technique to upper bound the Holevo quantity of the state ensemble by $N\cdot O(\varepsilon^2)$, which puts an upper bound on the achievable mutual information according to the Holevo theorem~\cite{holevo1973bounds}. We then conclude that $N=\Omega(n^2/\varepsilon^2)$ copies of the state $\rho$ are necessary.
\end{proof}

\begin{figure}
    \centering
    \includegraphics[width=0.5\linewidth]{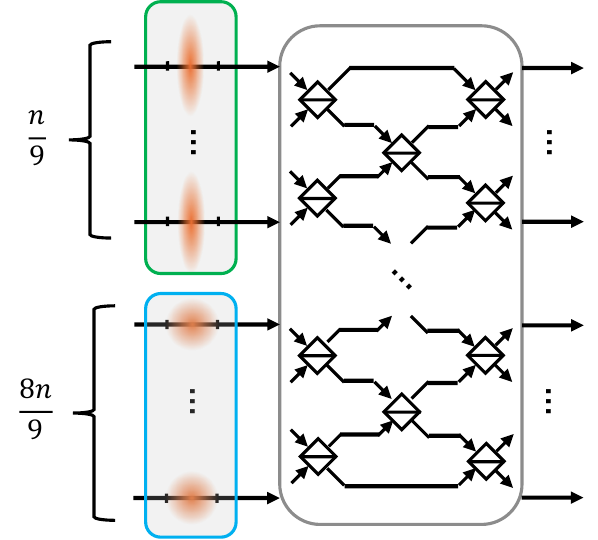}
    \caption{An illustration of the Gaussian state ensembles used in our lower bound proof. The left two boxes depict a total of $n$ initial single-mode Gaussian states. In this specific example, they consist of $n/9$ pure squeezed states and $8n/9$ vacuum states, which is used for the proof of Theorem~\ref{th:main_lo_any}. For the other lower bound proofs, we will make different choices of the initial states. The right box depicts a Gaussian unitary consisting solely of beam splitters. Each choice of the unitary defines an output Gaussian state of the ensemble, and we have $2^{\Omega(n^2)}$ different choices in total. This second part is the same for all ensembles that we use in different lower bound proofs.}
    \label{fig:construction}
\end{figure}

At first glance, Theorem~\ref{th:main_wigner} appears to address our problem. Indeed, learning Gaussian Wigner distributions from quantum samples has similar sample complexity to learning Gaussian distributions from classical samples, and Gaussian measurement achieves the nearly optimal performance. 
However, this is \emph{not} the same as learning the state up to trace distance.
We obtain the following result on the relation between trace distance and Wigner TV distance for Gaussian states. 

\begin{lemma}[Sharp bounds between trace distance and Wigner TV distance for Gaussian states]\label{le:main_wigner_relation}
    There exist absolute constants $c_0,c_1,c_2,c_3\ge0$ such that the following holds:
    For any $n\in\mbb N_+$, let $\rho_1$ and $\rho_2$ be any $n$-mode Gaussian states such that $D_\mr{tr}(\rho_1,\rho_2)\le c_0$. Then,
    \bb
        c_1\mr{TV}(W_{\rho_1},W_{\rho_2})\le D_\mr{tr}({\rho_1,\rho_2})\le c_2\sqrt n\,\mr{TV}(W_{\rho_1},W_{\rho_2}).
    \ee
    Furthermore, when both $\rho_1$ and $\rho_2$ are pure,
    \bb
        D_\mr{tr}({\rho_1,\rho_2})\le c_3\mr{TV}(W_{\rho_1},W_{\rho_2}).
    \ee
    All inequalities are sharp for any $n$ up to constants.
\end{lemma}
\noindent Although prior works have noticed that the trace distance and Wigner TV distance can be arbitrarily far apart~\cite{sabapathy2021bosonic, Lami_2021_dh}, to the best of our knowledge, Lemma~\ref{le:main_wigner_relation} is the first to sharply characterize their relationship in the context of Gaussian states. 

While the full proof of Lemma~\ref{le:main_wigner_relation} is postponed to App.~\ref{sec:wigner_relation}, let us highlight one example where \emph{the trace distance between two Gaussian states can be a factor of $O(\sqrt n)$ larger than their Wigner TV distance}:  Fix any $\varepsilon\in(0,1)$. Let $\Sigma_0=\frac12I_{2n}$ and $\Sigma_1=\frac12(1+\frac\varepsilon n)I_{2n}$ be the covariance matrices of two zero-mean Gaussian states $\rho_0$ and $\rho_1$, respectively. Physically, $\rho_0$ is the $n$-mode vacuum state while $\rho_1$ is a near-vacuum thermal state. The TV distance between their Wigner distributions can be upper bounded by $\mr{TV}(W_{\rho_0},W_{\rho_1})\le \frac1{\sqrt 2}\|\Sigma_0^{-1/2}\Sigma_1\Sigma_0^{-1/2}-I_{2n}\|_\mr{F} = \varepsilon n^{-1/2}$; In contrast, the trace distance satisfy $D_\mr{tr}(\rho_0,\rho_1)\ge \varepsilon/4$. This can been seen by applying a photon number counting measurement on both states. While $\rho_0$ always yields vacuum, the probability $\rho_1$ yields vacuum is $\braket{0_n|\rho_1|0_n}=(1+\varepsilon/2n)^{-n} \le 1 - \varepsilon/4$. Details for these calculations become clear in App.~\ref{sec:pre}. In conclusion, {$D_\mr{tr}(\rho_0,\rho_1)\ge\frac14\sqrt n\cdot\mr{TV}(W_{\rho_0},W_{\rho_1})$}
as in our claim.

Consequently, learning a Gaussian Wigner function to $\varepsilon$ TV distance is generally weaker than learning a Gaussian state to $\varepsilon$ trace distance, which is the problem we wish to address. For the latter, what we can now conclude by combining Theorem~\ref{th:main_wigner} and Lemma~\ref{le:main_wigner_relation} is a sample complexity lower bound of $\Omega(n^2/\varepsilon^2)$ for any measurements, and an upper bound of $\widetilde O(n^3/\varepsilon^2)$ using only \emph{Gaussian measurements}~\cite{bittel2025energy}. 
The remainder of this paper is devoted to developing a more fine-grained understanding of this task.

\subsection{Lower bounds on Gaussian state learning}

\noindent Our first result is a lower bound of $\Omega(n^3/\varepsilon^2)$ for any classical measurements:
\begin{theorem}[{Lower bound with classical measurements}]\label{th:main_lo_gaussian}
    $N=\Omega(n^3/\varepsilon^2)$ copies are necessary to learn an $n$-mode Gaussian states $\rho(\mu,\Sigma)$ to $\varepsilon$ trace distance closeness with probability $2/3$ using classical measurements, even with the promise that $\|\Sigma\|_\mr{op}=O(1)$ and $\rho$ is passive. A POVM is classical if all of its elements have non-negative Wigner functions.
\end{theorem}
\noindent Recall that Gaussian measurements are classical measurements.
Theorem~\ref{th:main_lo_gaussian} implies that \cite{bittel2025energy}'s adaptive learning algorithm with sample complexity $\widetilde O(n^3/\varepsilon^2)$ is \emph{nearly optimal} among all Gaussian measurement schemes, up to doubly-logarithmic energy dependence which is not physically significant. 
Theorem~\ref{th:main_lo_gaussian} also implies that, if one uses only classical measurements, learning an $n$-mode Gaussian state to $\varepsilon$ trace distance is \emph{strictly harder} than learning its Wigner distributions to $\varepsilon$ TV distance, as well as learning a classical $2n$-dimensional Gaussian to $\varepsilon$ TV distance from classical samples. 
We also remark that the lower bound holds even for learning states without entanglement, as all passive Gaussian states are non-entangled (see App.~\ref{sec:pre}).

\begin{proof}[Proof Sketch for Theorem~\ref{th:main_lo_gaussian}]

One important fact we use is the following: any classical measurement $\{E\}$ on $N$ copies of an unknown Gaussian state $\rho$ can be exactly simulated by $N$ classical samples from its Wigner distribution $W_\rho$~\cite{mari2012positive}. Indeed, since both the unknown state and the POVM elements have non-negative Wigner functions, the measurement outcome distribution can be expressed as $\Pr(E|\rho^{\otimes N})=(2\pi)^{nN}\int \mr{d}^{2nN}r\,W_{E}(r)W_{\rho^{\otimes N}}(r) = \int \mr{d}^{2nN}r\,\Pr(E|r)\prod_{i=1}^N W_{\rho}(r_i)$\sloppy. Thus, with $N$ classical samples $\{r\coleq r_1,\cdots,r_N\}$ from $W_\rho$, one just has to sample from the conditional distribution $P(E|r)$ to exactly simulate the measurement statistics.

Based on the above fact, the proof is conceptually similar to that of Theorem~\ref{th:main_wigner}: We first construct an ensemble of Gaussian states of size $2^{\Omega(n^2)}$ with pairwise trace distance at least $\Omega(\varepsilon)$. 
Here, the ensemble is constructed by applying passive Gaussian unitaries on a collection of single-mode vacuum states and thermal states, thus all states in the ensemble are passive, see Fig.~\ref{fig:construction}.
Then, consider a \emph{classical communication} (instead of quantum communication in Theorem~\ref{th:main_wigner}) task between Alice and Bob: Alice selects a state $\rho$ from the ensemble uniformly at random, and then sends $N$ classical samples from its Wigner distribution $W_\rho$ to Bob. Bob then tries to guess Alice's choice by processing the samples. If a classical learning scheme exists, Bob can guess correctly with probability at least $2/3$ (using the simulation argument from the previous paragraph~\cite{mari2012positive}). Fano's inequality implies the mutual information between Alice and Bob is at least $\Omega(n^2)$. With some calculation, one can upper bound the mutual information by $N\cdot O(\varepsilon^2/n)$. We then conclude $N =\Omega(n^3/\varepsilon^2)$. Full details can be found in App.~\ref{sec:gaussian-lower}.
\end{proof}

\medskip
\noindent Going beyond classical measurements, the learner can do something not simulatable using Wigner samples. We have the following lower bound for arbitrary measurements:
\begin{theorem}[{Lower bound with arbitrary measurements}]\label{th:main_lo_any}
$N=\Omega\left({n^2}/{\varepsilon^2}\right)$ copies are necessary to learn an $n$-mode Gaussian states $\rho(\mu,\Sigma)$ to $\varepsilon$ trace distance closeness with probability $2/3$ using any POVM measurements, even with the promise that $\|\Sigma\|_\mr{op}=O(1)$ and $\rho$ has zero mean. Furthermore, even if $\rho$ is promised to be pure, $N=\widetilde{\Omega}(n^2/\varepsilon^2)=\Omega(\frac{n^2}{\varepsilon^2\log(n/\varepsilon)})$ copies are still necessary. 
\end{theorem}

\begin{proof}[Proof Sketch for Theorem~\ref{th:main_lo_any}]
The $\Omega(n^2/\varepsilon^2)$ lower bound is a direct corollary of Theorem~\ref{th:main_wigner} and Lemma~\ref{le:main_wigner_relation}.
For the pure state lower bound, we plug in a different Gaussian state ensemble to the proof framework of 
Theorem~\ref{th:main_wigner}: The ensemble is constructed by applying a beam splitter network on a collection of single-mode vacuum states and pure squeezed states, thus all states in the ensemble are pure, see Fig.~\ref{fig:construction}. We then consider the same quantum communication task as in the proof of Theorem~\ref{th:main_wigner} and apply the Holevo theorem. For the pure ensemble, we use a different method to obtain an upper bound of $N\cdot O(\varepsilon^2\log(n/\varepsilon))$ on the Holevo quantity. This thus yields an sample complexity lower bound of $\Omega\left(\frac{n^2}{\varepsilon^2\log(n/\varepsilon)}\right)$.
Details can be found in App.~\ref{sec:any-lower}. 
\end{proof}

\medskip
\noindent Finally, using a similar ensemble but with highly-squeezed states, we also obtain the following lower bound for heterodyne measurements. 
In particular, it matches the heterodyne upper bound analyzed in \cite{bittel2025energy}, thus establishing a tight bound of $\Theta(\bar E^2 n^3/\varepsilon^2)$. 
The proof strategy is very similar and is postponed to App.~\ref{sec:heterodyne}.
\begin{theorem}[Lower bound with heterodyne measurements]\label{th:main_lo_heterodyne}
    $N=\Omega\left(\bar E^2 n^3/\varepsilon^2\right)$ copies are necessary 
    to learn an $n$-mode Gaussian state $\rho(\mu,\Sigma)$ with the promise that $\|\Sigma\|_\mr{op}\le\bar E$ to $\varepsilon$ trace distance closeness with probability $2/3$ using heterodyne measurements. 
\end{theorem}

\subsection{Upper bounds on Gaussian state learning}

At this point, the major open problem is whether there exists a learning algorithm that achieves the $\Omega(n^2/\varepsilon^{2})$ lower bound from Theorem~\ref{th:main_lo_any}. 
In the following, we provide new upper bounds on learning two subclasses of Gaussian states, and provide a reduction result on achieving the optimal sample complexity for learning all Gaussian states.

\medskip
\noindent We have the following upper bound for learning pure Gaussian states:
\begin{theorem}[{Nearly-optimal learning of pure Gaussian states}]\label{th:main_up_pure}
    A number of
    \bb\label{eq:upp_puregauss} 
    N=\widetilde{O}(n^2/\varepsilon^2)=O\!\left(n^2/\varepsilon^2 + (n +\log\log\log \bar E )\log\log \bar E\right)
    \ee
    copies are sufficient to learn an $n$-mode Gaussian state $\rho(\mu,\Sigma)$ to $\varepsilon$ trace distance closeness with probability $2/3$ given the promise that $\rho$ is pure and $\|\Sigma\|_\mr{op}\le\bar E$. Moreover, this can be achieved by Gaussian measurements. 
\end{theorem}
\noindent Recall that Theorem~\ref{th:main_lo_any} gives a lower bound of $\widetilde\Omega(n^2/\varepsilon^2)$  for learning pure Gaussian states. 
{Hence, $\widetilde{\Theta}(n^2/\varepsilon^2)$ copies are both necessary and sufficient for learning pure Gaussian states. Interestingly, this optimal scaling can be achieved using only Gaussian measurements. This means non-Gaussian resources provide essentially no advantages for pure Gaussian state tomography.

\begin{proof}[Proof Sketch for Theorem~\ref{th:main_up_pure}]
    The key technical ingredient is from Lemma~\ref{le:main_wigner_relation}: For any two pure Gaussian states, their trace distance is upper bounded by their Wigner TV distance up to a constant.
    The proof for Theorem~\ref{th:main_up_pure} has three steps:
    First use the adaptive learning algorithm in~\cite{bittel2025energy} to learn the Wigner distribution of $\rho$ to $\varepsilon$ TV distance using $\widetilde{O}(n^2/\varepsilon^2)$ samples, then project the estimator to a pure Gaussian state with carefully controlled error, and finally use Lemma~\ref{le:main_wigner_relation} to convert from Wigner TV distance to trace distance. The full details are presented in App.~\ref{sec:pure-upper}.
\end{proof}

\medskip
\noindent The second subclass of Gaussian states that allows for improved efficiency of tomography are passive Gaussian states. Recall that those are Gaussian states that can be obtained by applying passive linear optical operations on thermal states. We have that:
\begin{theorem}[{Non-Gaussian advantage in learning passive Gaussian states}]\label{th:main_up_passive}
        A number of 
        \bb
        N=\widetilde{O}(n^2/\varepsilon^2)=O\!\left(n^2/\varepsilon^2 + (n +\log\log\log \bar E )\log\log \bar E\right)
        \ee
        copies are sufficient to learn an $n$-mode Gaussian states $\rho$ to $\varepsilon$ trace distance closeness with probability $2/3$ given the promise that $\rho$ is passive and $\|\Sigma\|_\mr{op}\le\bar E$, which nearly matches the lower bound of $N=\Omega(n^2/\varepsilon^2)$. However, classical measurements {(including Gaussian measurements)} need $\Theta(n^3/\varepsilon^2)$ samples to complete the same task. 
\end{theorem}
Crucially, the above theorem yields a rigorous \emph{non-Gaussian advantage} (or, a \emph{non-classical advantage}) for learning passive Gaussian states. To our knowledge, this is the first provable and rigorous separation demonstrating a non-Gaussian advantage in quantum learning theory with continuous-variable systems~\cite{mele2025learning,bittel2025energy}\sloppy. 
Similar notions of non-Gaussian advantages have appeared in quantum metrology research (e.g.,~\cite{gardner2025stochastic,tsang2016quantum,wang2025limitations}) but are quite different from our setting.

}%

\begin{proof}[Proof Sketch for Theorem~\ref{th:main_up_passive}]
    The lower bound \(\Omega(n^3/\varepsilon^2)\) for learning passive Gaussian states via classical measurements follows from Theorem~\ref{th:main_lo_gaussian}, while the fact that heterodyne detection attains the same scaling can be shown by using the heterodyne upper bound obtained in~\cite{bittel2025energy}. 
    The lower bound \(\Omega(n^2/\varepsilon^2)\) for learning passive Gaussian states via general measurements is very similar to Theorem~1 and is proven in Theorem~\ref{th:non-gaussian-passive-lower}.

    The key to achieving the $\widetilde{O}(n^2/\varepsilon^2)$ upper bound is a recently developed toolkit known as the passive random purification channel~\cite{mele_purif_2025,walter_purif_2025}, denoted by $\Lambda^{(n,N)}$, which is an extension of similar objects in discrete-variable quantum systems~\cite{tang2025,pelecanos2025,random_pur_simple} to the bosonic regime. $\Lambda^{(n,N)}$ is a quantum channel that maps $N$ copies of any $n$-mode passive Gaussian state, $\rho^{\otimes N}$, to $N$ copies of its purification $\Psi^{\otimes N}_O$, although $\Psi_O$ is random (but fixed among the $N$ copies) and is labeled by an orthogonal symplectic operator $O$. Moreover, each $\Psi_O$ is guaranteed to be a $2n$-mode Gaussian state whose covariance matrix satisfies $\|\Sigma_U\|_\mr{op}\le4\|\Sigma\|_\mr{op}$. In math,
    \bb
    \Lambda^{(n,N)}\!\left(\rho_A^{\otimes N}\right)
    =\mathbb{E}_{O}\!\left[\ketbra{\Psi_O}{\Psi_O}^{\otimes N}\right]
    =
    \mathbb{E}_{O}\!\left[\left(
    \left(I_A\otimes U_O\right)
    \ketbra{\Psi}{\Psi}
    \left(I_A\otimes U_O\right)^\dagger
    \right)^{\otimes N}\right],
    \ee
    where $\Psi$ is a canonical purification of $\rho$ and each $\ket{\Psi_O}\equiv I_A\otimes U_O\ket{\Psi}$ is also a purification of $\rho$. The expectation is with respect to the Haar measure over the symplectic orthogonal group. Note that, $\Lambda^{(n,N)}$ can be an \emph{entangling} and \emph{non-Gaussian} channel in general, but its action on $N$ copies of a passive Gaussian state is guaranteed to yield a mixture of pure product Gaussian states. 
    We also note that, despite its clear mathematical existence~\cite{mele_purif_2025}, finding an efficient circuit to implement $\Lambda^{(n,N)}$ is an open problem beyond the scope of the current work.
    
    Thanks to this, our learning algorithm first applies $\Lambda^{(n,N)}$ to $N$ copies of the unknown passive Gaussian state, then applies the nearly-optimal pure Gaussian state tomography algorithm in Theorem~\ref{th:main_up_pure}, and finally traces out the purifying ancillary modes. 
    Since the partial trace never increase trace distance, this clearly yields an upper bound of $\widetilde{O}(n^2/\varepsilon^2)$.
    Importantly, $\Lambda^{(n,N)}$ is not a Gaussian channel, thus the protocol is not a Gaussian measurement scheme and not limited by Theorem~\ref{th:main_lo_gaussian}.
    More details are presented in App.~\ref{sec:passive-upper}.
\end{proof}

One open problem in this direction is how to construct a random purification channel for general (non-passive) Gaussian states. The main technical challenge there dealing with the non-compactness of the symplectic group, which is necessary to describe squeezing. Though solving this problem is left for future works, we provide the following reduction result, which follows from the same reasoning as in the proof of Theorem~\ref{th:main_up_passive}. 
\begin{proposition}
    If there exists a random purification channel for all Gaussian states, such that all the purified states are Gaussian states whose covariance matrices $\Sigma_U$ satisfy $\|\Sigma_U\|_\mr{op} \le \mr{poly}(\|\Sigma\|_\mr{op})$, then $N=\widetilde O(n^2/\varepsilon^2)$ samples are sufficient to learn any $n$-mode Gaussian states to $\varepsilon$ trace distance closeness with $2/3$ probability.
\end{proposition}

\subsection{On Energy Dependence and Roles of Adaptivity}\label{sec:energy_adaptivity}

Finally, we seek to address the following question: \emph{What resources or capacities enable the nearly energy-independent sample complexity?}  Focusing on the experimentally feasible class of schemes using non-entangling (i.e., single-copy) Gaussian measurements, we find that \emph{adaptivity} is the key ingredient for energy independence, as can be elucidated even in learning a single-mode Gaussian state:
\begin{theorem}[Adaptivity is necessary for energy-independent tomography via Gaussian schemes]\label{th:main_nonada}
    Focus on one-mode systems and non-adaptive, non-entangling, Gaussian measurements. Then $N=\Omega(\bar E/\varepsilon^2)$ copies are necessary to learn a Gaussian state $\rho(\mu,\Sigma)$ to $\varepsilon$ trace distance closeness with probability $2/3$ given the promise that $\|\Sigma\|_\mr{op}\le\bar E$. Furthermore, there exists a scheme that achieves this using $N=O(\bar E\log\bar E/\varepsilon^2)$ samples.
\end{theorem}
That is, $\widetilde\Theta(\bar E\varepsilon^{-2})$ copies are necessary and sufficient for non-adaptive single-copy Gaussian measurements in learning one-mode Gaussian states.
Note that the adaptive learning algorithm in~\cite{bittel2025energy} also uses single-copy Gaussian measurements and only needs $O(\varepsilon^{-2} +\log\log\bar E \log\log\log\bar E)$ samples; it works by recursively and adaptively refining the measurement seed to be more aligned with the unknown state, thus effectively ``unsqueezing'' it. Theorem~\ref{th:main_nonada} shows that adaptivity is necessary to achieve this very weak energy dependence. Also recall that using only heterodyne measurements requires $\Theta(\bar E^2\varepsilon^{-2})$ copies (see Theorem~\ref{th:main_lo_heterodyne}), which has quadratically worse energy dependence than Theorem~\ref{th:main_nonada}. Indeed, our nearly-optimal non-adaptive protocol uses phase-space angle-randomized homodyne measurements. 
To our knowledge, this is the first rigorous and nearly-tight sample-complexity analysis of randomized homodyne measurements in the context of Gaussian state tomography, despite their widespread experimental use~\cite{breitenbach1997measurement,lvovsky2009continuous} and earlier theoretical investigation~\cite{ohliger2012compressedsensing}. 
In Fig.~\ref{fig:one-mode}, we illustrate an example comparing how heterodyne, random-angle homodyne, and the adaptive general-dyne~\cite{bittel2025energy} protocols work in learning an unknown single-mode squeezed state.

\begin{figure}[!htp]
    \centering
    \includegraphics[width=0.8\linewidth]{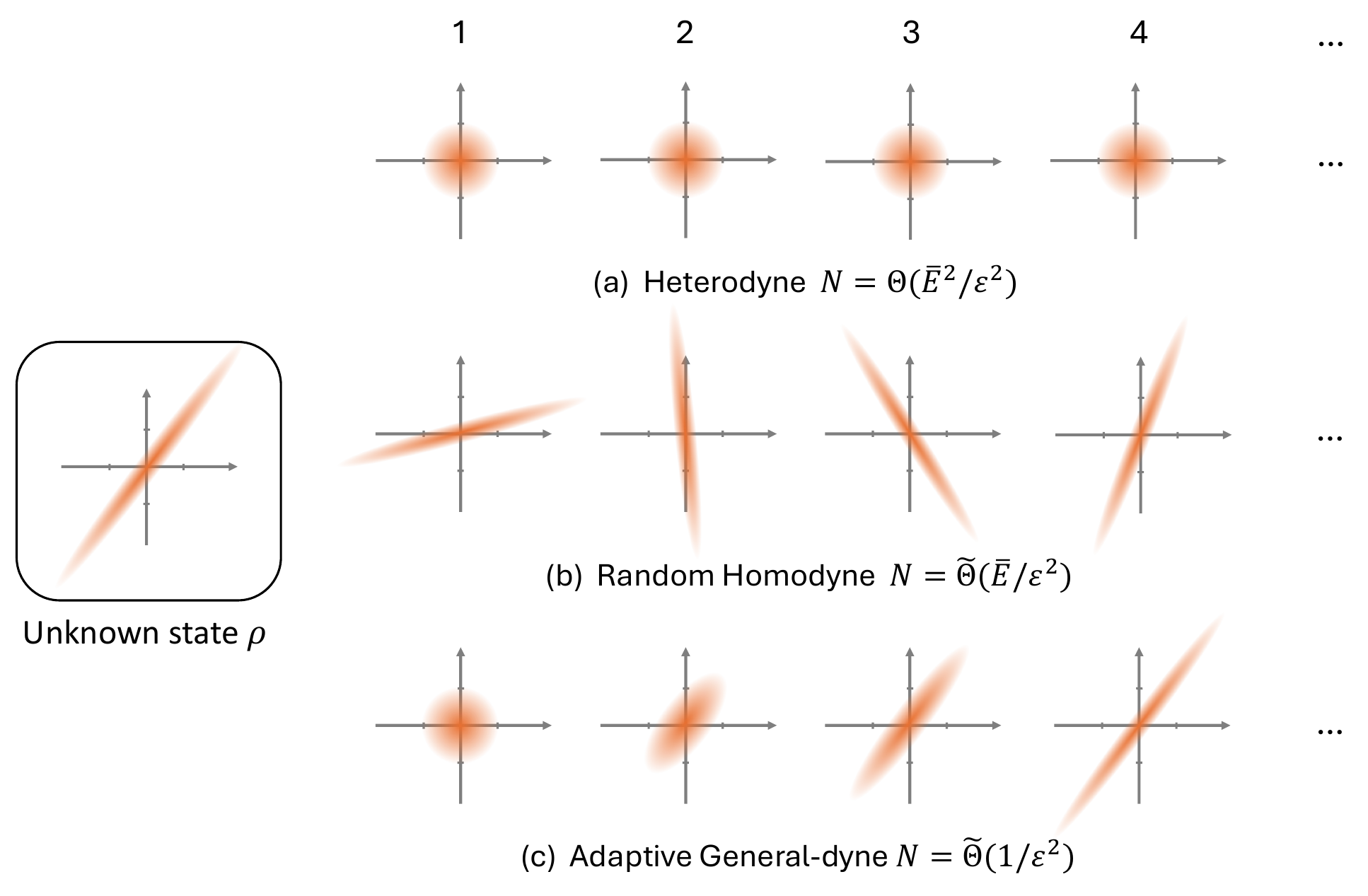}
    \caption{Comparison of three different single-copy Gaussian measurement schemes. The left box shows the Wigner function of an unknown single-mode Gaussian state $\rho$. On the right show the Wigner functions of the Gaussian measurement seeds for the first few copies of $\rho$ of three different protocols. (a) Heterodyne measurements, with sample complexity $\Theta(\bar E^2/\varepsilon^2)$, whose seeds are always the vacuum state. (b) Angle-randomized homodyne, which is a nearly optimal non-adaptive scheme proposed in Theorem~\ref{th:main_nonada} with sample complexity $\widetilde{\Theta}(\bar E/\varepsilon^2)$, whose seeds are angle-randomized squeezed states (which approximate the ideal homodyne). (c) Adaptive general-dyne measurements as proposed in~\cite{bittel2025energy}, which achieves nearly energy-independent sample complexity of $\widetilde{\Theta}(1/\varepsilon^2)$. The seeds are recursively and adaptively refined to be aligned with $\rho$ so as to maximize sensitivity.}
    \label{fig:one-mode}
\end{figure}

\begin{proof}[Proof Sketch for Theorem~\ref{th:main_nonada}]
For the lower bound, we apply a proof technique inspired by~\cite{chen2024adaptivity} in discrete-variable quantum systems. We construct a continuous family of binary hypothesis testing tasks parameterized by $\phi\in[0,\pi)$: distinguishing two single-mode zero-mean Gaussian states $\rho_{0,\phi}$ and $\rho_{1,\phi}$ with covariance matrices $\Sigma_{0,\phi}=\frac12R_\phi\mr{diag}(a;a^{-1})R_\phi^T$ and $\Sigma_{1,\phi}=\frac12R_\phi\mr{diag}(a;(1+2\varepsilon)a^{-1})R_\phi^T$ respectively. Here $R_\phi$ is the standard 2-by-2 rotation matrix and $a = 2\bar E$. Physically, both $\rho_{0,\phi}$ and $\rho_{1,\phi}$ are squeezed along the phase space angle $\phi+\frac\pi2$, while the latter is slightly less squeezed than the former, to the extent that their trace distance is $\Omega(\varepsilon)$. Intuitively, for a Gaussian measurement to properly distinguish these two states, the seed covariance matrix $V$ needs to be squeezed along the same angle to prevent quantum noises from destroying the signal. However, when $\phi$ is randomized over $[0,\pi)$, no fixed $V$ can distinguish all pairs of $\rho_{0,\phi}$ and $\rho_{1,\phi}$ with high probability, due to the constraints $V+\frac12i\Omega\ge0$. While an adaptive scheme can quickly search for a pretty good $V$ as in~\cite{bittel2025energy}, non-adaptive schemes must pay a factor of $\Theta(\bar E)$ of copies to distinguish with high overall success probability. More details about the lower bound are presented in App.~\ref{sec:non-ada-lower}.

On the other hand, we present an explicit non-adaptive protocol to achieve the upper bound in App.~\ref{sec:non-ada-upper}. 
At a high level, the protocol conducts homodyne measurements\footnote{Here, the homodyne measurement does not need to be ideal -- it can be realized by a general-dyne with finitely-squeezed $V$. This is explicitly analyzed in App.~\ref{sec:non-ada-upper}.} along a few random phase space angles on $\Theta(\bar E\log \bar E/\varepsilon^2)$ copies and heterodyne measurements on $\Theta(1/\varepsilon^2)$ copies. 
Based on the homodyne outcomes, the algorithm first roughly estimate the condition number of $\Sigma$, i.e., the ratio between the largest and smallest eigenvalues. Denote the estimator by $\hat\kappa$. If $\hat\kappa$ is smaller than a constant, one simply uses the heterodyne estimator from~\cite{bittel2025energy} which suffices to solve the learning tasks using $\Theta(1/\varepsilon^2)$ copies. Otherwise, one uses the randomized homodyne data to first roughly estimate the highly squeezed direction of $\Sigma$, and then solves for $\Sigma$ and $\mu$ using homodyne outcomes along a few carefully post-selected directions. The overhead of $\bar E$ in sample complexity is to ensure such ``good'' phase space angles exist among the random angles sampled with high probability. 

The use of a randomized homodyne direction is mainly for the simplicity of analysis. In practice, we expect a uniformly spaced choice of homodyne angles would work as well. We also believe the small number of additional heterodyne measurements is not necessary, given a more careful analysis, and that this algorithm can be extended to learn multi-mode Gaussian states.
These will be left for future work.   
\end{proof}

\section{Discussion}\label{sec:discussion}

\subsection{Prior works and context of our results} 
{
Historically, the problem of quantum state tomography first emerged in the early 1990s in the setting of continuous-variable quantum systems~\cite{smithey_measurement_1993,lvovsky2009continuous}, and only later became a central object of study for finite-dimensional systems. In the continuous-variable regime, experimental tomography protocols have traditionally relied on homodyne and heterodyne detection~\cite{dariano_tomographic_1995,lvovsky2009continuous,wallentowitz_reconstruction_1995,smithey_measurement_1993,babichev_homodyne_2004} to reconstruct phase-space descriptions of the state~\cite{serafini2023quantum}, including the Wigner, characteristic, and Husimi functions~\cite{lvovsky2009continuous,leonhardt_measuring_1995,smithey_measurement_1993,serafini2023quantum}. While these techniques have long been routine in quantum optics, they historically did not provide guarantees for the reconstruction error in trace distance---arguably the most operationally meaningful notion of distance between quantum states.

Only recently has continuous-variable quantum state tomography with provable trace-distance guarantees begun to attract systematic theoretical attention~\cite{mele2025learning,Bittel_2025,fanizza2025,bittel2025energy,mele2025symplecticrank,Zhao_2025,iosue2025}. Specifically, the problem of learning Gaussian states in trace distance has recently been investigated in Refs.~\cite{mele2025learning,Bittel_2025,fanizza2025,bittel2025energy}, primarily with the goal of establishing upper bounds on the sample complexity. Prior to the present work, however, essentially no lower bounds on the sample complexity of Gaussian state tomography were known, leaving its fundamental limits largely unresolved.

}

The problem of Gaussian state learning was first introduced in~\cite{mele2025learning}, where simple tomography protocols based on estimating the first and second moments of the state via homodyne detection were shown to achieve sample complexity scaling polynomially with both the number of modes and the energy. These protocols were subsequently refined in~\cite{fanizza2025,Bittel_2025}, leading to progressively tighter upper bounds on the sample complexity. All these algorithms rely on estimating the first and second moments of the Gaussian state via heterodyne detection, with the best known upper bound for heterodyne-based tomography being $O(\bar{E}^2 n^3/\varepsilon^2)$~\cite{bittel2025energy}. In our work, we show that this scaling is in fact optimal for heterodyne-based tomography, proving that its sample complexity is 
$\Theta(\bar{E}^2 n^3/\varepsilon^2)$
(see Theorem~\ref{th:main_lo_heterodyne}). Moreover, we show that such a polynomial dependence on the energy is unavoidable for any single-copy, {non-adaptive} Gaussian tomography algorithm (see Theorem~\ref{th:main_nonada}), and is therefore not specific to heterodyne- or homodyne-based schemes. This observation naturally motivates the study of single-copy Gaussian algorithms that additionally allow for adaptivity.

Indeed, in~\cite{bittel2025energy}, a single-copy adaptive Gaussian tomography algorithm was proposed that achieves an almost energy-independent sample complexity, with only a doubly logarithmic dependence on the energy. The key idea~\cite{bittel2025energy} is to first adaptively ``unsqueeze'' the state and then estimate its first and second moments. This yields a tomography algorithm that relies exclusively on Gaussian operations and has sample complexity $\widetilde O(n^3/\varepsilon^2)$, which currently represents the best known upper bound on the sample complexity of learning Gaussian states. Our work complements this result by proving that \emph{any} tomography protocol that relies exclusively on Gaussian operations must require at least $\Omega(n^3/\varepsilon^2)$ samples (see Theorem~\ref{th:main_lo_gaussian}). In particular, we show that the protocol of~\cite{bittel2025energy} is sample-optimal within the class of tomography algorithms that employ only Gaussian operations, thereby establishing that the sample complexity of learning Gaussian states using Gaussian operations is $\widetilde\Theta(n^3/\varepsilon^2)$.  Finally, dropping the Gaussianity assumption on the operations, we show that any arbitrary (non-Gaussian) tomography algorithm must require at least $\Omega(n^2/\varepsilon^2)$ samples. Altogether, this implies that the sample complexity of learning Gaussian states is lower bounded by $\Omega(n^2/\varepsilon^2)$ and upper bounded by $\widetilde O(n^3/\varepsilon^2)$ (see Theorem~\ref{th:main_lo_any}).

A central technical ingredient underlying the sample-complexity upper bounds discussed above~\cite{mele2025learning,Bittel_2025,fanizza2025,bittel2025energy} is the use of estimates that relate the trace distance between two Gaussian states to suitable distances between their first and second moments~\cite{mele2025learning,Bittel_2025,fanizza2025,holevo2024estimates,mele2025achievablerates}. Such estimates make it possible to convert guarantees on moment estimation into corresponding guarantees on the trace-distance error of the reconstructed Gaussian state. Early bounds of this type were not tight~\cite{mele2025learning}, but were subsequently refined in~\cite{holevo2024estimates,fanizza2025}, ultimately culminating in a tight characterization of the trace distance between mixed Gaussian states in~\cite{Bittel_2025,bittel2025energy} (see also~\cite{mele2025achievablerates,martinezcifuentes2026calculating} for an algorithm that computes the trace distance between two Gaussian states up to a prescribed precision). Related trace-distance bounds in the non-Gaussian setting were derived in~\cite{mele2025symplecticrank}. {In this work, we further contribute to this line of research by establishing sharp bounds between the trace distance and Wigner TV distance of Gaussian states (Lemma~\ref{le:main_wigner_relation}): For general Gaussian states $\rho,\sigma$, it holds that $\Omega(\mr{TV}(W_{\rho},W_{\sigma}))\le D_{\mr{tr}}(\rho,\sigma)\le O(\sqrt n\cdot\mr{TV}(W_{\rho},W_{\sigma}))$; for pure Gaussian states, $D_{\mr{tr}}(\psi,\phi)= \Theta(\mr{TV}(W_{\psi},W_{\phi}))$.
These bounds are key ingredients in our main results: they yield the optimal sample complexity $\widetilde\Theta(n^2/\varepsilon^2)$ for learning pure Gaussian states (Theorem~\ref{th:main_up_pure}) and the lower bound $\Omega(n^2/\varepsilon^2)$ for learning general Gaussian states (Theorem~\ref{th:non-gaussian-lower-stronger}). 
We expect these results to have much more implications, such as for Gaussian data hiding~\cite{sabapathy2021bosonic, Lami_2021_dh,wang2025gaussian,turner2026optimal}.

 }%

Beyond tomography, our work also contributes to the literature on provable quantum advantages in learning tasks with continuous-variable systems~\cite{oh2024entanglement,Liu_2025,coroi2025}. In particular, Refs.~\cite{oh2024entanglement,Liu_2025,coroi2025} establish provable quantum advantages enabled by entanglement with an ancillary quantum memory, showing that learning without access to such a resource is provably less efficient. Here, we establish a different type of quantum advantage enabled by non-Gaussian resources for the task of learning passive Gaussian states (see Theorem~\ref{th:main_up_passive}): while Gaussian protocols require $\Omega(n^3/\varepsilon^2)$ samples, allowing non-Gaussian operations enables learning with only $\widetilde O(n^2/\varepsilon^2)$ samples. This yields a strict polynomial separation between Gaussian and non-Gaussian algorithms for learning passive Gaussian states.

More broadly, the present work fits into the rapidly growing literature on quantum learning theory with continuous-variable systems~\cite{mele2025learning,Bittel_2025,fanizza2025,bittel2025energy,becker_cs,gandhari2023,oh2024entanglement,Liu_2025,coroi2025,fawzi2024optimalfidelity,Wu_2024,mobus2023dissipation,Zhao_2025,mele2025symplecticrank,iosue2025,fanizza2025_cv_unitary,girardi2025gaussian,markovich2025,rosati2024learning}. This line of research includes studies on tomography of Gaussian states~\cite{mele2025learning,Bittel_2025,fanizza2025,bittel2025energy}, tomography of structured classes of non-Gaussian states~\cite{mele2025learning,mele2025symplecticrank,Zhao_2025,iosue2025,markovich2025}, tomography of Gaussian unitary channels~\cite{fanizza2025_cv_unitary}, bosonic variants of classical shadows~\cite{becker_cs,gandhari2023}, property testing~\cite{girardi2025gaussian}, fidelity estimation~\cite{fawzi2024optimalfidelity}, provable quantum advantages in learning tasks~\cite{oh2024entanglement,Liu_2025,coroi2025}, and Hamiltonian learning~\cite{fanizza2025,mobus2023dissipation}\sloppy.

Besides quantum learning theory, quantum sensing and metrology in bosonic systems and with Gaussian states have also been well investigated~\cite{giovannetti2011advances,nichols2018multiparameter}, with applications in quantum imaging~\cite{boto2000quantum,tsang2016quantum,nair2016far,tham2017beating}, gravitational wave detection~\cite{aasi2015advanced,tse2019quantum}, etc. 
Quantum metrology mainly focuses on the asymptotic regime where the estimation error is much smaller than all the other parameters. In that regime, quantum Fisher information and quantum Cram\'er-Rao bounds are powerful tools to characterize the ultimate efficiency. In contrast, our results works in the non-asymptotic regime where many features cannot be revealed from a metrological analysis, such as the role of adaptivity in energy-dependent sample complexity scaling. It is an interesting direction to better understand the connections and differences between quantum learning theory and quantum metrology~\cite{chen2026instance}.

\subsection{Open problems}

\paragraph{Sample-optimal learning of general Gaussian states.}
A central open problem arising from our work is to determine the optimal sample complexity for learning mixed Gaussian states. We conjecture that the correct scaling is $\widetilde{\Theta}(n^2/\varepsilon^2)$, thereby saturating the general lower bound established in Theorem~\ref{th:main_lo_any} and coinciding with the sample complexity for learning pure Gaussian states established in Theorem~\ref{th:main_up_pure}. A promising route toward achieving this result would be establishing the existence of a random purification channel~\cite{tang2025,pelecanos2025,random_pur_simple} for general Gaussian states, extending recent developments for passive Gaussian states~\cite{mele_purif_2025,walter_purif_2025}. If such a channel can be realized, an optimal learning strategy would proceed by first applying this channel to the copies of the unknown mixed Gaussian state, thereby producing copies of a Gaussian purification, and then performing the sample-optimal tomography for pure Gaussian states introduced in Theorem~\ref{th:main_up_pure} on the resulting Gaussian purifications. This approach would yield a non-Gaussian tomography algorithm with sample complexity $\widetilde{O}(n^2/\varepsilon^2)$, matching the lower bound of Theorem~\ref{th:main_lo_any} up to logarithmic factors.

A related open question concerns the role of energy constraints. The best-known tomography protocols for both mixed Gaussian states~\cite{bittel2025energy} and pure Gaussian states (see Theorem~\ref{th:main_up_pure}) exhibit a mild double-logarithmic dependence on the energy. It remains unclear whether this mild dependence reflects a fundamental limitation or is merely an artifact of existing techniques.

\paragraph{Adaptivity for more general measurements.} We have shown that non-adaptive non-entangled (i.e., single-copy) Gaussian measurements suffer from a linear-in-energy sample complexity lower bound. An open question is whether such a bound holds against entangled (i.e., multi-copy) measurement. In particular, \cite{bittel2025energy} shows non-adaptive continuous-variable Bell measurements on a Gaussian state $\rho$ and its conjugate $\rho^T$ can achieve energy-independent sample complexity. When conjugate states are not available, it remains unclear whether entangled Gaussian measurements need adaptivity to achieve energy independence. It also remains unclear whether non-Gaussian measurements need adaptivity. 

\paragraph{Learning Gaussian processes.}
Another important direction for future work concerns tomography of Gaussian processes. Can our techniques be extended to establish lower and upper bounds on the query complexity of learning Gaussian channels? At present, upper bounds are known only in the restricted setting of Gaussian unitary channels~\cite{fanizza2025_cv_unitary}, while the general case of non-unitary Gaussian channels remains open and lower bounds are entirely unexplored. {Physically, this problem is even more closely tied to gravitational wave and dark matter searches, as these signals can usually be modeled as a Gaussian random displacement channel~\cite{gardner2025stochastic}. It is interesting to investigate whether non-Gaussian input states and coherent feedback provide an advantage for learning such channels.}

\paragraph{Experimental demonstrations.}
A big motivation for Gaussian state tomography comes from experiments. Thus, it is interesting to implement our sample-optimal learning protocols in real experiments. While some of our protocols (e.g., Theorem~\ref{th:main_nonada}) requires only Gaussian measurements that are experimentally feasible, others (e.g., Theorem~\ref{th:main_up_passive}) require highly complicated operations such as a random purification channel. For the latter, can we find an experimentally friendly non-Gaussian protocol that is also sample-optimal?

\subsection{Potential practical applications}
Beyond theoretical interest, we expect the formalism of sample-optimal Gaussian state learning to have practical impact. Here, we discuss two directions for potential practical applications. The first is in broadband weak signal searches. Examples include stochastic gravitational wave detection~\cite{allen1999detecting} and axionic dark matter searches~\cite{backes2021quantum, malnou2019squeezed}. In those settings, the signal detector is often modeled by a multi-mode bosonic Hilbert space, where each mode corresponds to a frequency bin. If we assume the detector is initialized to a Gaussian state (as is usually done in experiments), the output signal-carrying state will be a multi-mode Gaussian state, and the signal detection task becomes a multi-mode Gaussian state learning task. For detecting a broadband weak signal, which corresponds to a large number of modes $n$ and a small error tolerance $\varepsilon$, achieving optimal sample complexity with respect to both $n$ and $\varepsilon$ is thus highly advantageous, as it will maximize the chance of successful detection with fixed experimental resources. 
{In practice, the signal-carrying multi-mode Gaussian state is unlikely to be completely arbitrary; rather, it typically possesses additional structure. For example, frequency modes may exhibit finite-range correlations or decay of correlations beyond a characteristic bandwidth}.
An interesting open question is how to model the structure and design efficient learning algorithms that exploit such structure. 
{Another example from gravitational-wave detection is the search for the waveforms' deviations from general relativity; such deviations are often parametrized as a linear combination of basis functions, and one would like to bound the distribution of coefficients in this combination~\cite{Isi2019hierarchical,Seymour2026gravitational}. }

Another area of potential applications is quantum device characterization and benchmarking. In recent years, bosonic systems have become strong candidates for quantum computation. A series of works aims at demonstrating quantum computational advantage in Gaussian boson sampling~\cite{hamilton2017gaussian,zhong2020quantum,madsen2022quantum}, which applies photon number counting on randomly generated Gaussian states to create distributions that are classically hard to simulate. 
A sample-optimal Gaussian state learning algorithm can be used to efficiently characterize the states generated in this setup; this might be helpful for certifying the advantage, or might identify structure that makes classical simulation easier~(for example, by leveraging photon loss in the protocol~\cite{oh2024classical}).

\phantomsection
\section*{Acknowledgments} 
\addcontentsline{toc}{section}{Acknowledgments}
We thank Rana Adhikari, Lennart Bittel, Su Direkci, Jens Eisert, Ludovico Lami, Shruti Maliakal, Antonio Anna Mele, Lee McCuller, and Sisi Zhou for helpful discussions.
We thank Zhihan Zhang for pointing out a minor gap in the proof of Theorem~\ref{th:non-ada-upper} in the first version of this manuscript and for helping us address it.
S.C., A.L., Z.M., H.H., Y.C., J.P. acknowledge funding provided by the Institute for Quantum Information and Matter, an NSF Physics Frontiers Center (NSF Grant PHY-2317110). F.A.M.\ acknowledges financial support from the European Union (ERC StG ETQO, Grant Agreement no.\ 101165230). F.A.M.\ thanks California Institute of Technology for hospitality.
H.H. and J.P. acknowledge support the U.S. Department of Energy, Office of Science, National Quantum Information Science Research Centers, Quantum Systems Accelerator. H.H. acknowledges support from the Broadcom Innovation Fund.
{Y.C.\ is supported by the Simons Foundation (Award Number 568762), and by NSF Grant PHY-2309231.}

\paragraph{Statement of AI usage.}
The main results of this paper are obtained by human from late 2025 to early 2026. During this time, ChatGPT 5.4 was used to check correctness and contribute minor ideas. Its only significant technical contribution is in Theorem~\ref{th:non-gaussian-lower} where it suggests a way to bound the von Neumann entropy as described in Eq.~\eqref{eq:bound_the_entropy}.
The whole manuscript is human-written with minor assistance from ChatGPT 5.5 and 5.6 Sol for checking typos and errors.
The authors take full responsibility for this work.

\newpage
\appendix
\renewcommand{\thelemma}{\Alph{section}.\arabic{lemma}}
\renewcommand{\theHlemma}{\Alph{section}.\arabic{lemma}}
\numberwithin{lemma}{section}
\setcounter{theorem}{0}
\setcounter{lemma}{0}
\setcounter{figure}{0}

\renewcommand{\thetheorem}{S\arabic{theorem}}
\renewcommand{\theHtheorem}{S\arabic{theorem}}

\renewcommand{\theequation}{S\arabic{equation}}
\renewcommand{\theHequation}{S\arabic{equation}}

\renewcommand{\thefigure}{S\arabic{figure}}
\renewcommand{\theHfigure}{S\arabic{figure}}

\section{Notations and Preliminaries}\label{sec:pre}

Denote the quadrature operators as $\hat R\coleq[\hat x_1,\cdots,\hat x_n,\hat p_1,\cdots,\hat p_n]$, satisfying
\bb
[\hat R_j, \hat R_k] = i\Omega_{jk}.
\ee
Here $\Omega \coleq \begin{pmatrix}
    0_n & I_n\\-I_n & 0_n
\end{pmatrix}.$
Define the annihilation operator as
\bb
\hat a \coleq \frac{\hat x + i\hat p}{\sqrt2}
\ee
Recall that a Gaussian state is a state whose Wigner function is a Gaussian distribution.
The mean $\mu$ and covariance $\Sigma$ of a Gaussian state $\rho(\mu,\Sigma)$ is defined as
\bb
\mu_i\coleq\Tr(\rho \hat R_i),\quad \Sigma_{ij} \coleq \frac12\Tr(\rho\{\hat R_i-\mu_i,\hat R_j - \mu_j\}).
\ee
In particular, for Vacuum states $\Sigma_{\text{vacuum}}=\frac12 I$. Note that our choice of convention is different from from~\cite{bittel2025energy}.

\medskip

A sufficient and necessary condition for a positive definite matrix $\Sigma$ to be a valid Gaussian covariance matrix is $\Sigma+\frac12i\Omega\ge 0$. Given a valid $\Sigma$, it represents a pure Gaussian state if and only if $\det\Sigma = 4^{-n}$. 
Equivalently, A positive definite matrix $\Sigma$ is a valid Gaussian covariance matrix if and only if $\Sigma^{-1} \le 4\Omega^T \Sigma \Omega$, and the equality holds if and only if it represents a pure Gaussian state.

Any real positive-definite matrix $V$ has the following Williamson decomposition: $V=SDS^T$. Here $S$ is a symplectic matrix, meaning that $S\Omega S^T =\Omega$; $D=\mr{diag}(\nu_1,\nu_2,\cdots,\nu_n,\nu_1,\nu_2,\cdots,\nu_n)$ where $\nu_i>0$ are the symplectic eigenvalues of $V$. $V$ is a valid Gaussian covariance matrix if and only if $\nu_i\ge1/2$ for all $i$. For pure Gaussian covariance matrices, all symplectic eigenvalues are $1/2$ and $V=\frac12SS^T$.

Any symplectic matrix $S$ admits an Euler decomposition $S=K_1ZK_2$ where $K_1,K_2$ are orthogonal symplectic matrix and $Z=(e^{r_1},\cdots,e^{r_n},e^{-r_1},\cdots,e^{-r_n})$. Combined with the Williamson decomposition, one can see that (1) the (standard) eigenvalues of any pure Gaussian covariance matrix take the form $\{\frac12\lambda_1,\cdots,\frac12\lambda_n,\frac12\lambda_1^{-1},\cdots,\frac12\lambda_n^{-1}\}$; and (2) the eigenvalues of a general Gaussian covariance matrix can be paired in a way that the product of each pair is at least $1/4$. %
{One can prove claim (2) as follows. By Williamson's decomposition, for any covariance matrix $V$ there exist a symplectic matrix $S$ and a diagonal matrix $D\ge \frac{I}{2}$ such that $V(\rho)=SDS^\intercal$, and since $D\ge \frac{I}{2}$ this directly gives $V\ge \frac12 SS^\intercal$. Moreover, by the Euler decomposition, the eigenvalues of $SS^\intercal$ come in reciprocal pairs, i.e.~, $\{z_1,1/z_1,\ldots,z_n,1/z_n\}$ with $z_1,\ldots,z_n\ge1$. Finally, combining $V\ge \frac12 SS^\intercal$ with Weyl's monotonicity theorem~\cite{BHATIA-MATRIX} gives the desired conclusion.}

{A Gaussian state is called \emph{passive} if it has zero mean and its covariance matrix can be written in Williamson form as \(V = K D K^T\), where \(K\) is symplectic orthogonal. Equivalently, a passive Gaussian state is a state obtained by applying a passive Gaussian unitary, i.e.~a Gaussian unitary associated with a symplectic orthogonal matrix \(K\), to a thermal state. Note that any passive Gaussian state is unentangled. More precisely, every \(n\)-mode passive Gaussian state is fully separable across the \(n\) modes, i.e.~it can be written as a (Gaussian) convex combination of tensor products of $n$ single-mode (Gaussian) states. This follows from two facts: first, every single-mode thermal state admits a Gaussian convex decomposition into single-mode coherent states~\cite{serafini2023quantum}; second, a passive Gaussian unitary maps any tensor product of single-mode coherent states into another tensor product of single-mode coherent states~\cite{serafini2023quantum}. Therefore, a passive Gaussian state, being obtained by applying a passive Gaussian unitary to a tensor product of single-mode thermal states, is itself a convex combination of tensor products of single-mode coherent states, and hence is fully separable.}

\medskip

Given two probability density functions $p,~q$ over the same measurable space $\mc X$ the total variation distance is $\mr{TV}(p,q)\coleq\int\mr{d}x\frac12\|p(x)-q(x)\|$; the KL divergence is $\mr{KL}(p\|q)\coleq\int\mr{d}x\,p(x)\log\frac{p(x)}{q(x)}$; and the $\chi^2$ divergence is $\chi^2(p\|q)\coleq\int\mr{d}x\, \frac{\left(p(x)-q(x)\right)^2}{q(x)}$. For the latter two to be well-defined, it requires that $q(x)>0$ whenever $p(x)>0$. For two density operators $\rho,~\sigma$ on the $n$-mode bosonic Hilbert space, the trace distance is $D_{\mr{tr}}(\rho,\sigma)\coleq\frac12\Tr|\rho-\sigma|$ where $|\cdot|$ denotes the operator absolute value. All operator norm in this work should be understood as Schatten p-norm, i.e., $\|A\|_{p}\coleq \Tr(|A|^p)^{\frac1p}$. Additionally, $\|A\|_\mr{op}\coleq\|A\|_\infty$ and $\|A\|_\mr{F}\coleq\|A\|_2$.

$\log$ is base $e$ throughout the paper.

\subsection{Known results for bosonic Gaussian states}

\begin{lemma}\cite{spedalieri2012limit}
For two $n$-mode Gaussian states $\rho_1=\rho(\mu_1,\Sigma_1)$ and $\rho_2=\rho(\mu_2,\Sigma_2)$, when at least one of them is pure, their fidelity is given by,
\bb
  F(\rho_1,\rho_2) \coleq \Tr(\rho_1\rho_2) =  \frac{1}{\sqrt{\det(\Sigma_1+\Sigma_2)}}\exp\left(-\frac12(\mu_1-\mu_2)^T(\Sigma_1+\Sigma_2)^{-1}(\mu_1-\mu_2)\right).
\ee
\end{lemma}
\noindent Throughout this paper, we use the squared convention of fidelity. 
Note that \cite{spedalieri2012limit} uses a different vacuum noise convention than ours.
The fidelity between two mixed $n$-mode Gaussian states is known from~\cite{banchi2015quantum}, but we will not use it here due to its complicated expression.

\begin{lemma}\cite[Theorem~8]{bittel2025energy}\label{le:perturbation}
    For any two Gaussian states $\rho(\mu_1,\Sigma_1)$ and $\rho(\mu_2,\Sigma_2)$, it holds that 
    \bb 
        D_\mr{tr}(\rho(\mu_1,\Sigma_1),\rho(\mu_2,\Sigma_2)) \le \frac{1}{\sqrt{2}}\|\Sigma_1^{-1/2}(\mu_2-\mu_1)\|_2 + \cfrac{1 + \sqrt3}{8}\Tr\left((\Sigma_1^{-1} + \Sigma_2^{-1})\left|\Sigma_1-\Sigma_2\right|\right).
    \ee
\end{lemma}
\noindent Note that the additional factor of $1/\sqrt2$ in the first term is due to our different convention of covariance matrices from~\cite{bittel2025energy}.
Related but different upper bounds on Gaussian state trace distance are also given in~\cite{holevo2024estimates,fanizza2025}.

\subsection{Known results for classical Gaussians}

Here, we review some known results about classical Gaussian distributions. We refer the readers to~\cite{canonne_gaussian_testing} for a more detailed review about learning and testing classical Gaussians.

\medskip
\noindent The following result from \cite{arbas2023polynomial} gives a tight characterization of the TV distance between two Gaussian distributions. A weaker version of this result that applies to same-mean Gaussians appears in~\cite{devroye2018total}.

\begin{lemma}\cite[Theorem~1.8]{arbas2023polynomial}\label{le:mahalanobis}
  Let $\mu_1,\mu_2\in\mbb R^d$ and $\Sigma_1,\Sigma_2$ be $d\times d$ positive definite matrices. Suppose ${\mathrm{TV}}(\mc N(\mu_1,\Sigma_1), \mc N(\mu_2,\Sigma_2))\le 1/600$. Define
  \bb
    \Delta \coleq \max\left\{\|\Sigma_2^{-1/2}\Sigma_1\Sigma_2^{-1/2}-I\|_F,\|\Sigma_2^{-1/2}(\mu_2-\mu_1)\|_2\right\}.
  \ee
  Then
  \bb
    \frac1{200}\Delta \le {\mathrm{TV}}(\mc N(\mu_1,\Sigma_1), \mc N(\mu_2,\Sigma_2)) \le \frac1{\sqrt2}\Delta.
  \ee
  Note that swapping the labels 1 and 2 also yields a valid bound.
\end{lemma}

The following result characterizes the optimal sample complexity for learning a Gaussian distribution to small total variation distance in the realizable setting. This was explicitly proved in \cite{ashtiani2020near} (specifically, the lower bound is in \cite[Theorem 6.3]{ashtiani2020near}), while the upper bound is in \cite[Theorem C.1]{ashtiani2020near}); see also \cite[Theorem 1.16]{diakonikolas2019} for a concise statement of the result.

\begin{lemma}\cite{ashtiani2020near}
  There exists an algorithm taking $O((d^2+d\log\delta^{-1})/\varepsilon^2)$ i.i.d. samples $v_1,\cdots,v_{2m}$ from an unknown $d$-dimensional Gaussian distribution $\mc N(\mu,\Sigma)$ and run in $O((d^4+d^3\log\delta^{-1})/\varepsilon^2)$ time, that outputs an estimate $\hat\mu,\hat\Sigma$ such $\mr{TV}(\mc N(\mu,\Sigma), \mc N(\hat\mu,\hat\Sigma))\le \varepsilon/2$ with probability at least $1-\delta$. In particular, the estimators are the empirical mean and covariance,
  \bb
    \hat\mu = \frac1m \sum_{i=1}^m v_i, \quad \hat\Sigma = \frac1{2m} \sum_{i=1}^m (v_{2i} - v_{2i-1})(v_{2i} - v_{2i-1})^T.
  \ee
  Moreover, the sample complexity lower bound for this task is $N=\Omega(\frac{d^2}{\varepsilon^2})$ (fixing success probability to be at least $2/3$), matching the upper bound\,\footnote{The lower bound given in~\cite{ashtiani2020near} is $\Omega(\frac{d^2}{\varepsilon^2\log(1/\varepsilon)})$. A more careful use of Fano's inequality (e.g. as in the proof for Theorem~\ref{th:gaussian-lower} in the current paper) removes the log factor. See also~\cite{devroye2020minimax} for this tightened lower bound.}.
\end{lemma}

The following results are about different divergence measures between two Gaussian distributions and the connections to their mean values and covariance matrices.

\begin{lemma}\cite{ashtiani2020near, gil2013renyi}
  For any two full-rank Gaussians $\mc N(\mu_1,\Sigma_1)$ and $\mc N(\mu_2,\Sigma_2)$,
  \bb
    & 2\mathrm{TV}(\mc N(\mu_1,\Sigma_1), \mc N(\mu_2,\Sigma_2))^2
    \\\le~& \mathrm{KL}(\mc N(\mu_1,\Sigma_1)\|\mc N(\mu_2,\Sigma_2))
    \\=~&\frac12\left(\Tr(\Sigma_2^{-1}\Sigma_1 - I) - \log\det(\Sigma_2^{-1}\Sigma_1) + (\mu_1 - \mu_2)^T\Sigma_2^{-1}(\mu_1 - \mu_2)\right).
    \\\leqt{(i)}~& \chi^2(\mc N(\mu_1,\Sigma_1)\|\mc N(\mu_2,\Sigma_2)) 
    \\=~& \frac{\det(\Sigma_2)}{\sqrt{\det(\Sigma_1)\det(2\Sigma_2 - \Sigma_1)}}\exp\left((\mu_1 - \mu_2)^T(2\Sigma_2 - \Sigma_1)^{-1}(\mu_1 - \mu_2)\right) - 1.
  \ee
  For (i) to hold, it further requires $2\Sigma_2 - \Sigma_1> 0$.
\end{lemma}

{\subsection{Clarifications on energy constraint}\label{sec:energy_constraint}

Throughout the paper, we assume that the unknown state $\rho$ obeys the energy constraint
\bb
\|\Sigma\|_{\mathrm{op}} \le E,
\ee
where $\Sigma$ denotes the covariance matrix of $\rho$. A widely used alternative in the Gaussian-tomography literature is to assume a bound on the mean energy of $\rho$~\cite{mele2025learning,Bittel_2025,bittel2025energy}. We adopt the operator-norm formulation mainly for notational convenience, since our sample-complexity bounds depend explicitly on $\|\Sigma\|_{\mathrm{op}}$; all of our results can be translated directly to a mean-energy constraint.

The condition $\|\Sigma\|_{\mathrm{op}} \le E$ also has a simple operational meaning. As shown in Lemma~\ref{lem:meaning_op} below, $\|\Sigma\|_{\mathrm{op}}$ coincides with the largest variance that a fixed quadrature observable (e.g., the position quadrature of the first mode $\hat{x}_1$) can attain after applying an arbitrary passive Gaussian unitary to $\rho$. In other words, it controls the worst-case quadrature fluctuations that can be generated by mixing the modes through a passive interferometer. Moreover, for states with vanishing first moment, the same constraint upper bounds the maximum energy that can be concentrated onto any fixed output mode using passive linear optics.

\begin{lemma}\label{lem:meaning_op}
Let $\rho$ be an $n$-mode state, and let $\hat{x}_1$ denote the position operator of the first mode. Then,
\bb\label{eq:meaning_op}
\|\Sigma\|_{\mathrm{op}}
=
\sup_{G}\Bigl(\Tr[\hat{x}_1^2\, G\rho G^\dagger]-\Tr[\hat{x}_1\, G\rho G^\dagger]^2\Bigr),
\ee
where the supremum ranges over all passive Gaussian unitaries $G$.

Furthermore, if $\rho$ has vanishing first moment, then
\bb\label{eq:energy_concentration}
\sup_{G}\Tr[\hat{E}_1\, G\rho G^\dagger]\le \|\Sigma\|_{\mathrm{op}},
\ee
where $\hat{E}_1\coloneqq (\hat{x}_1^2+\hat{p}_1^2)/2$ denotes the energy operator of the first mode.
\end{lemma}

\begin{proof}
By the variational characterization of the operator norm, there exists $\bar{\bm{w}}\in\mathbb{R}^{2n}$ with
$\bar{\bm{w}}^\intercal \bar{\bm{w}}=1$ such that $\|\Sigma\|_{\mathrm{op}}=\bar{\bm{w}}^\intercal \Sigma\,\bar{\bm{w}}$. By using the definition of the covariance matrix, we obtain
\bb\label{eq:var_form}
\|\Sigma\|_{\mathrm{op}}
=
\Tr\!\left[\rho\,(\bar{\bm{w}}^\intercal \hat{\bm{R}})^2\right]
-
\Tr\!\left[\rho\,(\bar{\bm{w}}^\intercal \hat{\bm{R}})\right]^2.
\ee
Moreover, \cite[Lemma 13]{Lami_2020} guarantees the existence of a passive Gaussian unitary $\bar G$ such that
$\bar{\bm{w}}^\intercal \hat{\bm{R}}=\bar G^\dagger \hat{x}_1 \bar G$. Substituting into~\eqref{eq:var_form} gives $\|\Sigma\|_{\mathrm{op}}
=
\Tr[\hat{x}_1^2\,\bar G\rho \bar G^\dagger]
-
\Tr[\hat{x}_1\,\bar G\rho \bar G^\dagger]^2$, and therefore
\bb
\|\Sigma\|_{\mathrm{op}}
\le
\sup_{G}\Bigl(\Tr[\hat{x}_1^2\, G\rho G^\dagger]-\Tr[\hat{x}_1\, G\rho G^\dagger]^2\Bigr).
\ee
Conversely, for any passive Gaussian unitary $G$, there exists
$\bm{w}\in\mathbb{R}^{2n}$ with $\bm{w}^\intercal\bm{w}=1$ such that $G^\dagger \hat{x}_1 G=\bm{w}^\intercal \hat{\bm{R}}$.
Hence,
\bb
\Tr[\hat{x}_1^2\, G\rho G^\dagger]-\Tr[\hat{x}_1\, G\rho G^\dagger]^2
=
\bm{w}^\intercal \Sigma\,\bm{w}
\le
\|\Sigma\|_{\mathrm{op}}.
\ee
Taking the supremum over $G$ establishes~\eqref{eq:meaning_op}.

We now prove~\eqref{eq:energy_concentration}. If $\rho$ has vanishing first moment, then so does $G\rho G^\dagger$ for any passive Gaussian unitary $G$; and hence~\eqref{eq:meaning_op} implies that
\bb
\Tr[\hat{x}_1^2\, G\rho G^\dagger]\le \|\Sigma\|_{\mathrm{op}}
\qquad\text{and}\qquad
\Tr[\hat{p}_1^2\, G\rho G^\dagger]\le \|\Sigma\|_{\mathrm{op}}\,,
\ee
which directly yields~\eqref{eq:energy_concentration}.
\end{proof}
}

\section{On Gaussian Wigner distributions}

\subsection{Bounds between trace distance and Wigner TV distance for Gaussian states} \label{sec:wigner_relation}

In this section, we prove the following relation between Gaussian state trace distance and Gaussian Wigner TV distance:
\begin{theorem}\label{th:wigner_relation}
    There exist absolute constants $c_0,c_1,c_2,c_3\ge0$ such that the following holds:
    For any $n\in\mbb N_+$, let $\rho_1$ and $\rho_2$ be any $n$-mode Gaussian states such that $D_\mr{tr}(\rho_1,\rho_2)\le c_0$. Then,
    \bb
        c_1\mr{TV}(W_{\rho_1},W_{\rho_2})\le D_\mr{tr}({\rho_1,\rho_2})\le c_2\sqrt n\,\mr{TV}(W_{\rho_1},W_{\rho_2}).
    \ee
    Furthermore, when both $\rho_1$ and $\rho_2$ are pure,
    \bb
        D_\mr{tr}({\rho_1,\rho_2})\le c_3\mr{TV}(W_{\rho_1},W_{\rho_2}).
    \ee
    All inequalities are sharp up to constants.
\end{theorem}

\medskip
\noindent We decompose the above theorem into a few lemmas and prove them one by one.

\begin{lemma}\label{le:WignerTV-smaller-than-Dtr}
    For any two $n$-mode Gaussian states $\rho(\mu_1,\Sigma_1)$ and $\sigma(\mu_2,\Sigma_2)$, it holds that,
    \begin{equation}
        D_\mr{tr}(\rho,\sigma)\ge \frac{\sqrt2}{400}\,\mr{TV}(W_\rho,W_\sigma),
    \end{equation}
    given that $D_\mr{tr}(\rho,\sigma)\le\sqrt2/240000$.
\end{lemma}
\noindent All constants here are not fundamental and are improvable given a more careful analysis.

\begin{proof}
    Denote the covariance matrices of $\rho$ and $\sigma$ by $\Sigma_1$ and $\Sigma_2$, respectively.
    \bb
        &D_\mr{tr}(\rho,\sigma)\\
        \geqt{(i)}~&
        \sup_{V:\, V+\frac{i}{2}\Omega\ge0}\mr{TV}\left(\mc {N}(\mu_1,\Sigma_1 + V),\,\mc {N}(\mu_2,\Sigma_2+V)\right)\\
        \geqt{(ii)}~& \sup_{V:\, V+\frac{i}{2}\Omega\ge0}\frac1{200}\max\left\{\left\|(\Sigma_1+V)^{-\frac12}(\Sigma_2+V)(\Sigma_1+V)^{-\frac12}-I_{2n}\right\|_\mr{F},~\left\|(\Sigma_1+V)^{-\frac12}(\mu_2-\mu_1)\right\|_2\right\}\\
        \geqt{(iii)}~&\frac1{400}\max\left\{\left\|\Sigma_1^{-1/2}\Sigma_2\Sigma_1^{-1/2}-I_{2n}\right\|_\mr{F},~\left\|\Sigma_1^{-\frac12}(\mu_2-\mu_1)\right\|_2\right\}\\
        \geqt{(iv)}~& \frac{\sqrt2}{400}\mr{TV}( W_{\rho},W_{\sigma}).
    \ee    
    Here (i) uses that trace distance is lower bounded by TV distance of any measurement outcome distributions, in particular, any general-dyne measurement. (ii) uses the characterization for Gaussian TV distance given in Lemma~\ref{le:mahalanobis}. (iii) holds by choosing $V\coleq\Sigma_1$, which is a valid choice as $\Sigma_1$ is a valid Gaussian covariance matrix. (iv) uses the other inequality in Lemma~\ref{le:mahalanobis}.
    Note that, the assumption on $D_{\tr}(\rho,\sigma)$ is chosen to guarantee Lemma~\ref{le:mahalanobis} can be applied.
\end{proof}

\begin{lemma}\label{le:at-most-sqrtn}
    Let $\rho$ and $\sigma$ be $n$-mode Gaussian states with Wigner functions satisfying $\mr{TV}(W_\rho,W_\sigma)<\frac{1}{600}$. Then, the trace distance between the Gaussian states can be upper bounded in terms of the total variation distance between the Wigner functions as
    \bb
        D_\mr{tr}(\rho,\sigma)\le 200\left(\frac{1}{\sqrt{2}} + \frac{1+\sqrt 3}{4}\sqrt{2n}\right) \mr{TV}(W_\rho,W_\sigma)\,.
    \ee
\end{lemma}
\begin{proof}
    Let $\mu_1$ and $\Sigma_1$ denote the first and second moments of $\rho$, respectively, and let $\mu_2$ and $\Sigma_2$ denote the first and second moments of $\sigma$, respectively.  By exploiting the proof of~\cite[Theorem~8]{bittel2025energy}, we obtain that
    \bb\label{eq_bmem}
        &D_\mr{tr}(\rho,\sigma)\le \frac{1}{\sqrt{2}}\|\Sigma_1^{-1/2}(\mu_2-\mu_1)\|_2\\
        &\quad+ \frac{1+\sqrt 3}{4}\int_{0}^1\mathrm{d}\alpha\left\|\big[\Sigma_1+\alpha (\Sigma_2-\Sigma_1)\big]^{-1/2}(\Sigma_1-\Sigma_2)\big[\Sigma_1+\alpha (\Sigma_2-\Sigma_1)\big]^{-1/2}\right\|_1 \,.
    \ee
    We can bound the second term as
    \bb\label{eq_1_bound}
        &\int_{0}^1\mathrm{d}\alpha\left\|\big[\Sigma_1+\alpha (\Sigma_2-\Sigma_1)\big]^{-1/2}(\Sigma_1-\Sigma_2)\big[\Sigma_1+\alpha (\Sigma_2-\Sigma_1)\big]^{-1/2}\right\|_1 \\
        &\le \sqrt{2n} \int_{0}^1\mathrm{d}\alpha\left\|\big[\Sigma_1+\alpha (\Sigma_2-\Sigma_1)\big]^{-1/2}(\Sigma_1-\Sigma_2)\big[\Sigma_1+\alpha (\Sigma_2-\Sigma_1)\big]^{-1/2}\right\|_F\\
        &\le  \sqrt{2n\int_{0}^1\mathrm{d}\alpha\left\|\big[\Sigma_1+\alpha (\Sigma_2-\Sigma_1)\big]^{-1/2}(\Sigma_1-\Sigma_2)\big[\Sigma_1+\alpha (\Sigma_2-\Sigma_1)\big]^{-1/2}\right\|_F^2}\\
        &= \sqrt{2n\int_{0}^1\mathrm{d}\alpha\Tr\left[\big[\Sigma_1+\alpha (\Sigma_2-\Sigma_1)\big]^{-1}(\Sigma_1-\Sigma_2)[(\Sigma_1+\alpha (\Sigma_2-\Sigma_1)\big]^{-1}(\Sigma_1-\Sigma_2)\right]}\\
        &=\sqrt{2n\int_{0}^1\mathrm{d}\alpha\,\Tr\big[(I+\alpha B )^{-1}B(I+\alpha B)^{-1}B\big]}\,,
    \ee
    where we defined the symmetric matrix $B\coloneqq \Sigma_1^{-1/2}\Sigma_2\Sigma_1^{-1/2}-I$. 
    Denote the eigenvalues of $B$ as $\{\lambda_j\}_{j=1}^{2n}$. 
    Also define $\tilde B\coleq \Sigma_2^{-1/2}\Sigma_1\Sigma_2^{-1/2}-I$ whose eigenvalues are $\{-\lambda_j/(1+\lambda_j)\}_{j=1}^{2n}$, which can be seen by noticing that $\Sigma_2^{-1/2}\Sigma_1\Sigma_2^{-1/2}$ and $(\Sigma_1^{-1/2}\Sigma_2\Sigma_1^{-1/2})^{-1}$ are conjugate.
    We thus obtain that
    \bb\label{eq_2_bound}
        \int_{0}^1\mathrm{d}\alpha\,\Tr\big[(I+\alpha B )^{-1}B(I+\alpha B)^{-1}B\big]
        &=\sum_{j=1}^{2n} \lambda_j^2 \int_{0}^1\mathrm{d}\alpha \,\frac{1}{(1+\alpha\lambda_j)^2}\\
        &= \sum_{j=1}^{2n} \frac{\lambda_j^2}{1+\lambda_j}\\
        &\leqt{(i)} \| B\|_F \| \tilde B\|_F
    \ee
    where (i) uses Cauchy-Schwarz.
    Without loss of generality, we assume $\| \tilde B\|_F\le \| B\|_F$. Otherwise, we can just swap the labels of $\sigma$ and $\rho$. This will not make any difference as both the trace distance and TV distance are symmetric.
    Consequently, by substituting \eqref{eq_1_bound} and  \eqref{eq_2_bound} into  \eqref{eq_bmem}, we have that
    \bb
        D_\mr{tr}(\rho,\sigma)\le \frac{1}{\sqrt{2}}\|\Sigma_1^{-1/2}(\mu_2-\mu_1)\|_2 + \frac{1+\sqrt 3}{4}\sqrt{2n} \| \Sigma_1^{-1/2}\Sigma_2\Sigma_1^{-1/2}-I\|_F\,.
    \ee
    Finally, by Lemma~\ref{le:mahalanobis}, we conclude that 
    \bb
        D_\mr{tr}(\rho,\sigma)\le 200\left(\frac{1}{\sqrt{2}} + \frac{1+\sqrt 3}{4}\sqrt{2n}\right) \mr{TV}(W_\rho,W_\sigma)\,.
    \ee
\end{proof}

The following lemma provides a useful upper bound on the trace distance between two pure Gaussian states in terms of the TVD between the corresponding classical Gaussian distributions\footnote{
An related upper bound on the trace distance between two pure Gaussian states has been derived in~\cite{holevo2024estimates}. However, that upper bound does not seem to imply Lemma~\ref{le:pure-trace-tvd}.
}.

\begin{lemma}\label{le:pure-trace-tvd}
  For any two $n$-mode pure Gaussian states $\rho(\mu_1,\Sigma_1)$ and $\rho(\mu_2,\Sigma_2)$,
  \bb
    D_\mr{tr}(\rho(\mu_1,\Sigma_1),\rho(\mu_2,\Sigma_2)) < 150\cdot \mr{TV}(\mc N(\mu_1,\Sigma_1),\mc N(\mu_2,\Sigma_2)),
  \ee
  given that $\mr{TV}(\mc N(\mu_1,\Sigma_1),\mc N(\mu_2,\Sigma_2))<1/600$.
\end{lemma}
\begin{proof}
  Define $\delta\mu\coleq\mu_2-\mu_1$. Let $\mr{TV}(\mc N(\mu_1,\Sigma_1),\mc N(\mu_2,\Sigma_2)) = \varepsilon \le 1/600$. Lemma~\ref{le:mahalanobis} implies that
  \bb\label{eq:tvd-maha}
    \max\left\{\|\Sigma_2^{-1/2}\Sigma_1\Sigma_2^{-1/2}- I_{2n}\|_\mr{F},~\|\Sigma_1^{-1/2}\delta\mu\|_2,~  \|\Sigma_2^{-1/2}\delta\mu\|_2\right\} \le c\varepsilon,
  \ee
  where $c\coleq 200$.
  Denote $\rho_x \coleq \rho(\mu_x,\Sigma_x)$ for $x=1,2$. The trace distance between two pure states satisfies $D_\mr{tr}(\rho_1,\rho_2) = \sqrt{1 - F(\rho_1,\rho_2)}$, so we just need to lower bound the fidelity. Recall the fidelity formula for pure Gaussian states,
  \bb
    F(\rho_1,\rho_2) = \frac{1}{\sqrt{\det(\Sigma_1 + \Sigma_2)}} \exp\left(-\frac12 \delta\mu^T(\Sigma_1 + \Sigma_2)^{-1}\delta\mu\right).
  \ee
  { For the second factor, we have that
  \bb
  \exp\left(-\frac12 \delta\mu^T(\Sigma_1 + \Sigma_2)^{-1}\delta\mu\right) &\geqt{(i)}\exp\left(-\frac18\delta\mu^T(\Sigma_1^{-1} 
  + \Sigma_2^{-1})\delta\mu\right)\\
  &\geqt{(ii)}\exp\left(-\frac14c^2\varepsilon^2\right)\\
  &\geq 1 - \frac14 c^2\varepsilon^2\,,
  \ee
    where in (i) we used the operator convexity of $x\mapsto x^{-1}$, which establishes that for any two positive-definite matrices $A,B$, one has~\cite[Eq.~(1.33)]{BHATIA}
\begin{equation}
\left(\frac{A+B}{2}\right)^{-1}\le \frac{A^{-1}+B^{-1}}{2}\,;
\end{equation}
and in~(ii) we exploited Eq.~\eqref{eq:tvd-maha}.}

  To bound the first factor, note that $\frac12\Sigma_2^{-1/2}\Sigma_1\Sigma_2^{-1/2}$ is a valid pure Gaussian covariance matrix. Indeed, 
  \bb
    \frac12\Sigma_2^{-1/2}\Sigma_1\Sigma_2^{-1/2} + \frac12i\Omega ~\eqt{(i)}~ \frac12\Sigma_2^{-1/2}\left(\Sigma_1 + \frac12i\Omega\right)\Sigma_2^{-1/2}
    ~\geqt{(ii)}~ 0.
  \ee
  (i) uses that $\frac{1}{\sqrt2}\Sigma_2^{-1/2}$ is symplectic and hence $\Sigma_2^{-1/2}\Omega\Sigma_2^{-1/2} = 2\Omega$. %
  {One way to see that $\frac{1}{\sqrt2}\Sigma_2^{-1/2}$ is symplectic is as follows. Since $\Sigma_2$ is the covariance matrix of a pure Gaussian state, there exists a symplectic matrix $S$ such that $\Sigma_2=\tfrac12\,SS^\intercal$, and by the Euler (Bloch--Messiah) decomposition we can write $S=O_1 Z O_2$ with $O_1,O_2$ orthogonal symplectic and $Z=\mathrm{diag}(z_1,\ldots,z_n,z_1^{-1},\ldots,z_n^{-1})$, hence $\Sigma_2=\tfrac12\,O_1 Z^2 O_1^\intercal$ and therefore $\frac{1}{\sqrt{2}}\Sigma_2^{-1/2}=O_1 Z^{-1} O_1^\intercal$, which is symplectic because it is a product of symplectic matrices.} Moreover, (ii) uses that $\Sigma_1$ is a valid Gaussian covariance matrix. One also have $\det(\frac12\Sigma_2^{-1/2}\Sigma_1\Sigma_2^{-1/2})=4^{-n}$. This confirms it is indeed a pure Gaussian state covariance matrix. Consequently, the eigenvalues of $\Sigma_2^{-1/2}\Sigma_1\Sigma_2^{-1/2}$ take the form $\{\alpha_1,\cdots,\alpha_n,\alpha_1^{-1},\cdots,\alpha_n^{-1}\}$. Eq.~\eqref{eq:tvd-maha} implies
  \begin{gather}
    \sum_{i=1}^n (\alpha_i - 1)^2 \le \|\Sigma_2^{-1/2}\Sigma_1\Sigma_2^{-1/2}-I_{2n}\|_\mr{F}^2 \le c^2\varepsilon^2.
    \\ \forall i\in[n],\quad \alpha_i \ge 1 - c\varepsilon \ge \frac23.
  \end{gather}
  We are ready to bound the first term:
  \bb
    \frac{1}{\sqrt{\det(\Sigma_1 + \Sigma_2)}} &= \left(\det(\Sigma_2)\det(\Sigma_2^{-1/2}\Sigma_1\Sigma_2^{-1/2}+I_{2n})\right)^{-1/2}
    \\&= \left(4^{-n} \prod_{i=1}^n (1+\alpha_i)(1+\alpha_i^{-1})\right)^{-1/2}
    \\&= \prod_{i=1}^n \left(1 + \frac{(\alpha_i-1)^2}{4\alpha_i}\right)^{-1/2}
    \\&\ge \prod_{i=1}^n \left(1 - \frac{(\alpha_i-1)^2}{8\alpha_i}\right)
    \\&\geqt{(i)} 1 - \sum_{i=1}^n \frac{3}{16}(\alpha_i-1)^2
    \\&\ge 1 - \frac{3}{16} c^2 \varepsilon^2.
  \ee
  Here, (i) uses $\prod_{i=1}^n(1-x_i) \ge 1-\sum_{i=1}^n x_i$ for $x_i\in[0,1]$. Combining the two factors, we have
  \bb
    F(\rho_1,\rho_2) \ge \left(1 - \frac{3}{16} c^2 \varepsilon^2\right)\left(1 - \frac14 c^2\varepsilon^2\right) \ge 1 - \frac{7}{16} c^2 \varepsilon^2.
  \ee
  Therefore,
  \bb
    D_\mr{tr}(\rho_1,\rho_2)\le \frac{\sqrt 7}{4}c\varepsilon < 150\varepsilon.
  \ee
  This completes the proof of Lemma~\ref{le:pure-trace-tvd}.
\end{proof}

\begin{proof}[Proof for Theorem~\ref{th:wigner_relation}]
    The three inequalities follows from Lemma~\ref{le:WignerTV-smaller-than-Dtr}, Lemma~\ref{le:at-most-sqrtn}, and Lemma~\ref{le:pure-trace-tvd}. 
    The sharpness (i.e., exist $\rho,\sigma$ such that $D(\rho,\sigma)=\Theta(\sqrt n)\,\mr{TV}(W_\rho,W_\sigma)$) has been proven by the example discussed in the main text Sec.~\ref{sec:main_wigner}.
\end{proof}

\medskip

\subsection{Nearly-tight bound on learning Gaussian Wigner distributions} \label{sec:wigner_learning}

\begin{theorem}[Nearly-optimal learning of Gaussian Wigner distributions]\label{th:wigner}
    $N=\Omega(n^2/\varepsilon^2)$ copies of an $n$-mode Gaussian state $\rho(\mu,\Sigma)$ are necessary to learn its Wigner distribution to $\varepsilon$ TV distance with probability $2/3$ using any measurements. Furthermore, there exists a Gaussian measurement scheme that achieves this using $N=O\!\left(n^2/\varepsilon^2 + (n +\log\log\log \bar E )\log\log \bar E\right)$ copies given the promise that $\|\Sigma\|_\mr{op}\le \bar E$.
\end{theorem}

The lower bound proof will be postponed to App.~\ref{sec:any-lower-stronger}, where we will also prove Theorem~\ref{th:non-gaussian-lower-stronger} that shows an $\Omega(n^2/\varepsilon^2)$ lower bound for learning Gaussian states to $\varepsilon$ trace distance. The two proofs are almost identical.

\medskip

For the upper bound, we need the following Lemma from \cite{bittel2025energy} shows that one can efficiently learn to unsqueeze any unknown Gaussian state using adaptive Gaussian measurements. 
\begin{lemma}\cite{bittel2025energy}\label{le:unsqueeze}
  For any unknown Gaussian state $\rho(\mu,\Sigma)$, there exists an adaptive Gaussian measurement protocol using {$$O\!\left(\left(n + \log \!\left(\frac{\log\log\|\Sigma^{-1}\|_\mr{op}}{\delta}\right)\right)\log\log\|\Sigma^{-1}\|_\mr{op}\right)$$}%
  copies of $\rho$ to learn a symplectic operator $S$ such that, with probability at least $1-\delta$, the unsqueezed state {$\rho(\mu,\Sigma') \coleq U_{S^{-1}}\rho(\mu,\Sigma) U_{S^{-1}}^\dagger$ }%
  satisfies $\|{\Sigma'}^{-1}\|_\mr{op} \le 4$.
    \end{lemma}

\medskip

\begin{proof}[Proof for the upper bound of Theorem~\ref{th:wigner}]
The upper bound follows from analyzing the adaptive learning algorithm in \cite{bittel2025energy}. For completeness, we describe their algorithm in below:
\begin{enumerate}
    \item Run the learn-to-unsqueeze protocol~\cite{bittel2025energy} in Lemma~\ref{le:unsqueeze} using {\bb
    \Theta\!\left(\left(n + \log \!\left(\frac{\log\log\bar E}{\delta}\right)\right)\log\log\bar E\right)
    \ee}%
    copies of $\rho$ to learn a symplectic matrix $S$ such that {$\rho_\mr{unsqueezed}(\mu,\Sigma')\coleq U_{S^{-1}}\rho(\mu,\Sigma) U^\dagger_{S^{-1}}$ }%
    satisfies $\|{\Sigma'}^{-1}\|_\mr{op}\le 4$ with probability at least $1-\delta/2$. 
    Note that $\|\Sigma^{-1}\|_\mr{op}\le4\|\Sigma\|_\mr{op}\le\bar E$. {Indeed, by Williamson’s decomposition we can write $\Sigma=SDS^\intercal$ with $S$ symplectic and $D\ge I/2$, hence $D^{-1}\le 4D$ and therefore $\Sigma^{-1}=S^{-\intercal}D^{-1}S^{-1}\le 4\,S^{-\intercal}DS^{-1}$; using $S^{-\intercal}=\Omega S\Omega^\intercal$, $S^{-1}=\Omega S^\intercal\Omega^\intercal$, and the fact that $D$ commutes with $\Omega$, we get $S^{-\intercal}DS^{-1}=\Omega S D S^\intercal\Omega^\intercal=\Omega\Sigma\Omega^\intercal$, so $\Sigma^{-1}\le 4\,\Omega\Sigma\Omega^\intercal$, and taking operator norms yields $\|\Sigma^{-1}\|_\mr{op}\le 4\|\Sigma\|_\mr{op}$, since $\Omega$ is orthogonal.}

    In the following step, we replace all copies of $\rho$ by $\rho_\mr{unsqueezed}$. 
    \item Run heterodyne measurements on $2m = O((n^2+n\log\delta^{-1})/\varepsilon^2)$ copies of $\rho_\mr{unsqueezed}$ to get outcomes $v_1,\cdots,v_{2m} \in \mbb R^{2n}$. Define the estimators,
    \bb
      \hat\mu \coleq \frac{1}{m}\sum_{i=1}^m v_i,\quad \hat\Sigma\coleq \frac{1}{2m}\sum_{i=1}^m (v_{2i}-v_{2i-1})(v_{2i}-v_{2i-1})^T - \frac12 I_{2n}.
    \ee
    Return $\mc N(\hat\mu,S\hat\Sigma S^T)$ as an estimate of the Wigner function of $\rho$.
  \end{enumerate}
  Conditioned on Step 1 succeeding, since TV distance is preserved under invertible transformations, we just need to learn the Wigner function of $\rho_\mr{unsqueezed}$ to $\varepsilon$ TV distance. For notational simplicity, we will just relabel $\rho_\mr{unsqueezed}$ by $\rho(\mu,\Sigma)$, with the guarantee that $\|\Sigma^{-1}\|_\mr{op}\le\bar E$.  
  Following standard concentration results, see e.g. \cite[Appendix C]{ashtiani2020near}, the estimators defined in Step 2 with probability at least $1-\delta/2$ satisfy
  \begin{gather}
    (\hat\mu - \mu)^T (\Sigma +\frac12 I_{2n})^{-1} (\hat\mu - \mu) \le \varepsilon^2/2,\\
    \left\|(\Sigma+\frac12 I_{2n})^{-1/2}(\hat\Sigma_0 +\frac12 I_{2n})(\Sigma+\frac12 I_{2n})^{-1/2} - I_{2n}\right\|_\mr{F}^2 \le \varepsilon^2/2.
  \end{gather}
  Consider the spectrum decomposition $\Sigma=\sum_{i=1}^{2n}\lambda_i\ketbra{e_i}{e_i}$ which satisfies $\lambda_i^{-1}\le 4$.
  Therefore,
  \bb
    (\Sigma+\frac12 I_{2n})^{-1} &= \sum_{i=1}^{2n}(\lambda_i + \frac12)^{-1}\ketbra{e_i}{e_i}
    \ge \frac13 \sum_{i=1}^{2n} \lambda_i^{-1}\ketbra{e_i}{e_i} 
    = \frac13 \Sigma^{-1}.
  \ee
  Consequently,
  \bb
    \|\Sigma^{-1/2}(\hat\mu-\mu)\|_2^2&=(\hat\mu-\mu)^T\Sigma^{-1}(\hat\mu-\mu)\le \frac32\varepsilon^2.\\
    \|\Sigma^{-1/2}\hat\Sigma\Sigma^{-1/2}-I\|_\mr{F}^2 &=\Tr\left(\Sigma^{-1}(\hat\Sigma-\Sigma)\Sigma^{-1}(\hat\Sigma-\Sigma)\right)\\
    &\le 9 \Tr\left((\Sigma+\frac12 I)^{-1}(\hat\Sigma-\Sigma)(\Sigma+\frac12 I)^{-1}(\hat\Sigma-\Sigma)\right)
    \\&\le \frac92\varepsilon^2.
  \ee
  We then use Lemma~\ref{le:mahalanobis} to conclude that
  \bb
    \mr{TV}(\mc N(\mu,\Sigma),\mc N(\hat\mu,\hat\Sigma)) \le \frac32\varepsilon.
  \ee
  Rescaling $\varepsilon$ by a factor of $\frac32$ achieves the learning yields what we want. Note that, by choosing $\delta=1/3$ and applying the union bound, the whole protocol succeed with at least $2/3$ chance.
\end{proof}

\section{Lower bounds on Gaussian state learning}

\subsection{Nearly-tight lower bounds for general Gaussian measurements}\label{sec:gaussian-lower}

\begin{theorem}\label{th:gaussian-lower}
  For any algorithm that learns an unknown zero-mean $n$-mode Gaussian state $\rho(0,\Sigma)$ to trace distance precision $\varepsilon \le 1/36$ with probability at least $2/3$ using general (possibly adaptive, ancilla-assisted, multi-copy) Gaussian measurements, the number of copies must satisfy $N = \Omega\left({\frac{n^3}{\varepsilon^{2}}}\right)$. This holds even with the promise that $\|\Sigma\|_\mr{op}=O(1)$ and/or that $\rho$ is a passive Gaussian state.
\end{theorem}

The above theorem actually applies for any (collective) POVM whose elements all have non-negative Wigner functions, including Gaussian measurements as a special case. The key is that such measurements can be simulated using samples from the Wigner distribution of the unknown Gaussian state.

Theorem~\ref{th:gaussian-lower} shows that the protocol in \cite{bittel2025energy} with sample complexity {\bb
N=O\!\left( \frac{n^3}{\varepsilon^2}+ (n+\log\log\log E)\log\log E \right)
\ee}%
is sample-optimal among all Gaussian measurement schemes, up to a double-logarithmic factor in $E$. {Here $E$ is a known parameter such that the unknown covariance matrix is promised to satisfy $\|\Sigma\|_{\mathrm{op}}\le E$.} 
Also note that the sample complexity to learn a classical $2n$-dimensional Gaussian distribution to $\varepsilon$ TVD is $\widetilde\Theta(n^2/\varepsilon^2)$, different from the quantum case by a factor of $n$. 
Finally, we emphasize that the lower bound does not apply for non-Gaussian measurements. Actually, it seems plausible from the proof that non-Gaussian measurements could potentially yield a better sample complexity. 
In Section~\ref{sec:passive-upper}, we show that non-Gaussian measurements can indeed achieve $O(n^2/\varepsilon^2)$ sample complexity for learning passive Gaussian states, breaking the lower bound in Theorem~\ref{th:gaussian-lower} and shows a separation between Gaussian and non-Gaussian measurements.

\begin{proof}[Proof of Theorem~\ref{th:gaussian-lower}]
  We will lift the construction of a large family of hard-to-distinguish Gaussian distributions from \cite{ashtiani2020near} to Gaussian quantum states. The key observation is that, when the Gaussian distributions in that family have pairwise TVD $\Theta(\varepsilon_0)$, the corresponding Gaussian states can have pairwise trace distance as large as $\Omega(\sqrt n\varepsilon_0)$. 
  The way to establish the trace distance lower bound is to construct an explicit non-Gaussian POVM and use the data-processing property. 
  We then use the Gaussian measurement assumptions to connect our problem to learning classical Gaussian distributions, and thereby use Fano's inequality to conclude hardness.
  
  \medskip
  Consider the following $2n$-by-$2n$ positive definite matrices:
  \bb
    \Sigma_a \coleq \frac12 I_{2n} + \frac{\varepsilon}{2n}\begin{pmatrix}
      U_aU_a^T &0_n \\ 0_n& U_aU_a^T
    \end{pmatrix}
  \ee
  Here $\varepsilon\le 5$. $U_a$ is an $n$-by-$s$ real matrix with orthogonal columns. We set $s = \ceil{n/9}$. Since $\Sigma\ge\frac12 I$, it is a valid covariance matrix of an $n$-mode Gaussian state. Here we arrange the quadrature operators as $(\hat x_1,\cdots,\hat x_n, \hat p_1,\cdots,\hat p_n)$, and define $\rho_a$ as a Gaussian state with zero mean and covariance matrix $\Sigma_a$.  
  Note that $\rho_a$ is a passive Gaussian state. To see this, define $V_a = [U_a; W_a]$ as an $n$-by-$n$ orthogonal matrix that contains $U_a$ as its first $s$ columns, and apply the following orthogonal symplectic transformation,
  \bb
  \Sigma_a&\mapsto
    \begin{pmatrix}
      V_a^T & 0_n\\ 0_n & V^T_a
    \end{pmatrix}  \Sigma_a
    \begin{pmatrix}
      V_a & 0_n\\ 0_n &V_a
    \end{pmatrix} = \frac12 I_{2n} + \frac{\varepsilon}{2n}\begin{pmatrix}
      \begin{pmatrix}
        I_s & 0 \\ 0 & 0
      \end{pmatrix}
      &0_n \\ 0_n& \begin{pmatrix}
        I_s & 0 \\ 0 & 0
      \end{pmatrix}
    \end{pmatrix}
  \ee
  The R.H.S. is clearly a thermal state covariance matrix, thus $\rho_a$ is passive.
  
  \medskip

  We will use the following results from \cite{ashtiani2020near}:
  \begin{lemma}\cite[Lemma 6.4]{ashtiani2020near}\label{le:ashtiani6-4}
    Suppose $n > 9$. There exists $M=2^{\Omega(n^2)}$ $n$-by-$\ceil{n/9}$ matrices $\{U_a\}_{a=1}^M$ with orthogonal columns and satisfy $\|U_a^TU_b\|_{\mr F}^2\le \frac n{18}$ for all distinct pairs $a,b\in[M]$.
  \end{lemma}
  We can thus pick $\{U_a\}_{a=1}^M$ from the above lemma to define a family of $n$-mode Gaussian states $\{\rho_a\}_{a=1}^M$ of size $M=2^{\Omega(n^2)}$. Now we show that, for any distinct pair $a,b\in[M]$, it holds that $D_{\tr}(\rho_a,\rho_b)\ge \varepsilon/90$. 
  To see this, apply the same transformation as before, now on $\rho_b$:
  \bb
    \Sigma_b&\mapsto
    \begin{pmatrix}
      V_a^T & 0_n\\ 0_n & V_a^T
    \end{pmatrix}  \Sigma_b
    \begin{pmatrix}
      V_a & 0_n\\ 0_n &V_a
    \end{pmatrix} = \frac12 I_{2n} + \frac{\varepsilon}{2n}\begin{pmatrix}
      V_a^TU_bU_b^TV_a &0_n \\ 0_n & V_a^TU_bU_b^TV_a
    \end{pmatrix}.
  \ee
  Denoted the states after the transformation as $\rho'_a$ and $\rho'_b$, respectively.
  Now, we trace out the first $s$ modes and apply Fock basis measurement (i.e. photon number counting) on the remaining $n-s$ modes. 
  For $\rho'_a$, the outcome is obviously vacuum with probability 1. For $\rho'_b$, the probability of obtaining vacuum is
  \bb
    \Pr(0_{n-s}|b) &= \braket{0_{n-s}|\Tr_{\le s}\rho'_b|0_{n-s}}\\
    &\eqt{(i)} \left(\det\left(I_{2(n-s)} + \frac{\varepsilon}{2n}\begin{pmatrix}
      W_a^TU_bU_b^TW_a &0_n \\ 0_n & W_a^TU_bU_b^TW_a
    \end{pmatrix}\right)  \right)^{-1/2}\\
    &\eqt{} \left(\det\left(I_{n-s} + \frac{\varepsilon}{2n}W_a^TU_bU_b^TW_a\right)  \right)^{-1}\\
    &\leqt{(ii)} \left(1 + \frac{\varepsilon}{2n}\Tr(W_a^TU_bU_b^TW_a)\right)^{-1}\\
    &\leqt{(iii)} 1 - \frac{\varepsilon}{3n}\Tr(W_a^TU_bU_b^TW_a)\\
    &\eqt{(iv)} 1 - \frac{\varepsilon}{3n}(s - \|U_a^TU_b\|_{\mr F}^2)\\
    &\leqt{(v)} 1 - \frac{\varepsilon}{54}.
  \ee
  Here, (i) uses the fact that vacuum state is a Gaussian state, and thus the probability can be computed using the fidelity formula between two Gaussian states; (ii) uses $\det(I+X)\ge 1+\Tr X$ for any positive semidefinite matrix $X$; (iii) uses $(1+x)^{-1}\le 1 - \frac{2}{3}x$ for $0\le x\le 1/2$. Note that $\varepsilon\le5$ suffices; (iv) uses the fact that $\Tr(U_a^TU_bU_b^TU_a)+\Tr(W_a^TU_bU_b^TW_a) = \Tr(U_bU_b^T) = s$; (v) uses Lemma~\ref{le:ashtiani6-4} and the choice of $s=\ceil{n/9}$. Using the data-processing inequality of trace distance, we conclude that
  \bb
    D_\mr{tr}(\rho_a,\rho_b) \ge \varepsilon/54,\quad~\text{for all}~ a\neq b\in[M].
  \ee
  Next, we use the following standard fact about POVM with non-negative Wigner functions~\cite{wigner1932quantum,hillery1984distribution}. In particular, it holds for all Gaussian measurements.
  \begin{lemma}\cite{mari2012positive}\label{le:positive-POVM}
    Given $N$ i.i.d. copies of an unknown $n$-mode Gaussian state $\rho(\mu,\Sigma)$ and a POVM $\{M_x\}_x$ on the $N \times n$-mode Hilbert space, such that each POVM element $M_x$ has non-negative Wigner function, then the measurement outcome probability distribution $p(x|\rho^{\otimes N})=\Tr(M_x\rho^{\otimes N})$ can be exactly simulated by sampling $N$ i.i.d. samples from the Gaussian distribution $\mc N(\mu, \Sigma)$ followed by classical post-processing independent of $(\mu,\Sigma)$.
  \end{lemma}
  \begin{proof}
    The outcome probability can be written as
    \bb
      p(x|\rho^{\otimes N}) =  (2\pi)^{nN}\int \mathrm{d}^{2nN}\!\alpha~ W_{M_x}(\alpha) W_{\rho^{\otimes N}}(\alpha),
    \ee
    where $W_{O}(\alpha)$ is the Wigner function of operator $O$. By our assumptions on $M_x$ and the Gaussianity of $\rho$, both Wigner functions are non-negative. In particular, $W_{\rho^{\otimes N}}$ can be interpreted as a probability density function over $\alpha$. Thus, this distribution can be simulated by first sampling $\alpha$ from $W_{\rho^{\otimes N}}(\alpha) = \mc N(\mu,\Sigma)^{\otimes N}$, and then sampling $x$ from the conditional distribution $p(x|\alpha) = (2\pi)^{nN} W_{M_x}(\alpha)$,\footnote{In our convention, the identity operator has a Wigner function of constant value $(2\pi)^{-nN}$.} which is independent of $(\mu,\Sigma)$.
  \end{proof}
  
  Consider the following communication task: Suppose Alice and Bob both know exactly the description of $\{\Sigma_a\}_{a=1}^M$. Now, Alice uniformly randomly chooses $a\in[M]$, and then prepares $N$ i.i.d. (classical) samples $x_{1:N}$ from $p_a\coleq\mc N(0,\Sigma_a)$ and sends them to Bob. Bob then tries to guess Alice's choice of $a$ by processing the samples. 
  Suppose there exists a Gaussian measurement protocols that uses $N$ i.i.d. copies of $\rho(0,\Sigma)$ to learn the state to trace distance precision $\varepsilon'=\varepsilon/180 \le 1/36$ with probability at least $2/3$. Then, thanks to Lemma~\ref{le:positive-POVM} and the fact that $D_\mr{tr}(\rho_a,\rho_b)\ge \varepsilon/90$, Bob can learn $\rho_a$ from the classical samples he received and hence guess $a$ correctly with probability at least $2/3$. Using Fano's inequality~\cite{cover1999elements}, this implies
  \bb
    I(x_{1:N}:a) \ge \frac23 \log M - \log 2 = \Omega(n^2).
  \ee
  On the other hand, the mutual information can be upper bounded by
  \bb
    I(x_{1:N}:a) &\eqt{(i)} \E_{a} \mr{KL}( p_a^{\otimes N} \| \E_{b} p_b^{\otimes N} )\\
    &\leqt{(ii)} \E_{a,b} \mr{KL}( p_a^{\otimes N} \| p_b^{\otimes N} )\\
    &\eqt{(iii)} \E_{a,b}N~\mr{KL}( p_a\| p_b).
  \ee
  Here, (i) is by definition of mutual information; (ii) applies the convexity of KL divergence; (iii) uses the chain rule of KL divergence. The last step is to upper bound $\mr{KL}(p_a\|p_b)$, which is essentially done in \cite{ashtiani2020near} but we include the calculation for completeness. Note that we only need to consider $a\neq b$. Also notice that $\Sigma_a,\Sigma_b$ have the same eigenvalues,
  \bb\label{eq:KL-calculation}
    \mr{KL}(p_a\|p_b) &\equiv \mr{KL}(\mc N(0,\Sigma_a)\|\mc N(0,\Sigma_b))\\
    &\eqt{(i)} \Tr(\Sigma_b^{-1}\Sigma_a-I_{2n})\\
    &\eqt{} 2\Tr\left((I_n-\frac{\varepsilon}{n+\varepsilon}U_bU_b^T)(I_n+\frac{\varepsilon}{n}U_aU_a^T)-I_n\right)\\
    &\eqt{} - \frac{2\varepsilon s }{n+\varepsilon} + \frac{2\varepsilon s}{n} - \frac{2\varepsilon^2}{n(n+\varepsilon)}\|U_b^TU_a\|_{\mr F}^2\\
    &\leqt{} \frac{2\varepsilon^2 s}{n(n+\varepsilon)}\\
    &\eqt{(ii)} O\left(\frac{\varepsilon^2}{n}\right).
  \ee
  Here (i) uses the formula of KL divergence between two zero-mean Gaussian distributions. Note that the determinant term vanishes; (ii) uses $s=\ceil{n/9}$. Putting everything together, we conclude that
  \bb
    N = \Omega\left(\frac{n^3}{\varepsilon^2}\right).
  \ee
  Note that all Gaussian states used in the proof have $\|\Sigma_a\|_\mr{op} = O(1)$.  This completes the proof. 
\end{proof}

\subsection{Lower bounds for general POVM measurements}\label{sec:any-lower}

We first prove the lower bound of $\Omega(n^2/(\varepsilon^2\log(n/\varepsilon))$ that holds even with the promise of pure states. Then, in Sec.~\ref{sec:any-lower-stronger}, we prove the tighter lower bound of $\Omega(n^2/\varepsilon^2)$ for learning general Gaussian states, together with the same lower bound for learning Gaussian Wigner distributions as in Theorem~\ref{th:wigner}.

\begin{theorem}\label{th:non-gaussian-lower}
  For any algorithm that learns an unknown zero-mean $n$-mode Gaussian state $\rho(0,\Sigma)$ to trace distance precision $\varepsilon\le0.03$ with probability at least $2/3$ using any POVM measurements, the number of copies must satisfy $N = \Omega\left({\frac{n^2}{\varepsilon^{2}\log(n/\varepsilon^2)}}\right)$. This holds even with the promise that $\|\Sigma\|_\mr{op}=O(1)$ and/or the unknown state is a zero-mean pure state.
\end{theorem}
Up to a logarithmic factor, this lower bound has the same scaling as learning a $2n$-dimensional classical Gaussian distribution $\widetilde{\Omega}(n^2/\varepsilon^2)$~\cite{ashtiani2020near}. Also note that this theorem allows collective and non-Gaussian POVMs on multiple copies of the unknown states.

\begin{proof}
  The proof is similar to that of Theorem~\ref{th:gaussian-lower}, except that we will use a slightly different family of \emph{pure} Gaussian states, and use the Holevo theorem to upper bound the mutual information in the Fano's inequality step.

  \medskip

  Consider the following $2n$-by-$2n$ positive definite matrices:
  \bb\label{eq:def_of_sigma_prime}
    \Sigma'_a \coleq \frac12 \begin{pmatrix}
      I_n - \frac{\varepsilon}{\sqrt n + \varepsilon}U_aU_a^T&0_n \\ 0_n& I_n + \frac\varepsilon{\sqrt n} U_aU_a^T.
    \end{pmatrix}
  \ee
  Here we choose $\varepsilon\le1$ and $\{U_a\}_{a=1}^M$ is the same ensemble of $n$-by-$\ceil{n/9}$ column-orthogonal matrices of size $M=2^{\Omega(n^2)}$ such that $\|U_a^T U_b\|_\mr{F}^2\le n/18$ for all distinct pairs $a,b\in[M]$ from Lemma~\ref{le:ashtiani6-4}. Denote $s=\ceil{n/9}$. We define $\sigma_a$ to be the zero-mean $n$-mode Gaussian state with covariance matrix $\Sigma'_a$. Note that $\Sigma_a'$ represents a pure Gaussian state, as $\Sigma'_a + \frac12i\Omega \ge 0$ and $\det\Sigma'_a = 4^{-n}$.

  We first show the trace distance between any distinct pair of states within the ensemble is $\Omega(\varepsilon)$:
  \bb
    1 - D^2_\mr{tr}(\sigma_a,\sigma_b) &\eqt{(i)} {F(\sigma_a,\sigma_b)}\\
    &\eqt{(ii)} {\left(\det\begin{pmatrix}
      I_n - \frac{\varepsilon}{2(\sqrt n+\varepsilon)}(U_aU_a^T + U_bU_b^T) & 0_n \\ 0_n & I_n + \frac{\varepsilon}{2\sqrt n}(U_aU_a^T + U_bU_b^T)
    \end{pmatrix}\right)^{-\frac12}}\\
    &=\left(\det \left(I_n + \frac{\varepsilon^2}{4(n+\varepsilon\sqrt n)}(U_aU_a^T-U_bU_b^T)^2\right)\right)^{-1/2}\\
    &\leqt{(iii)} \left( 1 + \frac{\varepsilon^2}{8n}\Tr\left((U_aU_a^T-U_bU_b^T)^2\right)\right)^{-1/2}\\
    &\leqt{(iv)} 1 - \frac{\varepsilon^2}{20n}\Tr\left((U_aU_a^T-U_bU_b^T)^2\right)\\
    &\eqt{} 1 - \frac{\varepsilon^2}{20n} \cdot \left(2s - 2\|U_a^TU_b\|_\mr{F}^2\right)
    \\&\leqt{(v)} 1 - \frac{\varepsilon^2}{180}.
  \ee
  Here, (i) uses the relation between trace distance and fidelity between two pure states; (ii) uses the fidelity formula between pure Gaussian states; (iii) uses $\det(I + X)\ge 1+ \Tr X$ for all $X\ge 0$; (iv) uses $(1+x)^{-1/2}\le 1 - \frac{2}{5}x$ for $0\le x\le 1/3$. Note that assuming $\varepsilon\le1$ suffices; (v) uses the property of our state ensemble from Lemma~\ref{le:ashtiani6-4}. Therefore, we have
  \bb
    D_\mr{tr}(\sigma_a,\sigma_b) \ge \frac{\varepsilon}{6\sqrt5},\quad~\text{for all}~ a\neq b\in[M].
  \ee
  Next, we again consider a communication task: Suppose Alice and Bob both know the ensemble $\{\sigma_a\}_{a=1}^M$. Now, Alice uniformly randomly chooses $a\in[M]$, and then prepares $N$ i.i.d. copies of $\sigma_a$ and sends them to Bob (Note the difference from the proof of Theorem~\ref{th:gaussian-lower} -- here Alice sends quantum states instead of classical samples). Bob then tries to guess Alice's choice of $a$ by performing any POVM measurement on the quantum states and conducting any classical post-processing. Denote Bob's guess by $\hat a$. Suppose there exists a learning protocol that uses $N$ copies to learn a zero mean Gaussian states to trace distance precision $\varepsilon'=\varepsilon/(12\sqrt5)$ with probability at least $2/3$. Then, Bob can use this protocol to learn $\sigma_a$ and hence guess $a$ correctly with probability at least $2/3$. Using Fano's inequality, this implies
  \begin{equation}
    I(\hat a : a) \ge \frac23 \log M - \log 2 = \Omega(n^2).
  \end{equation}
  On the other hand, Holevo theorem \cite{holevo1973bounds} upper bounds the amount of information one can access from a quantum state ensemble via any POVM measurements:
  \bb
    I(\hat a : a) &\le S\!\left(\frac1M\sum_{a=1}^M\sigma_a^{\otimes N}\right) - \frac1M\sum_{a=1}^M S(\sigma_a^{\otimes N})\\
    &= S(\bar{\sigma}_N).
  \ee
  Here $S$ denotes the von Neumann entropy. We define $\bar\sigma_N\coleq\frac1M\sum_{a=1}^M\sigma_a^{\otimes N}$. The second term vanishes as $\sigma_a$ is pure. Computing $S(\bar{\sigma}_N)$ is challenging as $\bar{\sigma}_N$ is not a Gaussian state. However, one can still compute the mean values and covariance matrix of $\bar\sigma_N$:
  \bb
    \mu_{\bar\sigma_N} = 0,\quad
    \Sigma_{\bar\sigma_N} = \frac1M\sum_{a=1}^M {{\Sigma'_a}}^{\oplus N}.
  \ee
  Now we define $\sigma^\star_N$ as the Gaussian state with zero mean and covariance matrix $\Sigma_{\bar\sigma_N}$. 
  {Since $\frac1M\sum_{a=1}^M {{\Sigma'_a}}^{\oplus N}=\left(\frac1M\sum_{a=1}^M {{\Sigma'_a}}\right)^{\oplus N}$, it follows that }%
  $\sigma^\star_N = (\sigma^{\star})^{\otimes N}$ where $\sigma^\star$ is an $n$-mode Gaussian states with mean zero and covariance matrix $\frac1M\sum_{a=1}^M \Sigma'_a$.
  Using the fact that Gaussian states maximize the von Neumann entropy among all states with the same mean and covariance matrix~\cite{holevo1999capacity,wolf2006extremality}, we have that
  \bb\label{eq:bound_the_entropy}
    S(\bar\sigma_N) \le S(\sigma^\star_N) = N\cdot S(\sigma^\star).
  \ee
  Recall the entropy of a Gaussian state is given by $S(\rho(0,\Sigma)) = \sum_{i=1}^n g(\nu_i - 1/2)$ where $\{\nu_i\}_{i=1}^n$ are the symplectic eigenvalues of $\Sigma$ and $g(x) = (x+1)\log(x+1) - x\log x$~\cite{weedbrook2012gaussian}. Note that $g(x)$ is monotonically increasing for $x\ge 0$. The covariance matrix of $\sigma^\star$ is
  \bb
    \Sigma^\star = \frac12 \begin{pmatrix}
      I_n - \frac{\varepsilon}{\sqrt n + \varepsilon}\frac1M\sum_{a=1}^MU_aU_a^T&0_n \\ 0_n& I_n + \frac\varepsilon{\sqrt n} \frac1M\sum_{a=1}^MU_aU_a^T
    \end{pmatrix}\,.
  \ee
  {Let us determine the symplectic eigenvalues of $\Sigma^\star$. Define $B \coloneqq \frac{1}{M}\sum_{a=1}^M U_aU_a^\top$, and let $
OBO^\top=\mathrm{diag}(b_1,\dots,b_n)$ be its spectral decomposition, with $O\in\mathbb{R}^{n\times n}$ orthogonal. Since $B$ is an average of orthogonal projectors, each eigenvalue $b_i$ satisfies $b_i\in[0,1]$ for all $i$. Moreover, $S\coloneqq I_2\otimes O$ is an (orthogonal) symplectic matrix, and satisfies
\begin{equation}
S\,\Sigma^\star S^\top
=\frac12
\begin{pmatrix}
\mathrm{diag}\!\left(1-\frac{\varepsilon}{\sqrt n+\varepsilon}b_1,\dots,1-\frac{\varepsilon}{\sqrt n+\varepsilon}b_n\right) & 0_n\\[2pt]
0_n & \mathrm{diag}\!\left(1+\frac{\varepsilon}{\sqrt n}b_1,\dots,1+\frac{\varepsilon}{\sqrt n}b_n\right)
\end{pmatrix}.
\end{equation}
Hence, $\Sigma^\star$ has the same symplectic eigenvalues as the direct sum 
\begin{equation}
\bigoplus_{i=1}^n
\frac12
\begin{pmatrix}
1-\frac{\varepsilon}{\sqrt n+\varepsilon}b_i & 0\\
0 & 1+\frac{\varepsilon}{\sqrt n}b_i
\end{pmatrix}.
\end{equation}
Consequently, since a $2\times2$ positive-definite matrix $V$ has a single symplectic eigenvalue equal to $\sqrt{\det V}$, the symplectic eigenvalues of $\Sigma^\star$ are given, for $i=1,\dots,n$, by
\begin{align}
\nu_i
&=\frac12\sqrt{\left(1-\frac{\varepsilon}{\sqrt n+\varepsilon}b_i\right)\left(1+\frac{\varepsilon}{\sqrt n}b_i\right)} \\
&=\frac12\sqrt{1+\frac{\varepsilon^2}{n+\varepsilon\sqrt n}\,b_i(1-b_i)} \\
&\le \frac12\left(1+\frac{\varepsilon^2}{8n}\right),
\end{align}
where the final bound follows from $b_i(1-b_i)\le \tfrac14$ and $n+\varepsilon\sqrt n\ge n$.
} %
  Thus,
  \bb
    S(\sigma^\star) &\le n~g\!\left(\frac{\varepsilon^2}{16n}\right)\\
    &= n \left(\log(1+\frac{\varepsilon^2}{16n}) + \frac{\varepsilon^2}{16n}\log(1+\frac{16n}{\varepsilon^2}) \right)\\
    &= O\left(\varepsilon^2\log\frac{n}{\varepsilon^2}\right)\\
    &= O\left(\varepsilon^2\log\frac{n}{\varepsilon}\right).
  \ee
  Putting everything together, we conclude that
  \bb
    N = \Omega\left(\frac{n^2}{\varepsilon^2\log(n/\varepsilon)}\right).
  \ee
  Note that all $\rho_a$ are pure Gaussian states with $\|\Sigma_a\|_\mr{op} = O(1)$. This completes our proof.
\end{proof}

\subsubsection{An improved lower bound}~\label{sec:any-lower-stronger}
In this section, we prove a lower bound of $\Omega(n^2/\varepsilon^2)$ as in the following theorem. We also prove the same lower bound for learning Gaussian Wigner distributions to $\varepsilon$ TV distance as in Theorem~\ref{th:wigner}.
\begin{theorem}\label{th:non-gaussian-lower-stronger}
  For any algorithm that learns an unknown zero-mean $n$-mode Gaussian state $\rho(0,\Sigma)$ to trace distance precision $\varepsilon$ with probability at least $2/3$ using any POVM measurements, the number of copies must satisfy $N = \Omega({{n^2}/{\varepsilon^{2}}})$. This holds even with the promise that $\|\Sigma\|_\mr{op}=O(1)$.
\end{theorem}

Here we introduce some additional facts.
A Gaussian state is non-degenerate if and only if $\Sigma+\frac{i}2\Omega>0$. In this case, it can be written as a Gibbs state of a Hamiltonian operator which is quadratic in the position and momentum operators. In the zero-mean case, this reads: $\rho=\frac{e^{-\hat{R}{}^{\intercal}H\hat{R}}}{\Tr[e^{-\hat{R}{}^{\intercal}H\hat{R}}]}$, with $H$ being a real positive-definite $2n\times 2n$ matrix. The matrix $H$ is completely determined by the covariance matrix. We will use that if the Williamson form of $\Sigma$ is $\Sigma=S^{-1}DS^{-\intercal}$, the Hamiltonian matrix is $H=S^{\intercal}f(D)S$, with $x=(\frac{1}{2}\coth(f(x))))$, i.e. $f(x)=\frac{1}{2}\log\left(\frac{2x+1}{2x-1}\right)$. We will need an expression for the symmetrized relative entropy of non-degenerate Gaussian states, as shown in Proposition 4.1 in~\cite{fanizza2025} and reproduced here for convenience:

\begin{lemma}
Let $\rho_1$, $\rho_2$ be two zero-mean Gaussian states of covariance matrices $\Sigma_1$, $\Sigma_2$, respectively, and with the corresponding Hamiltonian matrices $H_1$, $H_2$. Then
\begin{equation}
D(\rho_1\|\rho_2)+D(\rho_2\|\rho_1)=\Tr[(\Sigma_1-\Sigma_2)(H_2-H_1)].
\end{equation}
\end{lemma}

\begin{proof}
First, note that
\begin{align}
D(\rho_1\|\rho_2)&=\Tr[\rho_1 \log \rho_1]-\Tr[\rho_1 \log \rho_2]\\&=\Tr[\rho_1 (-\hat{R}{^{\intercal}}H_1\hat{R})]-\log \Tr[e^{-\hat{R}{^{\intercal}}H_1\hat{R}}]-\Tr[\rho_1 (-\hat{R}{^{\intercal}}H_2\hat{R})]+\log \Tr[e^{-\hat{R}{^{\intercal}}H_2\hat{R}}]\\
&=\Tr[\rho_1(\hat{R}{}^{\intercal}(H_2-H_1)\hat{R})]-\log \Tr[e^{-\hat{R}{^{\intercal}}H_1\hat{R}}]+\log \Tr[e^{-\hat{R}{^{\intercal}}H_2\hat{R}}]\\
&=\Tr[\rho_1(\hat{R}{}^{\intercal}(H_2-H_1)\hat{R})]-\log \Tr[e^{-\hat{R}{^{\intercal}}H_1\hat{R}}]+\log \Tr[e^{-\hat{R}{^{\intercal}}H_2\hat{R}}],
\end{align}
and that 
\begin{align}
\Tr[\rho_1(\hat{R}{}^{\intercal}(H_2-H_1)\hat{R})]&=\sum_{i,j}\Tr[\rho_1\hat{R}_i\hat{R}_j(H_2-H_1)_{ij}]\\
&=\frac{1}{2}\sum_{ij}\Tr[\rho_1\{\hat{R}_i,\hat{R}_j\}(H_2-H_1)_{ij}]\\
&=\Tr[\Sigma_1(H_2-H_1)].
\end{align}
Similarly,
\begin{align}
D(\rho_2\|\rho_1)&=\Tr[\Sigma_2(H_1-H_2)]-\log \Tr[e^{-\hat{R}{^{\intercal}}H_2\hat{R}}]+\log \Tr[e^{-\hat{R}{^{\intercal}}H_1\hat{R}}].
\end{align}
Putting all together,
\begin{equation}
D(\rho_1\|\rho_2)+D(\rho_2\|\rho_1)=\Tr[(\Sigma_1-\Sigma_2)(H_2-H_1)].
\end{equation}
\end{proof}

\noindent We can now prove the following bound on the Holevo quantity:

\begin{lemma}\label{le:marco_holevo_bound}
Let $(p_a,\sigma_a^{\otimes N})$, $a\in[M]$, $N\in\mbb N_+$ be an ensemble of $N$ i.i.d. copies of Gaussian states, with the corresponding single-copy covariance matrices $\Sigma_a$ and Hamiltonian matrices $H_a$, with Holevo quantity 
\begin{equation}\chi(\{p_a,\sigma_a^{\otimes N}\})=D\left(\sum_{a=1}^{M}p_a\ketbra{a}{a}\otimes \sigma_a^{\otimes N}\Big\| \sum_{a=1}^{M}p_a\ketbra{a}{a}\otimes \sum_{a'=1}^{M}p_{a'} \sigma_{a'}^{\otimes N}\right).
\end{equation} 
Then
\begin{equation}
\chi(\{p_a,\sigma_a^{\otimes N}\})\leq N\sum_{a,a'=1}^{M}p_{a}p_{a'}\Tr[(\Sigma_a-\Sigma_{a'})(H_{a'}-H_{a})].
\end{equation}
\end{lemma}
\begin{proof}
First, note that $\sigma_a^{\otimes N}=\sum_{a'=1}^{M}p_{a'}\sigma_a^{\otimes N}$.
By joint convexity of the relative entropy,
\bb
\chi(\{p_a,\sigma_a^{\otimes N}\})&=D\left(\sum_{a=1}^{M}p_a\ketbra{a}{a}\otimes \sum_{a'=1}^{M}p_{a'}\sigma_a^{\otimes N}\Big\| \sum_{a=1}^{M}p_a\ketbra{a}{a}\otimes \sum_{a'=1}^{M}p_{a'} \sigma_{a'}^{\otimes N}\right)\\
&\leq \sum_{a,a'=1}^{M}p_{a}p_{a'}D\left(\ketbra{a}{a}\otimes \sigma_{a}^{\otimes N}\big\|\ketbra{a}{a}\otimes \sigma_{a'}^{\otimes N}\right)\\
&\eqt{(i)}N\sum_{a,a'=1}^{M}p_{a}p_{a'}D\left(\sigma_{a}\| \sigma_{a'}\right)\\
&\le N\sum_{a,a'=1}^{M}p_{a}p_{a'}\Tr[(\Sigma_a-\Sigma_{a'})(H_{a'}-H_{a})].
\ee
Here (i) uses the additivity of the quantum relative entropy.
\end{proof}

\noindent We are ready to prove the main result:
\begin{proof}[Proof for Theorem~\ref{th:non-gaussian-lower-stronger} and the lower bound of Theorem~\ref{th:wigner}]
    We start with the same ensemble as in the previous section, as in Eq.~\eqref{eq:def_of_sigma_prime}, but multiply all covariance matrices by a factor of $2$. That is,
    \bb
        \Sigma_a \coleq 2\Sigma_a'= \begin{pmatrix}
      I_n - \frac{\varepsilon}{\sqrt n + \varepsilon}U_aU_a^T&0_n \\ 0_n& I_n + \frac\varepsilon{\sqrt n} U_aU_a^T
      \end{pmatrix}.
    \ee
    With slight recycling of notations, let $\rho_a$ be the zero-mean Gaussian state with covariance matrix $\Sigma_a$. Still use $\sigma_a$ to denote the zero-mean pure Gaussian state with covariance matrix $\Sigma_a'$. The pairwise trace distance of the new ensemble satisfies, given that $\varepsilon$ is smaller than some constant threshold,
    \bb
        D_\mr{tr}(\rho_a,\rho_b) &\geqt{(i)} \frac{\sqrt2}{400}\mr{TV}(W_{\rho_a},W_{\rho_b})\\
        &\eqt{(ii)}\frac{\sqrt2}{400}\mr{TV}(W_{\sigma_a},W_{\sigma_b})\\
        &\geqt{(iii)}\frac{\sqrt2}{60000}D_\mr{tr}(\sigma_a,\sigma_b)\\
        &\eqt{(iv)}\Omega(\varepsilon).
    \ee
    Here (i) uses Lemma~\ref{le:WignerTV-smaller-than-Dtr}. Note that we can always assume $\varepsilon$ to be upper bounded by a sufficiently small constant to ensure Lemma~\ref{le:WignerTV-smaller-than-Dtr} can be applied. (ii) uses the fact that rescaling the covariance matrices of two classical zero-mean Gaussian distributions by a common positive factor does not change their TV distance; (iii) uses Lemma~\ref{le:pure-trace-tvd} and the fact that $\sigma_a,\sigma_b$ are pure Gaussian states. (iv) uses the trace distance bound from the last section.

    Thanks to this, one can consider exactly the same quantum communication task (but with the new ensemble) as defined in the previous section. The mutual information between Alice and Bob needs to be at least $\Omega(n^2)$ if a learning algorithm exists. 
    Note that, we have shown both the pairwise trace distance and the pairwise Wigner TV distance are at least $\Omega(\varepsilon)$, thus the above holds true both for the setting of Theorem~\ref{th:non-gaussian-lower-stronger} and Theorem~\ref{th:wigner}.
    Now we upper bound the Holevo quantity. 
    Rewrite $\Sigma_a$ as
    \bb\label{eq:tighter-random-eq1}
        \Sigma_a=(I_2\otimes V_a)S^{-1}S^{-T}(I_2\otimes V_a^T).
    \ee
    Here $V_a$ is an orthogonal $n\times n$ matrix that contains $U_a$ as the first $s\coleq \lceil n/9\rceil$ columns. $S$ is a symmetric and symplectic matrix defined by
    \bb
    S^{-1}\coloneqq\begin{pmatrix}\sqrt{1-\frac{\varepsilon}{\sqrt{n}+\varepsilon}}& 0\\ 0 &\sqrt{1+\frac{\varepsilon}{\sqrt{n}}}\end{pmatrix}\otimes \Pi_{s}+I_2\otimes(I- \Pi_{s}),
    \ee
    where $\Pi_{s}$ is a projector on the first $s$ coordinates. Eq.~\eqref{eq:tighter-random-eq1} is in the Williamson form, thus the Hamiltonian matrix is given by
    \bb
        H_a &= (I_2\otimes V_a)S^{T}f(I_{2n})S(I_2\otimes V_a^T)
        \\&= \frac{\log3}2 \begin{pmatrix}
      I_n + \frac\varepsilon{\sqrt n} U_aU_a^T&0_n
      \\ 0_n& I_n - \frac{\varepsilon}{\sqrt n + \varepsilon}U_aU_a^T \end{pmatrix}.
    \ee
    Now use Lemma~\ref{le:marco_holevo_bound}. We have
    \bb
        \chi(\{p_a,\sigma_a^{\otimes N}\})&\le N \max_{a\neq a'}\Tr[(\Sigma_a-\Sigma_{a'})(H_{a'}-H_a)]\\
        &= N\frac{\varepsilon^2\log3}{n+\varepsilon\sqrt n}\max_{a\neq a'}\Tr\left[(U_aU_a^T-U_{a'}U_{a'}^T)^2\right]\\
        &\le N\frac{\varepsilon^2\log3}{n+\varepsilon\sqrt n}\cdot 2s\\
        &= N\cdot O(\varepsilon^2),
    \ee
    which upper bounds the mutual information thanks to the Holevo theorem~\cite{holevo1973bounds}. Putting everything together, we obtain
    \bb
        N = \Omega(n^2/\varepsilon^2).
    \ee
    This completes our proof.
\end{proof}

\subsubsection{An alternative improved lower bound for passive Gaussian states}~\label{sec:any-lower-stronger-passive}
In this section, we provide an alternative proof for the $\Omega(n^2/\varepsilon^2)$ lower bound. The advantage is that this lower bound works for passive Gaussian states. 

\begin{theorem}\label{th:non-gaussian-passive-lower}
  For any algorithm that learns an unknown zero-mean \textbf{passive} $n$-mode Gaussian state $\rho(0,\Sigma)$ to trace distance precision $\varepsilon$ with probability at least $2/3$ using any POVM measurements, the number of copies must satisfy $N = \Omega({{n^2}/{\varepsilon^{2}}})$.
\end{theorem}
\begin{proof}
  In the following, we assume $\varepsilon$ is upper bounded by a sufficiently small constant.
  Consider the passive Gaussian ensemble $\{\rho_a\}_{a=1}^M$ defined in the proof of Theorem~\ref{th:gaussian-lower} with a small twist:
  \bb
    \Sigma_a\coleq \frac{1+n^{-1}}{2}I_{2n} + \frac{\varepsilon}{2n}\begin{pmatrix}
      U_aU_a^T &0_n \\ 0_n& U_aU_a^T,
    \end{pmatrix}
  \ee
  where $\{U_a\}_{a=1}^M$ with $M=2^{\Theta(n^2)}$ is the same orthogonal matrix ensemble used throughout this paper. 
  The additional $n^{-1}$ term ensures the Gaussian states are non-degenerate. 
  We first show this modified ensemble still have pairwise trace distance $\Omega(\varepsilon)$: For any two distinct $a,b\in[M]$, consider the following measurement: First apply $V_a^\dagger\oplus V_a^\dagger$ (where $V_a \coleq [U_a;W_a]$ is any orthogonal matrix containing $U_a$ as the first $s$ columns), trace out the first $s$ modes, and measure the number of photon in the last $n-s$ modes.
  For $a$, the marginal covariance matrix on the last $n-s$ modes after applying $V_a^\dagger$ is simply $\frac{1+n^{-1}}{2}I_{2(n-s)}$. The probability of obtaining vacuum is thus
  \bb
    \Pr(\mr{Vac}|a) = \mr{det}^{-\frac12}\!\left(\left(1+\frac{1}{2n}\right)I_{2(n-s)}\right) = \left(1 +\frac1{2n}\right)^{-(n-s)}.
  \ee
  For $b$, the marginal covariance matrix on the last $n-s$ modes after applying $V_a^\dagger$ is given by
  \bb
    \frac{1+n^{-1}}{2}I_{2(n-s)} + \frac{\varepsilon}{2n}(W_a^T U_b U_b^T W_a)^{\oplus 2}.
  \ee
  The probability of obtaining vacuum is thus
  \bb
    \Pr(\mr{Vac}|b) &= \mr{det}^{-\frac12}\!\left(\left(1+\frac{1}{2n}\right)I_{2(n-s)}+ \frac{\varepsilon}{2n}(W_a^T U_b U_b^T W_a)^{\oplus 2}\right) 
    \\&= \left(1 +\frac1{2n}\right)^{-(n-s)}\mr{det}^{-1}\!\left(I + \frac{\varepsilon}{2n+1}W_a^T U_b U_b^T W_a\right)
    \\&\le \left(1 +\frac1{2n}\right)^{-(n-s)} \left(1 + \frac{\varepsilon}{2n+1}\Tr(W_a^T U_b U_b^T W_a) \right)^{-1}
    \\&\le \left(1 +\frac1{2n}\right)^{-(n-s)} \left(1 - \frac{\varepsilon}{8n}\Tr(W_a^T U_b U_b^T W_a) \right)
    \\&= \left(1 +\frac1{2n}\right)^{-(n-s)} \left(1 - \frac{\varepsilon}{8n}(s - \|U_b^TU_a\|^2_{\mr F})) \right)
    \\&\le \left(1 +\frac1{2n}\right)^{-(n-s)} \left(1 - \varepsilon/144 \right)
  \ee 
  Therefore,
  \bb
    D_\mr{tr}(\rho_a,\rho_b) \ge |\Pr(\mr{Vac}|b)-\Pr(\mr{Vac}|a)|\ge \left(1 +\frac1{2n}\right)^{-(n-s)} \frac{\varepsilon}{144} = \Omega(\varepsilon).
  \ee
  The last equality can be seen by
  \bb
    \left(1 +\frac1{2n}\right)^{-(n-s)} =\exp\!\left(-(n-s)\log(1+\frac1{2n})\right) \ge \exp(-1/2).
  \ee
  This completes the proof that $\{\rho_a\}$ is an $\varepsilon$-packing net in trace distance.

  \medskip
  \noindent  Now consider the uniformly-distributed $N$-copy ensemble $\{\frac1M,\rho_a^{\otimes N}\}$. By Lemma~\ref{le:marco_holevo_bound}, the Holevo $\chi$-quantity is upper bounded by
  \bb
    \chi(\{\frac1M,\rho_a^{\otimes N}\}) &\le N\! \max_{a\neq b\in[M]}\Tr[(\Sigma_a-\Sigma_b)(H_b-H_a)],
  \ee
  where $H_a$ is the Hamiltonian of $\rho_a$ as defined in the previous section. 
  Note that 
  \bb
    \Sigma_a &= \left(V_a \begin{pmatrix}
      (\frac12 + \frac{1+\varepsilon}{2n})\Pi_{\le s} & 0\\0 & (\frac12 + \frac{1}{2n})\Pi_{> s}
    \end{pmatrix}V_a^T\right)^{\oplus 2}
    \\&= \left(\frac{1+n^{-1}}{2}I + \frac{\varepsilon}{2n}U_aU_a^T\right)^{\oplus 2}
  \ee
  Thus,
  \bb
      H_a &= \left(V_a \begin{pmatrix}
      f(\frac12 + \frac{1+\varepsilon}{2n})\Pi_{\le s} & 0\\0 & f(\frac12 + \frac{1}{2n})\Pi_{> s}
    \end{pmatrix}V_a^T\right)^{\oplus 2}
    \\&=\left(V_a \begin{pmatrix}
      \frac12\log(1+\frac{2n}{\varepsilon+1})\Pi_{\le s} & 0\\0 & \frac12\log(1+2n)\Pi_{> s}
    \end{pmatrix}V_a^T\right)^{\oplus 2}
    \\&= \left(\frac12\log(1+2n)I + \frac12\log(\frac{1+\frac{2n}{\varepsilon+1}}{1+2n})U_aU_a^T\right)^{\oplus 2}
  \ee
  Therefore,
  \bb
    \Tr[(\Sigma_a-\Sigma_b)(H_b-H_a)] &= \frac{\varepsilon}{2n}\log(1 +\frac{2n\varepsilon}{1+\varepsilon+2n})\Tr[(U_aU_a^T - U_bU_b^T)^2]
    \\&\le \frac{\varepsilon}{2n}\cdot \varepsilon\cdot 2s
    \\&\le \varepsilon^2/9.
  \ee
  This implies an upper bound of $N\cdot O(\varepsilon^2)$ on the Holevo $\chi$ quantity. Combining this with Fano's argument and the Holevo theorem, we conclude that $N = \Omega(n^2/\varepsilon^2)$ copies are necessary to complete the learning task with $2/3$ success probability using any measurements.
\end{proof}

\subsection{Tight lower bound for heterodyne measurements}\label{sec:heterodyne}

\begin{theorem}\label{th:heterodyne-lower}
  For any algorithm that learns an unknown zero-mean $n$-mode Gaussian state $\rho(0,\Sigma)$, with the promise that $\|\Sigma\|_\mr{op}\le \bar{E}$ for some $\bar E>1$, to trace distance precision $\varepsilon<0.005$ with probability at least $2/3$ using only heterodyne measurements, the number of copies must satisfy $N = \Omega\left({{\bar{E}^2n^3}/{\varepsilon^{2}}}\right)$. 
\end{theorem}

Since $\|\Sigma\|_\mr{op}\le\bar E$ implies $\Tr\Sigma\le2n\bar E$, the heterodyne upper bound derived in~\cite{bittel2025energy} thus gives $N=O(\bar E^2n^3/\varepsilon^2)$, matching this lower bound.

\begin{proof}
  Consider the following Gaussian states ensemble:
  \bb
    \Sigma_a\coleq \frac12\begin{pmatrix}
      \lambda I_n & 0\\
      0 & \lambda^{-1} \left(I_n + \frac{\varepsilon}{n}U_aU_a^T\right)
    \end{pmatrix}
  \ee
  Here $\varepsilon\le 1$ and $\{U_a\}_{a=1}^M$ is the same ensemble of $n$-by-$\ceil{n/9}$ column-orthogonal matrices of size $M=2^{\Omega(n^2)}$ such that $\|U_a^T U_b\|_\mr{F}^2\le n/18$ for all distinct pairs $a,b\in[M]$ from Lemma~\ref{le:ashtiani6-4}. Denote $s=\ceil{n/9}$. We define $\rho_a$ to be the zero-mean $n$-mode Gaussian state with covariance matrix $\Sigma_a$. 
  Note that by choosing $\lambda = 2\bar E>2$, we have $\|\Sigma_a\|_\mr{op} = \bar E$.

  We first show the trace distance between any distinct pair of states within the ensemble is $\Omega(\varepsilon)$. Thanks to the unitary invariance of trace distance, we can consider the following unsqueezed states instead:
  \bb
    \Sigma'_a\coleq
    \frac12\begin{pmatrix}
      I_n & 0\\
      0 & I_n + \frac{\varepsilon}{n}U_aU_a^T
    \end{pmatrix}.
  \ee
  Using the same procedure of orthogonal symplectic transformation followed by Fock basis measurement as in the proof of Theorem~\ref{th:gaussian-lower}, we can show that 
  \bb
    D_\mr{tr}(\rho_a,\rho_b) &\ge 1 - \left(\det\left(I + \frac{\varepsilon}{2n}W_a^TU_bU_b^TW_a\right)\right)^{-1/2} \\
    &\ge \frac{\varepsilon}{5n}\Tr(W_a^TU_bU_b^TW_a)\\
    &\ge \frac\varepsilon{90}.
  \ee
  Here we use $(1+x)^{-1/2} \le 1 - \frac{2}{5}x$ for $0\le x\le 1/3$ and the property of our state ensemble from Lemma~\ref{le:ashtiani6-4}.

  Next, consider the communication task that Alice uniformly randomly pick $a\in[M]$, prepares $N$ i.i.d. copies of $\rho_a$, applies heterodyne measurement on each copy to get classical outcomes $x_{1:N}$, and sends them to Bob. Bob then tries to guess Alice's choice of $a$ by processing the classical samples he received. Suppose there exists a scheme using $N$ copies to learn any $\rho$ to trace distance precision $\varepsilon'=\varepsilon/180$ with probability at least $2/3$. Then, Bob can use this scheme to learn $\rho_a$ and hence guess $a$ correctly with probability at least $2/3$. Using Fano's inequality, this implies 
  \bb
    I(x_{1:N}:a) \ge \frac23 \log M - \log 2 = \Omega(n^2).
  \ee
  On the other hand, heterodyne measurements on $N$ copies of $\rho_a$ produce $N$ i.i.d. samples from the classical Gaussian distribution $\widetilde p_a \coleq \mc N(0,\Sigma_a + \frac12 I_{2n})$. The pairwise KL divergence can be upper bounded as
  \bb
    \mr{KL}(\widetilde p_a \| \widetilde p_b) &\eqt{} \mr{KL}\left(\mc N\!\left(0, \frac{1+\lambda^{-1}}2I +\frac{\lambda^{-1}\varepsilon}{2n}U_aU_a^T\right)~\|~ \mc N\!\left(0, \frac{1+\lambda^{-1}}2I +\frac{\lambda^{-1}\varepsilon}{2n}U_bU_b^T\right) \right) \\
    &\eqt{(i)} \mr{KL}\left(\mc N\!\left(0, \frac{1}2I +\frac{\varepsilon}{2n(1+\lambda)}U_aU_a^T\right)~\|~ \mc N\!\left(0, \frac{1}2I +\frac{\varepsilon}{2n(1+\lambda)}U_bU_b^T\right) \right)\\
    &\eqt{(ii)} O\left(\frac{\varepsilon^2}{n\lambda^2}\right).
  \ee
  Here, (i) uses the fact that KL divergence is invariant under rescaling; (ii) follows from the same calculation as Eq.~\eqref{eq:KL-calculation} with $\varepsilon$ replaced by $\varepsilon/(1+\lambda)$. Therefore, the mutual information can be upper bounded by $I(x_{1:N}:a) \le N \cdot O\left({\varepsilon^2}/{(n\lambda^2)}\right)$. Putting everything together, we conclude that
  \bb
    N=\Omega\left(\frac{\lambda^2n^3}{\varepsilon^2}\right).
  \ee
  By substituting $\lambda = 2\bar E$, we complete the proof.
\end{proof}

\section{Upper bounds on Gaussian state learning}

\subsection{Nearly-optimal upper bound for pure Gaussian states}\label{sec:pure-upper}

\begin{theorem}\label{th:pure-upper}
  There exists a protocol using adaptive general-dyne measurements that learns an unknown $n$-mode pure Gaussian state $\rho(\mu,\Sigma)$, with the promise that {$\|\Sigma\|_\mr{op}\le E$}, %
  to trace distance precision $\varepsilon$ with probability at least $2/3$ using 
  {$O\!\left(n^2/\varepsilon^2 + (n +\log\log\log  E )\log\log  E\right)$ }%
  copies and runs in $\mr{poly}(n,\varepsilon^{-1},\log\log E)$ time.
\end{theorem}
\noindent Up to logarithmic factors in $n,\varepsilon$ and a double-logarithmic energy term, this matches the pure Gaussian state learning lower bound of $\widetilde\Omega(n^2/\varepsilon^2)$ from Theorem~\ref{th:non-gaussian-lower-stronger} that holds for any POVM measurements.
This means non-Gaussian measurements provide no (super-logarithmic) advantage for learning pure Gaussian states.

\medskip

To prove Theorem~\ref{th:pure-upper}, we need the adaptive learn-to-unsqueeze subroutine of~\cite{bittel2025energy} as reviewed in Lemma~\ref{le:unsqueeze}.
Another lemma we will use is Lemma~\ref{le:pure-trace-tvd} which says the trace distance between pure Gaussian states is upper bounded by their Wigner TV distance up to an absolute constant.

\begin{proof}[Proof of Theorem~\ref{th:pure-upper}]
  Our protocol contains the following steps:
  \begin{enumerate}
    \item Run the learn-to-unsqueeze protocol~\cite{bittel2025energy} in Lemma~\ref{le:unsqueeze} using {\bb\label{eq:copies_un}
    N=O\!\left(\left(n + \log \!\left(\frac{\log\log\|\Sigma\|_\mr{op}}{\delta}\right)\right)\log\log\|\Sigma\|_\mr{op}\right)
    \ee}%
    copies of $\rho$ to learn a symplectic matrix $S$ such that {$\rho_\mr{unsqueezed}(\mu,\Sigma')\coleq U_{S^{-1}}\rho(\mu,\Sigma) U^\dagger_{S^{-1}}$ }%
    satisfies $\|{\Sigma'}^{-1}\|_\mr{op}\le 4$ with probability at least $1-\delta/2$. {In Eq.~\eqref{eq:copies_un} we used Lemma~\ref{le:unsqueeze} together with the fact that the covariance matrix $\Sigma$ of a pure Gaussian state satisfies \bb\|\Sigma^{-1}\|_\mr{op}=4\|\Sigma\|_\mr{op}\,.
        \ee}In the following steps we replace all copies of $\rho$ by $\rho_\mr{unsqueezed}$. For notational simplicity, we relabeled the unsqueezed state by $\rho(\mu,\Sigma)$.
    \item Run heterodyne measurements on $2m = O((n^2+n\log\delta^{-1})/\varepsilon^2)$ copies of $\rho$ to get outcomes $v_1,\cdots,v_{2m} \in \mbb R^{2n}$. Define the estimators,
    \bb
      \hat\mu \coleq \frac{1}{m}\sum_{i=1}^m v_i,\quad \hat\Sigma_0\coleq \frac{1}{2m}\sum_{i=1}^m (v_{2i}-v_{2i-1})(v_{2i}-v_{2i-1})^T - \frac12 I_{2n}.
    \ee
    \item Project $\hat\Sigma_0$ to a pure Gaussian covariance matrix $\hat\Sigma$ by the following precedure: 
    write down the Williamson decomposition $\hat\Sigma_0 = \hat S \hat D \hat S^T$ and then define $\hat\Sigma\coleq \frac12\hat S\hat S^T$. Return $\hat\rho(\hat\mu,\hat\Sigma)$ as the final estimate of the unknown state (after unsqueezing).
  \end{enumerate}
  By taking $\delta=1/3$, the total number of copies used is $N = O(n\log\log E + n^2/\varepsilon^2)$ as desired. The run time is obviously polynomial in $n,\varepsilon^{-1},\log E$.

  Conditioned on Step 1 succeeding, we have $\|\Sigma^{-1}\|_\mr{op}\le 4$.   
  Since $\Sigma$ is a pure Gaussian covariance matrix, its eigenvalues take the form $\{\lambda_1,\cdots,\lambda_n,(4\lambda_1)^{-1},\cdots,(4\lambda_n)^{-1}\}$. We thus have $\lambda_i\in[1/4,1]$.
  Following the concentration analysis in \cite[Appendix C]{ashtiani2020near}, the estimators defined in Step 2 with probability at least $1-\delta/2$ satisfy
  \begin{gather}
    (\hat\mu - \mu)^T (\Sigma +\frac12 I_{2n})^{-1} (\hat\mu - \mu) \le \varepsilon^2/2,\label{eq:random1-1}\\
    \left\|(\Sigma+\frac12 I_{2n})^{-1/2}(\hat\Sigma_0 +\frac12 I_{2n})(\Sigma+\frac12 I_{2n})^{-1/2} - I_{2n}\right\|_\mr{op} \le \varepsilon/{\sqrt {4n}}.\label{eq:random1-2}
  \end{gather}
  Note that
  \bb
    (\Sigma+\frac12 I_{2n})^{-1} &= \sum_{i=1}^{2n}(\lambda_i + \frac12)^{-1}\ketbra{e_i}{e_i}
    \ge \frac13 \sum_{i=1}^{2n} \lambda_i^{-1}\ketbra{e_i}{e_i} 
    = \frac13 \Sigma^{-1}.
  \ee
  The inequality uses $\lambda_i^{-1}\le 4$. Thus, Eq.~\eqref{eq:random1-1} implies,
  \begin{gather}\label{eq:mean-bound}
    (\hat\mu-\mu)^T\Sigma^{-1}(\hat\mu-\mu)\le \frac32\varepsilon^2\quad\Leftrightarrow\quad \|\Sigma^{-1/2}(\hat\mu-\mu)\|_2 \le \sqrt{\frac32}\varepsilon.
  \end{gather}
  On the other hand, Eq.~\eqref{eq:random1-2} implies,
  \bb
    \|\hat\Sigma_0 - \Sigma\|_\mr{op} = \left\|(\hat\Sigma_0+\frac12 I) - (\Sigma+\frac12 I)\right\|_\mr{op}
    \le \frac{\varepsilon}{\sqrt{4n}}\cdot\left\|\Sigma+\frac12 I\right\|_\mr{op}\le \frac{3\varepsilon}{4\sqrt n}.
  \ee
  Denote $\xi \coleq \frac{3\varepsilon}{4\sqrt n}$. Next, we analyze the error in Step 3,  
  \bb
    \|\hat\Sigma - \hat\Sigma_0\|_\mr{op} &= \left\|\hat S (\hat D-\frac12 I) \hat S^T\right\|_\mr{op}
    \\&\leqt{} \|\hat D - \frac12 I_{2n}\|_\mr{op}\|\hat S\|^2_\mr{op}.
  \ee
  For the first factor\footnote{
  An alternative way to bound this factor is to apply the perturbation bound for symplectic eigenvalues~\cite[Theorem 3.1]{idel2017perturbation}.
  }, we make use of the following known result: for any real positive-definite matrix $V$, there exists a positive-definite matrix $A(V)\coleq -\Omega V\Omega V$ whose eigenvalues are the squares of $V$'s symplectic eigenvalues, each appears twice. Define $\Delta\coleq \hat\Sigma_0-\Sigma$. Then,
  \begin{gather}
    A(\hat\Sigma_0) - A(\Sigma) = -\Omega\Delta\Omega\Sigma -\Omega\Sigma\Omega\Delta-\Omega\Delta\Omega\Delta.\\
    \Longrightarrow\quad \|A(\hat\Sigma_0) - A(\Sigma)\|_\mr{op} \le 2\|\Sigma\|_\mr{op}\|\Delta\|_\mr{op} + \|\Delta\|_\mr{op}^2\le 3\xi.
  \end{gather}
  Using $\lambda_i(X)$ to denote the $i$-th largest eigenvalue of a matrix $X\ge0$, Weyl's inequality implies
  \bb
    |\lambda_i(A(\hat\Sigma_0)) - \lambda_i(A(\Sigma))| \le \|A(\hat\Sigma_0) - A(\Sigma)\|_\mr{op} \le 3\xi.
  \ee
  Further note that $\Sigma$ is a pure Gaussian covariance matrix, so its symplectic eigenvalues are all $1/2$. the above becomes $\max_i |\hat D_{ii}^2 - 1/4 |\le 3\xi$, which then implies 
  \bb
    \|\hat D - \frac12 I_{2n}\|_\mr{op} 
    =\max_{i\in[n]} |\hat D_{ii} - 1/2| = \max_{i\in[n]} \frac{|\hat D_{ii}^2 - 1/4|}{\hat D_{ii}+1/2}\le 6\xi.
  \ee
  For the second factor, 
  \bb
    \|\hat S\|^2_\mr{op} &= \|\hat S\hat S^T\|_\mr{op}\le \frac{\|\hat S\hat D\hat S^T\|_\mr{op}}{\min_i\hat D_{ii}} \le \frac{1 + \xi}{1/2 - 6\xi}\le 4.
  \ee
  The second inequality uses our bounds on $\hat D$ and $\hat\Sigma_0$. The last inequality holds given that $\xi\le 1/25$. Putting things together and using the triangle inequality, we have
  \bb
    \|\hat\Sigma - \Sigma\|_\mr{op} &\le \|\hat\Sigma - \hat\Sigma_0\|_\mr{op} + \|\hat\Sigma_0 - \Sigma\|_\mr{op}\\
    &\le 25\xi = \frac{75\varepsilon}{4\sqrt n}.
  \ee
  which implies 
  \bb\label{eq:cov-bound}
  \|\Sigma^{-1/2}\hat\Sigma\Sigma^{-1/2}-I_{2n}\|_\mr{F}&\le \sqrt{2n}\cdot\|\Sigma^{-1/2}\hat\Sigma\Sigma^{-1/2}-I_{2n}\|_\mr{op}\\
  &\le \sqrt{2n}\|\Sigma^{-1}\|_\mr{op}\|\hat\Sigma - \Sigma\|_\mr{op}
  \\&\le 75\sqrt2\varepsilon.
  \ee
  Combining Eq.~\eqref{eq:mean-bound} and Eq.~\eqref{eq:cov-bound}, Lemma~\ref{le:mahalanobis} yields that
  \bb
    \mr{TV}(\mc N(\mu,\Sigma),\mc N(\hat\mu,\hat\Sigma)) \le 75\varepsilon,
  \ee
  And finally Lemma~\ref{le:pure-trace-tvd} implies that
  \bb
    D_\mr{tr}(\rho(\mu,\Sigma),\rho(\hat\mu,\hat\Sigma)) \le 11250 \varepsilon.
  \ee
  Rescaling $\varepsilon$ by a constant and choosing $\delta=1/3$, we complete the proof of Theorem~\ref{th:pure-upper}.
\end{proof}

\subsection{Non-Gaussian advantages in learning passive Gaussian states}\label{sec:passive-upper}
By definition, a passive Gaussian state is any state that can be prepared by applying passive Gaussian unitaries (i.e.~Gaussian unitaries with no squeezing, also known as energy-preserving Gaussian unitaries) to a multimode thermal state. Equivalently, a passive Gaussian state is a Gaussian state with a vanishing first moment and a covariance matrix $\Sigma$ that admits a Williamson decomposition 
\bb\label{eq:will_passive}
    \Sigma=ODO^\intercal
\ee
via a symplectic matrix $O$ that is also orthogonal.

{The following theorem shows that the optimal sample complexity for learning \(n\)-mode passive Gaussian states using \emph{Gaussian} protocols is \(\Theta(n^3/\varepsilon^2)\), with no dependence at all on the energy. Moreover, the optimal measurement strategy is heterodyne detection.

\begin{theorem}\label{th:sample_passive}
    The sample complexity of learning \(n\)-mode passive Gaussian states via Gaussian operations is \(\Theta(n^3/\varepsilon^2)\). In particular, heterodyne tomography performed on \(O(n^3/\varepsilon^2)\) copies of an unknown passive Gaussian state learns the state to trace-distance accuracy \(\varepsilon\) with probability at least \(2/3\). Conversely, any protocol restricted to Gaussian operations requires \(\Omega(n^3/\varepsilon^2)\) copies.
\end{theorem}

\begin{proof}
    The lower bound \(\Omega(n^3/\varepsilon^2)\) follows immediately from Theorem~\ref{th:gaussian-lower}.
    
    The upper bound \(O(n^3/\varepsilon^2)\) follows from \cite[Theorem~10]{bittel2025energy}. Indeed, \cite[Theorem~10]{bittel2025energy} implies that heterodyne measurements on
    \bb\label{eq_N_suff_passive}
        N=O\left( \frac{(n+\Tr\Sigma^{-1})\,n^2}{\varepsilon^2} \right)
    \ee
    copies of an unknown Gaussian state \(\rho(\mu,\Sigma)\) are sufficient to learn the state to trace-distance precision \(\varepsilon\) with probability at least \(2/3\). For passive Gaussian states, we have
    \bb
        \Sigma^{-1}=O D^{-1} O^\intercal \le 2 I_{2n}\,,
    \ee
    where the equality uses the Williamson decomposition for passive Gaussian states in Eq.~\ref{eq:will_passive}, and the inequality follows from \(D\ge \tfrac12 I_{2n}\). Therefore,
    \(\Tr\Sigma^{-1}\le 4n\), and substituting this into Eq.~\eqref{eq_N_suff_passive} yields \(N=O(n^3/\varepsilon^2)\). Hence, \(O(n^3/\varepsilon^2)\) copies suffice for tomography of passive Gaussian states, completing the proof.
\end{proof}

In summary, the theorem above identifies \(\Theta(n^3/\varepsilon^2)\) as the optimal sample complexity achievable by \emph{Gaussian} protocols for learning unknown passive Gaussian states. This naturally raises the question: can one surpass this fundamental bound by allowing \emph{non-Gaussian} strategies? In the following, we show that this is indeed possible, and the key ingredient is the recently introduced \emph{passive random purification channel}~\cite{walter_purif_2025,mele_purif_2025}.

}

In Refs.~\cite{walter_purif_2025,mele_purif_2025}, it has been shown that, for any natural numbers~$N$ and~$n$, one can convert $N$ copies of an $n$-mode \emph{mixed} passive Gaussian state~$\rho$ into $N$ copies of a \emph{randomly} chosen $2n$-mode Gaussian purification. More precisely, let $\mathcal{H}_A$ and $\mathcal{H}_P$ be two $n$-mode systems. There exists a (non-Gaussian) channel, dubbed the \emph{passive random purification channel},
\bb\label{eq_random_purif}
\Lambda^{(N)}:\mathcal{H}_A^{\otimes N}\longrightarrow \left(\mathcal{H}_A\otimes \mathcal{H}_P\right)^{\otimes N},
\ee
such that, for any $n$-mode passive Gaussian state~$\rho_A$, one has
\bb
\Lambda^{(N)}\!\left(\rho_A^{\otimes N}\right)
=
\mathbb{E}_{O}\!\left[
\left(I_A\otimes U^{(P)}_O\right)
\ketbra{\psi_\rho}{\psi_\rho}_{AP}
\left(I_A\otimes U^{(P)}_O\right)^\dagger
\right]^{\otimes N},
\ee
where 
\bb\label{eq_pur}
\ket{\psi_\rho}_{AP}\coloneqq \sqrt{\rho_A}\otimes I_P\ket{\Gamma}_{AP}
\ee
is the standard purification of~$\rho_A$ (which is Gaussian~\cite{walter_purif_2025,mele_purif_2025}), $\ket{\Gamma}_{AP}$ is the unnormalised maximally entangled state with respect to the Fock basis, the expectation value is taken over uniformly random orthogonal symplectic matrices
$O\in \mathrm{Sp}(2n)\cap \mathrm{O}(2n)$, and $U_O$ denotes the (passive) Gaussian unitary associated with~$O$. In other words, the channel $\Lambda^{(N)}$ outputs $N$ copies of the Gaussian purification
$\left(I_A\otimes U^{(P)}_O\right)\ket{\psi_\rho}_{AP}$, with $O$ chosen uniformly at random.

Moreover, as observed in~\cite{mele_purif_2025}, each state
$\left(I_A\otimes U^{(P)}_O\right)\ket{\psi_\rho}_{AP}$
has mean photon number exactly twice that of $\rho_A$. Equivalently, since both states have vanishing first moments, letting $\Sigma$ and $\Sigma_{\mathrm{pure}}$ denote the covariance matrices of $\rho_A$ and of its purification $\left(I_A\otimes U^{(P)}_O\right)\ket{\psi_\rho}_{AP}$, respectively, one has
\bb\label{eq_cond_trace}
\Tr \Sigma_{\mathrm{pure}} = 2\,\Tr \Sigma\,.
\ee{We now bound the operator norm of $\Sigma_{\mathrm{pure}}$ in terms of that of $\Sigma$. This is established in the following lemma.

\begin{lemma}
Let $\Sigma$ and $\Sigma_{\mathrm{pure}}$ denote the covariance matrices of $\rho_A$ and of its purification $\left(I_A\otimes U^{(P)}_O\right)\ket{\psi_\rho}_{AP}$, respectively, where $\ket{\psi_\rho}_{AP}$ is the standard purification of $\rho_A$ defined in~\eqref{eq_pur} and $U^{(P)}_O$ is a passive Gaussian unitary acting on the purifying system. Then,
\[
\|\Sigma_{\mathrm{pure}}\|_{\mathrm{op}} \le 4\,\|\Sigma\|_{\mathrm{op}}.
\]
\end{lemma}

\begin{proof}
Since the operator norm is invariant under orthogonal transformations, it suffices to consider the case $U^{(P)}_O=I_P$.

We first compute the covariance matrix $\Sigma_{\mathrm{pure}}$ of the standard purification $\ket{\psi_\rho}_{AP}$. Note that the reduced state on the purifying system $P$ is the transposed state $\rho_A^{\intercal}$; we denote its covariance matrix by $\Sigma_{\mathrm{tr}}$. Consequently, the covariance matrix of the purification has the block form
\[
\Sigma_{\mathrm{pure}}
=
\begin{pmatrix}
\Sigma & X \\
X^\intercal & \Sigma_{\mathrm{tr}}
\end{pmatrix},
\]
for some real matrix $X$. Since $\Sigma_{\mathrm{tr}}$ and $\Sigma_{\mathrm{pure}}$ are positive definite, Schur's complement theorem~\cite[Theorem~1.3.3]{BHATIA} implies that
\bb
\Sigma \ge X\,\Sigma_{\mathrm{tr}}^{-1} X^\intercal .
\ee
In particular, this yields
\bb
\Sigma \ge \frac{X X^\intercal}{\|\Sigma_{\mathrm{tr}}\|_{\mathrm{op}}},
\ee
and hence
\bb\label{eq_X}
\|X\|_{\mathrm{op}}^2 \le \|\Sigma\|_{\mathrm{op}} \,\|\Sigma_{\mathrm{tr}}\|_{\mathrm{op}} .
\ee
Next, recall that the covariance matrix of the transposed state satisfies~\cite{serafini2023quantum}
\[
\Sigma_{\mathrm{tr}} = (I_n \otimes Z)\,\Sigma\,(I_n \otimes Z),
\]
where $Z=\mathrm{diag}(1,-1)$ flips the sign of all momentum quadratures. Since $I_n\otimes Z$ is orthogonal, it follows that
\[
\|\Sigma_{\mathrm{tr}}\|_{\mathrm{op}} = \|\Sigma\|_{\mathrm{op}} .
\]
Substituting this into~\eqref{eq_X} gives $\|X\|_{\mathrm{op}} \le \|\Sigma\|_{\mathrm{op}}$. Finally, using the triangle inequality for the operator norm, we obtain
\bb
\|\Sigma_{\mathrm{pure}}\|_{\mathrm{op}}
=
\left\|
\begin{pmatrix}
\Sigma & X \\
X^\intercal & \Sigma_{\mathrm{tr}}
\end{pmatrix}
\right\|_{\mathrm{op}}
\le
\|\Sigma\|_{\mathrm{op}} + \|X\|_{\mathrm{op}} + \|X^\intercal\|_{\mathrm{op}} + \|\Sigma_{\mathrm{tr}}\|_{\mathrm{op}}
\le
4\,\|\Sigma\|_{\mathrm{op}} .
\ee
\end{proof}

}
In summary, there is a procedure that transforms $N$ copies of an unknown $n$-mode mixed passive Gaussian state satisfying $\|\Sigma\|_\mathrm{op} \le E$ into $N$ copies of an unknown Gaussian purification satisfying $\|\Sigma\|_\mathrm{op} \le 4E$. This observation, together with Theorem~\ref{th:pure-upper}, naturally leads to the following tomography algorithm for learning mixed passive Gaussian states with $\|\Sigma\|_\mathrm{op} \le E$ up to precision $\varepsilon$ in trace distance with probability at least $2/3$:
\begin{enumerate}
    \item Given $N = O\!\left(n^2/\varepsilon^2 + (n +\log\log\log  E )\log\log  E\right)$ copies of the unknown passive Gaussian state, apply the passive random purification channel in \eqref{eq_random_purif}.
    \item Apply the pure Gaussian-state tomography algorithm of Theorem~\ref{th:pure-upper} to the output of the passive random purification channel.
    \item Output the reduced state of the estimated pure state.
\end{enumerate}
The correctness of the above protocol is proved in the following theorem.
\begin{theorem}[Non-Gaussian algorithm for tomography of passive Gaussian states]\label{th:pure-upper_passive}
There exists a (non-Gaussian) tomography algorithm that learns an unknown $n$-mode passive Gaussian state, with the promise that its covariance matrix satisfies $\|\Sigma\|_\mathrm{op} \le E$, to trace distance precision $\varepsilon$ with probability at least $2/3$ using
$N = O\!\left(n^2/\varepsilon^2 + (n +\log\log\log  E )\log\log  E\right)$ copies.
\end{theorem}

\begin{proof}
We analyze the protocol step-by-step.
\begin{enumerate}
\item Apply the passive random purification channel \eqref{eq_random_purif} to $N$ copies of the unknown passive Gaussian state $\rho_A$. By \eqref{eq_cond_trace}, this produces $N$ copies of an unknown Gaussian purification $\ket{\Psi}_{AP}$ such that $\|\Sigma\|_\mathrm{op} \le 4E$. 

\item Apply the pure Gaussian-state tomography algorithm of Theorem~\ref{th:pure-upper} to these $N$ copies of $\ket{\Psi}_{AP}$. Since
\bb
O\!\left(n^2/\varepsilon^2 + (n +\log\log\log  (4E) )\log\log(4E)\right)
=O\!\left(n^2/\varepsilon^2 + (n +\log\log\log  E )\log\log  E\right)
=
N,
\ee
Theorem~\ref{th:pure-upper} guarantees that the output of the algorithm, denoted by $\ket{\widetilde{\Psi}}_{AP}$, satisfies
\bb
D_{\mathrm{tr}}\!\left(\ketbra{\Psi}{\Psi}_{AP},\ketbra{\widetilde{\Psi}}{\widetilde{\Psi}}_{AP}\right)\le \varepsilon
\ee
with probability at least $2/3$.

\item Output the reduced state $\Tr_P \ketbra{\widetilde{\Psi}}{\widetilde{\Psi}}_{AP}$. 
\end{enumerate}
Note that, by the monotonicity of the trace distance under the partial trace and since $\ket{\Psi}_{AP}$ is a purification of $\rho_A$, one has
\bb
D_{\mathrm{tr}}\!\left(\rho_A,\Tr_P \ketbra{\widetilde{\Psi}}{\widetilde{\Psi}}_{AP}\right)
\le 
D_{\mathrm{tr}}\!\left(\ketbra{\Psi}{\Psi}_{AP}, \ketbra{\widetilde{\Psi}}{\widetilde{\Psi}}_{AP}\right).
\ee
Consequently, it follows that the returned state $\Tr_P \ketbra{\widetilde{\Psi}}{\widetilde{\Psi}}_{AP}$ satisfies
\bb
D_{\mathrm{tr}}\!\left(\rho_A,\Tr_P \ketbra{\widetilde{\Psi}}{\widetilde{\Psi}}_{AP}\right)
\le \varepsilon,
\ee
with probability at least $2/3$. This concludes the proof.
\end{proof}
Thanks to Theorem~\ref{th:pure-upper_passive}, there exists a non-Gaussian tomography algorithm that learns an unknown $n$-mode passive Gaussian state to trace-distance precision $\varepsilon$ with probability at least $2/3$ using
\bb
N = \widetilde O\!\left(\frac{n^2}{\varepsilon^2}\right)
\ee
copies, where the $\widetilde O(\cdot)$ notation suppresses polylogarithmic factors in $n$, $\varepsilon^{-1}$, and the energy constraint.
On the other hand, Theorem~\ref{th:gaussian-lower} shows that any Gaussian tomography protocol requires at least
\bb
N = \Omega\!\left(\frac{n^3}{\varepsilon^2}\right)
\ee
copies to achieve the same task (and Theorem~\ref{th:sample_passive} shows that heterodyne tomography matches this lower bound). This establishes a rigorous polynomial separation between Gaussian and non-Gaussian algorithms for tomography passive Gaussian states. We summarize this result in the following quotable theorem.
\begin{theorem}[Separation between Gaussian and non-Gaussian algorithms for tomography of passive Gaussian states]\label{th:gaussian-nongaussian-separation}
Consider the task of learning an unknown $n$-mode passive Gaussian state $\rho$ with covariance matrix $\Sigma$ satisfying $\| \Sigma\|_{\mathrm{op}} \le E$, to trace-distance precision $\varepsilon$ with success probability at least $2/3$. It holds that:

\begin{enumerate}[label=(\roman*)]
\item There exists a non-Gaussian tomography algorithm that accomplishes this task using
\bb
N = \widetilde O\!\left(\frac{n^2}{\varepsilon^2}\right)
\ee
copies of $\rho$.

\item Any tomography algorithm that uses only Gaussian operations requires at least
\bb
N = \Omega\!\left(\frac{n^3}{\varepsilon^2}\right)
\ee
copies of $\rho$.
\end{enumerate}
Consequently, non-Gaussian algorithms provide a strict polynomial advantage over Gaussian algorithms for tomography of passive Gaussian states.
\end{theorem}

\section{On energy dependence and adaptivity}

\subsection{Lower Bounds for non-adaptive single-copy Gaussian measurements (1 mode)}\label{sec:non-ada-lower}

\begin{theorem}\label{th:non-ada-lower}
  For any algorithm that learns an unknown zero-mean single-mode Gaussian state $\rho(0,\Sigma)$ to trace distance precision $\varepsilon \le 0.04$ with probability at least $2/3$, given the promise that $\Tr \Sigma\le E$ for some $E>2$, using non-adaptive, single-copy Gaussian seed measurements, the number of copies must satisfy $N = \Omega\left({E/\varepsilon^{2}}\right)$.  
\end{theorem}

\noindent Combined with the adaptive learning protocol proposed in~\cite{bittel2025energy} which achieves an upper bound of $N=O\left(\log\log E\,(\log\log\log E + 1/\varepsilon^{2})\right)$, we see that adaptivity provides a double-exponential advantage in the dependence on the energy of the unknown state.
The best known non-adaptive upper bound is $O(E^2/\varepsilon^{2})$ using heterodyne measurements~\cite{bittel2025energy}, which we has proven to be tight for heterodyne. In Sec.~\ref{sec:non-ada-upper}, we show that randomized homodyne measurements can reach the $\Omega(E/\varepsilon^2)$ non-adaptive lower bound up to logarithmic factors.

We also note that the lower bound only assume a Gaussian seed measurement, which encompass any measurement that can be realized by introducing arbitrarily many ancillary gaussian states and apply a joint Gaussian measurement. The only thing it does not include is ancilla-assisted adaptive Gaussian measurement, which is an open problem to be addressed.

\begin{proof}
  The intuition is that, to learn a highly squeezed Gaussian state to small trace-distance error, one needs to estimate the small-variance (i.e., squeezed) direction with high precision. This requires using a Gaussian measurement whose seed state is also highly squeezed in the same direction. However, if the highly-squeezed direction is uniformly random, then any non-adaptive Gaussian measurement protocol cannot align to that direction with high probability, thus leading to a large energy-dependent sample complexity. To make this intuition rigorous, we use Le Cam's method. Our construction is in particular inspired by~\cite{chen2024adaptivity} which shows adaptivity can provide exponential advantage in shadow estimation. 

  Let $\varepsilon\le 1/3$. Consider the following task: A referee first samples a random angle $\phi\sim\mathrm{Unif}[0,\pi)$. Then, he chooses one of the following two covariance matrices with equal probability:
  \bb
  \Sigma_{0,\phi} = R_\phi \begin{pmatrix}
      \frac12 a & 0\\0 & \frac12 a^{-1}
  \end{pmatrix} R_\phi^T,\quad\text{or}\quad 
  \Sigma_{1,\phi} = R_\phi \begin{pmatrix}
    \frac12 a & 0\\0 & (\frac12+\varepsilon) a^{-1}
  \end{pmatrix}R_\phi^T,
  \ee
  Here, $a>1$. $R_\phi$ is the rotation matrix by angle $\phi$. Then, a player is given $N$ copies of the Gaussian state $\rho_{x,\phi}\coleq\rho(0,\Sigma_{x,\phi})$ for $x\in\{0,1\}$ reflecting the referee's choice. The player can then run any non-adaptive, single-copy Gaussian measurement protocol. After the player has completed all measurements, the referee reveals the value of $\phi$. The player is then challenged to guess the value of $x$. 

  We first show that, a protocol that satisfies the assumption of Theorem~\ref{th:non-ada-lower} can be used to solve the above task with average success probability at least $2/3$. Note that, for any $\phi$, the trace distance between the two states satisfies
  \bb
  D_{\tr}(\rho_{1,\phi},\rho_{0,\phi}) &\geqt{(i)} 1 - F(\rho_{1,\phi},\rho_{0,\phi})\\
  &\eqt{(ii)} 1 - \frac{1}{\sqrt{\det\left(\Sigma_{0,\phi}+\Sigma_{1,\phi}\right)}}\\
  &= 1 - (1+\varepsilon)^{-\frac12}\\
  &\geqt{(iii)} \frac{\varepsilon}3. 
  \ee
  Here, (i) is due to Fuchs-van de Graaf inequality and the fact that $\rho_{0,\phi}$ is pure; (ii) uses the formula of fidelity between two Gaussian states~\cite{scutaru1998fidelity,banchi2015quantum} in the special case that one state is pure and both states have the same mean; (iii) holds thanks to the assumption that $\varepsilon \le 1/3$. Thus, one can just run the learning protocol with trace distance error $\varepsilon'\le\varepsilon/7$. Then, after the referee reveals $\phi$, the player can perfectly distinguish the two cases, conditioned on the learning being successful. Thus, the average success probability is at least $2/3$.

  \medskip

  Now, for any non-adaptive, single-copy Gaussian measurement protocol on $N$ samples, we can denote the outcome probability distribution conditioned on $\phi\in[0,\pi)$ and $x\in\{0,1\}$ by $P_{x,\phi}\coleq p_{x,\phi}^{(1)}p_{x,\phi}^{(2)}\cdots p_{x,\phi}^{(N)}.$\footnote{Allowing probabilistic mixture of such measurements will not increase TV, as can be seen via the joint convexity.} The average success probability can then be upper bounded by
  \bb
  p_{\mathrm{succ}} &\le \frac12 + \frac12 \E_\phi\mr{TV}(P_{1,\phi}, P_{0,\phi}).
  \ee
  Therefore, one must have $\E_\phi\mr{TV}(P_{1,\phi}, P_{0,\phi})\ge 1/3$ to possibly have $p_{\mathrm{succ}}\ge 2/3$. Thus,
  \bb
  \frac29 &\le 2\left(\E_\phi \mr{TV}(P_{1,\phi}, P_{0,\phi})\right)^2 
  \\&\leqt{(i)} 2 \E_\phi \left(\mr{TV}(P_{1,\phi}, P_{0,\phi})\right)^2
  \\&\leqt{(ii)} \E_\phi \mr{KL}(P_{1,\phi}\|P_{0,\phi})
  \\&\eqt{(iii)} \sum_{k=1}^N \E_\phi\mr{KL}(p_{1,\phi}^{(k)}\|p_{0,\phi}^{(k)})
  \\&\leqt{(iv)} \sum_{k=1}^N \E_\phi\chi^2(p_{1,\phi}^{(k)}\|p_{0,\phi}^{(k)})
  \\&\leqt{(v)} N \sup_{V:\,V+\frac12 i\Omega\ge 0} \E_\phi\chi^2\left(\mc N(0,\Sigma_{1,\phi}+V)\|\mc N(0,\Sigma_{0,\phi}+V)\right).
  \ee
  Here, (i) is Jensen; (ii) is Pinsker; (iii) uses the chain rule of KL divergence and the non-adaptivity assumption; (iv) uses the fact that $\mr{KL}(P\|Q)\le \chi^2(P\|Q)$; (v) explicitly writes the outcome probability distribution $p_{x,\phi}^{(k)}$ using a Gaussian seed covariance matrix $V$ and maximizes over it.

  Using the formula of $\chi^2$-divergence between two Gaussian distributions, we have
  \bb
    &\E_\phi\chi^2\left(\mc N(0,\Sigma_{1,\phi}+V)\|\mc N(0,\Sigma_{0,\phi}+V)\right)
    \\\eqt{(i)}\,& \E_\phi \frac{\det(\Sigma_{0,\phi}+V)}{\sqrt{\det(\Sigma_{1,\phi}+V)\det(2\Sigma_{0,\phi}-\Sigma_{1,\phi}+V)}} -1 
    \\\eqt{(ii)}\,& \E_\phi \frac{\det\left(\begin{pmatrix}
      a&0\\0&a^{-1}
    \end{pmatrix}
      +2R_\phi^T V R_\phi\right)}{\sqrt{\det\left(\begin{pmatrix}
        a&0\\0&(1+2\varepsilon)a^{-1}
      \end{pmatrix}+2R_\phi^T V R_\phi\right)\det\left(2\begin{pmatrix}
        a&0\\0&(1-2\varepsilon)a^{-1}
      \end{pmatrix}+2R_\phi^T V R_\phi\right)}} -1 
    \\\eqt{(iii)}\,& \E_\phi \frac{\det\left(A\right)}{\sqrt{\det\left(A + 2\varepsilon a^{-1} \bm e_2\bm e_2^T\right)\det\left(A - 2\varepsilon a^{-1} \bm e_2\bm e_2^T\right)}} -1 
  \ee
  Here, for (i) to holds, it requires $2\Sigma_{0,\phi}-\Sigma_{1,\phi}+V\ge 0$, which can be easily verified; (ii) uses the rotational invariance of determinant; (iii) defines $$A\coleq \begin{pmatrix}
    a & 0\\0 & a^{-1}
  \end{pmatrix} + 2R_\phi^T V R_\phi$$ and $\bm e_2 = [0,1]^T$.
  Now use the matrix determinant lemma: $\det(A+uu^T)=(1+u^TA^{-1}u)\det(A)$,
  \bb
  &\E_\phi\chi^2\left(\mc N(0,\Sigma_{1,\phi}+V)\|\mc N(0,\Sigma_{0,\phi}+V)\right)
  \\=\,& \E_\phi \frac{1}{\sqrt{\left(1 + 2\varepsilon a^{-1} \bm e_2^T A^{-1} \bm e_2\right)\left(1 - 2\varepsilon a^{-1} \bm e_2^T A^{-1} \bm e_2\right)}} -1
  \\\eqt{(i)}\,&\E_\phi \left({{1 - 4\varepsilon^2 a^{-2} (A^{-1})^2_{22}}}\right)^{-\frac12} -1
  \\\leqt{(ii)}\,& \E_\phi 4\varepsilon^2 a^{-2} (A^{-1})^2_{22}.
  \ee
  Here (i) introduce the notation $(A^{-1})_{22}\coleq \bm e_2^T A^{-1} \bm e_2$; (ii) uses the inequality $(1-x)^{-\frac12}-1\le x$ for $x\in[0,4/9]$ and the fact that $\varepsilon\le 1/3$ and $(A^{-1})_{22}\le a$. The latter can be seen from
  \bb
    A \ge \begin{pmatrix}
      a & 0\\0 & a^{-1}
    \end{pmatrix} \implies A^{-1} \le \begin{pmatrix}
      a^{-1} & 0\\0 & a
    \end{pmatrix} \implies (A^{-1})_{22}\le a.  
  \ee
  The last step is to upper bound $\E_\phi (A^{-1})^2_{22}$. Without loss of generality, one can assume $V$ is diagonal as the rotation part can be absorbed into $R_\phi$. One can also assume $V$ is pure, as any measurement with mixed seed can be simulated by a measurement with pure seed followed by classical post-processing, which can only decrease the $\chi^2$-divergence by its data-processing inequality. Thus, we can write $$
    V = \begin{pmatrix}
      \frac12b & 0 \\ 0  & \frac12b^{-1}
    \end{pmatrix}, \quad \text{for some}~b>0.
  $$
  So we have
  \bb
  \E_\phi (A^{-1})^2_{22}&\leqt{(i)} a \E_\phi (A^{-1})_{22}
  \\&= a\,\frac{1}{\pi}\int_0^\pi\mr{d}\phi~ \frac{ab(a+b\cos^2\phi+b^{-1}\sin^2\phi)}{(a+b)^2+(a^2-1)(b^2-1)\sin^2\phi}
  \\&= \frac{a^2}{1+a}.
  \ee
  Here, (i) uses $(A^{-1})_{22}\le a$ again (We note that this relaxation is only loose by a factor of $1/2$ asymptotically as $b\rightarrow\infty$ for a fixed $a$).
  The other steps are just direct calculation. Putting everything together, we have
  \bb
  \frac29 &\le  N \cdot 4\varepsilon^2 a^{-2} \cdot \frac{a^2}{1+a} = N\cdot\frac{4\varepsilon^2}{1+a}.
  \ee
  That is
  \begin{equation}
    N \ge \frac29 \cdot \frac{1+a}{4\varepsilon^2}.
  \end{equation}
  Finally, by setting $a = E$, $\Tr \Sigma_{x,\phi} \le E$ for any $E>2$ and $x\in\{0,1\}$. We hence conclude that to learn an unknown Gaussian state with the promise $\Tr\Sigma\le E$ to to trace distance precision $\varepsilon' = \varepsilon/7 \le 1/21$ with probability at least 2/3, the necessary number of samples is
  \bb
    N \ge \frac29 \cdot \frac{1+E/2}{4(7\varepsilon')^2} = \Omega\left(E\varepsilon'^{-2}\right).
  \ee
  By renaming $\varepsilon'$ back to $\varepsilon$, we finish the proof.
\end{proof}

\subsection{Nearly-optimal non-adaptive Gaussian measurement protocol (1 mode)}\label{sec:non-ada-upper}

\begin{theorem}\label{th:non-ada-upper}
  There exists a non-adaptive, single-copy Gaussian measurement protocol that learns any single-mode Gaussian state $\rho(\mu,\Sigma)$ to trace distance precision $\varepsilon$ with probability at least $2/3$, given the promise that the condition number of $\Sigma$ satisfies $\kappa(\Sigma)\coleq \lambda_{\max}(\Sigma)/\lambda_{\min}(\Sigma) \le E^2$ for some $E>2$, using $N = O\left({E\log E/\varepsilon^{2}}\right)$ copies and running in $\mr{poly}(E, \varepsilon^{-2})$ time.
\end{theorem}

Note that any single mode Gaussian states $\rho(\mu,\Sigma)$ that satisfies $\Tr\Sigma \le E_\mr{tot}$ also satisfies $\kappa(\Sigma)\le 4E_\mr{tot}^2$, thanks to the physical covariance matrix condition $\det{\Sigma}\ge 1/4$. More explicitly, 
\bb
  \frac{1}{4\lambda_{\min}^2}\le\kappa\equiv\frac{\lambda_{\max}}{\lambda_{\min}} \le 4\lambda_{\max}^2 \le 4E_\mr{tot}^2.
\ee
Thus, the above theorem provides a matching upper bound for Theorem~\ref{th:non-ada-lower} up to logarithmic factors, showing that $\widetilde\Theta(E/\varepsilon^2)$ copies of $\rho$ are necessary and sufficient for any non-adaptive, single-copy Gaussian measurement protocol to learn a single-mode Gaussian state with energy constraint $E$. 
Though homodyne measurement is a well-studied tool for reconstructing bosonic quantum states, to our knowledge, our algorithm and analysis is the first non-adaptive protocol that achieves nearly-linear-in-$E$ sample complexity for learning single-mode Gaussian states, which is provably more efficient than heterodyne protocol that needs $\Theta(E^2/\varepsilon^2)$ copies (The upper bound is shown in~\cite{bittel2025energy} and the lower bound is proven in Theorem~\ref{th:heterodyne-lower}). 

\begin{figure}[!htp]
  \begin{algorithm}[H]
  \begin{algorithmic}[1]
      \caption{Non-adaptive learning of single-mode Gaussian states}
      \label{alg:non-ada}
      \Require{$E>2$. $\varepsilon>0$. $N= 2KT + N_0$ i.i.d. copies of $\rho(\mu,\Sigma)$ such that $\kappa(\Sigma)\le E^2$. Here $K = \Theta(E)$, $T= \Theta(\log E/\varepsilon^2)$, $N_0 = \Theta(1/\varepsilon^2)$.}
      \Ensure{With probability at least $2/3$, output a Gaussian state $\hat\rho$ such that $D_\mr{tr}(\hat\rho,\rho)\le \varepsilon$.}
      \Procedure{Non-adaptive Data Collection}{}
      \State Sample $K$ angles $\{\phi_i\}_{i=1}^{K}$ i.i.d. uniformly from $[0,\pi)$.
      \For{$i=1\cdots K$}
        \State Conduct homodyne measurement along $\phi_i$ on $2T$ copies of $\rho$, obtaining outcomes $\{z_{i,j}\}_{j=1}^{2T}$.
      \EndFor
      \State Conduct heterodyne measurement on $N_0$ copies of $\rho$, obtaining outcomes $\{(x_i,p_i)\}_{i=1}^{N_0}$.
      \EndProcedure
      \Procedure{Classical data processing}{}
      \For{$i=1\cdots K$}
        \State $\hat\mu_{\phi_i} \coleq \frac{1}{T}\sum_{j=1}^T z_{i,j}$; \quad $\hat\Sigma_{\phi_i}\coleq \frac{1}{2T}\sum_{j=1}^T (z_{i,j}-z_{i,j+T})^2 - E^{-1}$.
      \EndFor
      \State $\hat{\Sigma}_{\min}\coleq \min_{i\in[K]} \hat\Sigma_{\phi_i}$;\quad $\hat{\Sigma}_{\max}\coleq \max_{i\in[K]} \hat\Sigma_{\phi_i}$; \quad $\hat\kappa\coleq \hat\Sigma_{\max}/\hat\Sigma_{\min}$.
      \If{$\hat\kappa < 24$}

      \Return $\hat\rho$ via the non-adaptive heterodyne algorithm from \cite{bittel2025energy} using $\{(x_i,p_i)\}_{i=1}^{N_0}$.
      \Else
      \State $\phi_{\min} \coleq \mr{argmin}_{\phi_i:\,i\in[K]} \hat\Sigma_{\phi_i}$. %
      \State From $\{\phi_i\}_{i=1}^K$ find a $\phi_- \in [\phi_{\min} - {4}{{\hat\kappa}^{-\frac12}}, \phi_{\min} - {3}{{\hat\kappa}^{-\frac12}}]_{\pi}$ and a $\phi_+ \in [\phi_{\min} + {3}{{\hat\kappa}^{-\frac12}}, \phi_{\min} + {4}{{\hat\kappa}^{-\frac12}}]_{\pi}$. Abort the algorithm if either such $\phi_\pm$ does not exist.
      \State Obtain $\check\Sigma$ from $\left((\phi_{\min},\hat\Sigma_{\min}),(\phi_+,\hat\Sigma_{\phi_+}),(\phi_-,\hat\Sigma_{\phi_-})\right)$ using estimators defined in Eq.~\eqref{eq:solve-sigma}.
      \State $\varphi_1 \coleq \mr{argmin}_{\phi_i:i\in[K]} |\phi_i - \hat\theta|_{\pi};\quad \varphi_2 \coleq \mr{argmin}_{\phi_i:i\in[K]} |\phi_i - \hat\theta-\frac\pi2|_\pi.$
      \State $\hat\mu \coleq \hat\mu_{\varphi_1}\begin{pmatrix}\cos\varphi_1\\\sin\varphi_1\end{pmatrix} 
      + \dfrac{\hat\mu_{\varphi_2} - \hat\mu_{\varphi_1}\cos(\varphi_2-\varphi_1)}{\sin(\varphi_2-\varphi_1)}\begin{pmatrix}-\sin\varphi_1\\\cos\varphi_1\end{pmatrix}$. 

      \Return{$\hat\rho\coleq\hat\rho(\hat\mu,\check\Sigma)$.}
      \EndIf
      \EndProcedure
  \end{algorithmic}
  \end{algorithm}
  \end{figure}

\begin{figure}[!htp]
    \centering
    \includegraphics[width=0.98\linewidth]{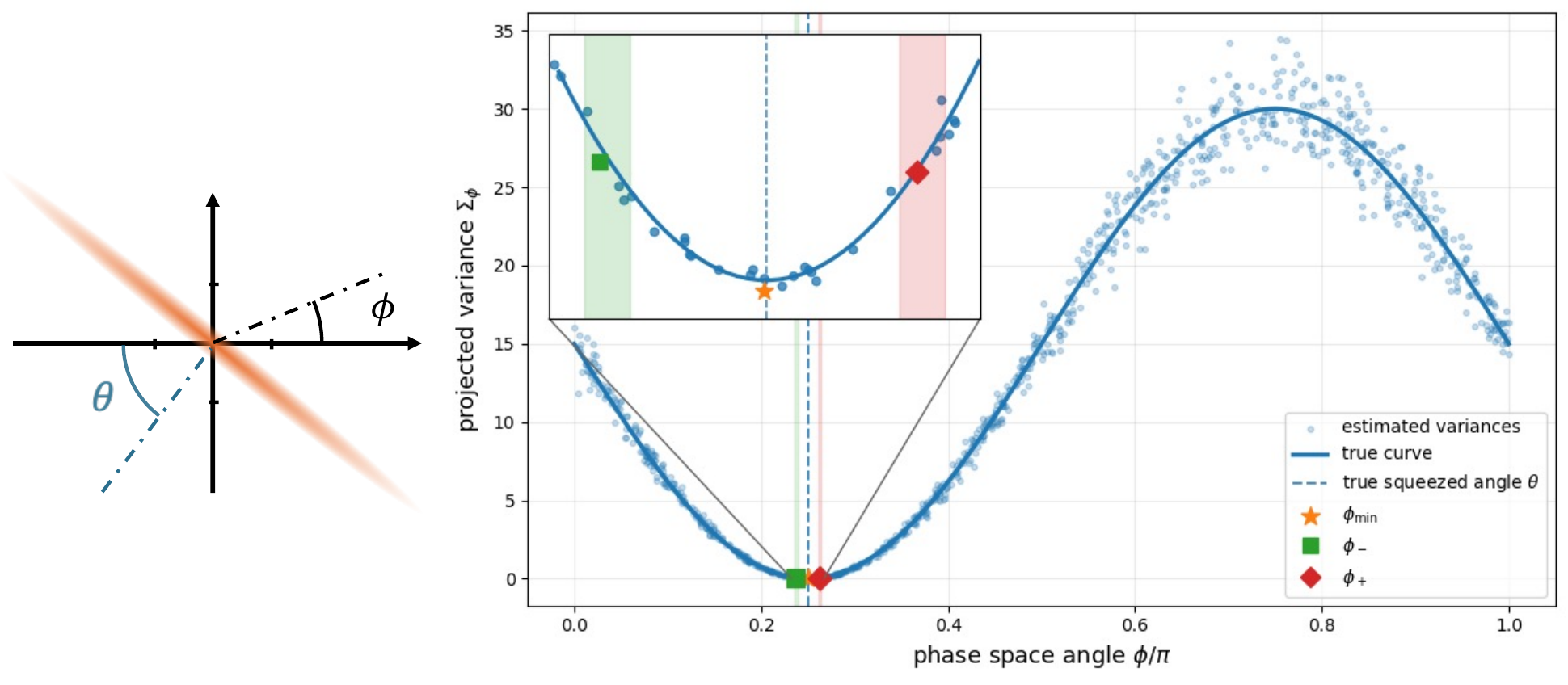}
    \caption{Example of Algorithm~\ref{alg:non-ada} on learning a one-mode zero-mean Gaussian state. The state is shown on the left hand side with parameter $E=60$, $b=1/2E$, $a=E/2$, and $\theta=\pi/4$. Here $\phi$ represents the angle of homodyne measurements (with added variance $1/E$). On the right hand side, the solid line is the true quadrature-projected variance $\Sigma_\phi$ scanned over $\phi\in[0,\pi)$.
    We sample $K=1000$ uniformly random homodyne angles and measure each with $500$ shots to obtain the dots (using Lemma~\ref{le:homodyne-concen}). 
    The inset illustrates how Algorithm~\ref{alg:non-ada} will pick three sampled angles $\phi_\mr{min}$, $\phi_-$ and $\phi_+$, which have sufficiently small shot noise and are sufficiently separated, to solve for the model parameters.
    }
    \label{fig:non-ada-alg}
\end{figure}

  Our protocol is described in Algorithm~\ref{alg:non-ada}. At a high level, the protocol conduct homodyne measurements along a few random phase space directions on $\Theta(E\log E/\varepsilon^2)$ copies and heterodyne measurements on $\Theta(1/\varepsilon^2)$ copies. Based on the homodyne outcomes, the algorithm first roughly estimates $\kappa(\Sigma)$. Denote the estimator by $\hat\kappa$. If $\hat\kappa$ is smaller than a constant, one simply uses the heterodyne estimator from~\cite{bittel2025energy} which suffices to solve the learning tasks using $\Theta(1/\varepsilon^2)$ copies; Otherwise, one uses the randomized homodyne data to first roughly estimates the highly squeezed direction of $\Sigma$, and then solves for $\Sigma$ and $\mu$ using homodyne outcomes along a few carefully chosen directions. The overhead of $E$ in sample complexity is to ensure such ``good'' homodyne angles exists among the random angles sampled with high probability. 
  In Fig.~\ref{fig:non-ada-alg}, we provide an intuitive illustration about how Algorithm~\ref{alg:non-ada} works.

  \medskip

  Before proving the correctness of Algorithm~\ref{alg:non-ada}, let us introduce a few definitions and lemmas. For a general single-mode Gaussian state $\rho(\mu,\Sigma)$, let $a\coleq\lambda_{\max}$ and $b\coleq\lambda_{\min}$ for notational simplicity, the covariance matrix can be parametrized as
  \bb
    \Sigma = R_\theta \begin{pmatrix}
      b & 0\\ 0 & a
    \end{pmatrix} R_\theta^T,\quad\text{where}~R_\theta = \begin{pmatrix}
      \cos\theta & -\sin\theta\\
      \sin\theta & \cos\theta
    \end{pmatrix},\quad\theta\in[0,\pi).
  \ee
  For any $\phi\in[0,\pi)$, an ideal homodyne measurement along angle $\phi$ is a Gaussian measurement with seed $V_{\phi,r} = R_\phi \mr{Diag}[c^{-1},c] R_\phi^T$ with $c\to\infty$, and then project the measurement outcome along $\bm e_{\phi}$.  
  For our purpose, it suffices to consider non-ideal homodyne with, say, $c = E$. We make this choice of $c$ for all homodyne measurement used in the following.
  \begin{lemma}\label{le:homodyne-outcomes}
    Let $z\in\mbb R$ be the outcome of a homodyne measurement with squeezing parameter $c=E$ along $\phi$ on $\rho(\mu,\Sigma)$. Then, $z\sim\mc N(\mu_\phi,\Sigma_\phi+ E^{-1})$ where 
    \bb
      \mu_\phi &= \bm e_{\phi}^T\mu = \mu_x\cos\phi + \mu_p\sin\phi,\\ \Sigma_{\phi} &= \bm e_{\phi}^T\Sigma \bm e_{\phi} = b + (a-b)\sin^2(\theta-\phi).
    \ee
  \end{lemma}
  \begin{proof}[Proof of Lemma~\ref{le:homodyne-outcomes}]
    The outcome of the Gaussian seed measurement is given by 
    \bb
    \begin{pmatrix}
      z_x \\ z_p
    \end{pmatrix} \sim \mc N\!\left(\begin{pmatrix}
      \mu_x \\ \mu_p
    \end{pmatrix},~ R_{\theta}\begin{pmatrix}
      b &0 \\ 0 & a
    \end{pmatrix}R_{\theta}^T +  R_{\phi}\begin{pmatrix}
      E^{-1} &0 \\ 0 & E
    \end{pmatrix}R_{\phi}^T\right).
    \ee
    Projecting this to $\bm e_{\phi}\coleq(\cos\phi,\sin\phi)^T$ then gives $z\sim\mc N(\mu_\phi,\Sigma_\phi + E^{-1})$ as claimed.
  \end{proof}
  \begin{lemma}\label{le:homodyne-concen}
    There exists an absolute constant $C>0$ such that the following holds: 
    Given any $\phi\in[0,\pi)$ and $\varepsilon,\delta\in(0,1)$, let $z_1,\cdots,z_{2T}$ be i.i.d. samples from the homodyne measurement along $\phi$ on $\rho(\mu,\Sigma)$. Define the estimators
    \begin{equation}
      \hat\mu_\phi\coleq \frac{1}{T}\sum_{j=1}^T z_j,\quad \hat\Sigma_\phi \coleq \frac{1}{2T}\sum_{j=1}^T (z_j - z_{j+T})^2 - E^{-1}.
    \end{equation}
    Then, by taking $T = \ceil{C\varepsilon^{-2}\log\delta^{-1}}$, it holds with probability at least $1-\delta$ that
    \bb
      |\hat\mu_{\phi} - \mu_\phi| \le \varepsilon\sqrt{\Sigma_{\phi}},\quad |\hat\Sigma_{\phi} - \Sigma_{\phi}| \le \varepsilon \Sigma_{\phi}.
    \ee
  \end{lemma}
  \begin{proof}[Proof of Lemma~\ref{le:homodyne-concen}]
    Using standard concentration bounds for Gaussian and Chi-square distributions (see e.g. \cite[Appendix C]{ashtiani2020near}), it is easy to see that $T=\ceil{C\varepsilon^{-2}\log\delta^{-1}}$ suffices to ensure $|\hat\mu_{\phi} - \mu_\phi| \le \varepsilon\sqrt{\Sigma_{\phi}+E^{-1}}$ and $|\hat\Sigma_{\phi} - \Sigma_{\phi}| \le \varepsilon (\Sigma_{\phi}+E^{-1})$. The lemma follows by realizing that $\Sigma_\phi\ge b\ge (4E)^{-1}$ and then rescaling $\varepsilon$ by a constant.
  \end{proof}

  \medskip

  \begin{proof}[Proof of Theorem~\ref{th:non-ada-upper}]
    We analyze Algorithm~\ref{alg:non-ada}. For simplicity, we assume $K = C_1 E$, $T= C_2 \log E/\varepsilon^2$, $N_0 = C_3 /\varepsilon^2$ where $C_1,C_2,C_3$ are sufficiently large absolute constants. 
    We also without loss of generality assume $\varepsilon$ is upper bounded by a sufficiently small constant.
    We will not try to figure out the exact values of all these constants, though this can be done through a more dedicated analysis.
  
    \medskip

    \noindent\textbf{Part 1: Rough estimate of condition number $\kappa$.} We first analyze the estimators $\hat\Sigma_{\min}$, $\hat\Sigma_{\max}$, and $\hat\kappa$ defined in Step~8. 
    We claim that with high probability, there exist $i_1, i_2\in[K]$ such that 
    \bb
    |\phi_{i_1} - \theta|_{\pi}\le 0.1\kappa^{-1/2},\quad |\phi_{i_2} - \theta - \pi/2|_\pi\le 0.1\kappa^{-1/2}.
    \ee
    Here $|x|_\pi\coleq \min_{k\in\mbb Z}|x-k\pi|$ denotes the distance of $x$ to the nearest multiple of $\pi$. 
    This holds because the probability for a randomly sampled $\phi\sim\mr{Unif}[0,\pi)$ to satisfy $|\phi - \theta|_{\pi}\le 0.1\kappa^{-1/2}$ is $0.1\kappa^{-1/2}/\pi\ge \frac{0.1}{2\pi} E^{-1}$, and that we have sampled $K=C_1 E$ angles for sufficiently large $C_1$. 
    Combining this with the homodyne concentration results from Lemma~\ref{le:homodyne-concen} (assuming $\varepsilon\le 1/400$) and using a union bound (note that the factor of $\log E$ in $T$ ensures the estimates along all $K$ angles succeed simultaneously with high probability), we have with high probability that
    \bb\label{eq:sigma-bound-s1}
      \hat\Sigma_{\min} &\le \hat\Sigma_{\phi_{i_1}} \le (1+\varepsilon)\Sigma_{\phi_{i_1}}\le (1+\varepsilon)\left(b + (a-b)\sin^2(\phi_{i_1}-\theta)\right)\le 1.013 b,\\
      \hat\Sigma_{\max} &\ge \hat\Sigma_{\phi_{i_2}} \ge (1-\varepsilon)\Sigma_{\phi_{i_2}}\ge (1-\varepsilon)\left(a - (a-b)\cos^2(\phi_{i_2}-\theta)\right) \ge 0.99 a.\\
    \ee
    We have used $\sin^2x\le |x|_{\pi}^2$.
    We also have $\hat\Sigma_{\min}\ge(1-\varepsilon)b \ge 0.99b$ and $\hat\Sigma_{\max}\le(1+\varepsilon)a\le 1.01a$. We thus conclude that
    \bb
      |\hat\Sigma_{\min} - b| \le 0.02b,\quad |\hat\Sigma_{\max} - a| \le 0.02a \quad \Longrightarrow\quad |\hat\kappa-\kappa|=\left|\frac{\hat\Sigma_{\max}}{\hat\Sigma_{\min}}-\frac{a}{b}\right|\le0.05\kappa.
    \ee
    Now, if the branching condition in Step~10 (i.e. $\hat\kappa < 24$) holds true, then we must have $b^{-1}\le 2\kappa < 50$. The heterodyne protocol analyzed in \cite{bittel2025energy} then returns an estimate $\hat\rho$ such that $D_\mr{tr}(\hat\rho,\rho)\le \varepsilon$ with probability at least $0.99$ using $N_0 = \Theta(1/\varepsilon^2)$ copies, completing the proof for this case.

    \medskip
    \noindent\textbf{Part 2: Rough estimate of squeezing direction $\theta$.}
    Now consider the case when $\hat\kappa \ge 24$, which allows us to safely conclude $\kappa>20$ with high probability. We claim that $\phi_{\min}\coleq\mr{argmin}_{\phi_i}\hat\Sigma_{\phi_i}$ defined in Step~12 must satisfy
    \bb
      |\phi_{\min}-\theta|_\pi<0.16 \hat\kappa^{-1/2}.
    \ee
    Indeed, consider any sampled angle $\phi_i$ such that $d_i\coleq|\phi_i-\theta|_\pi\ge0.16\hat\kappa^{-1/2}$. Since $d_i\in[0,\pi/2]$ and $0.16\hat\kappa^{-1/2}\le0.16/\sqrt{24}<1/2$, monotonicity of $\sin^2x$ on $[0,\pi/2]$ gives
    \bb
      \sin^2d_i\ge\sin^2(0.16\hat\kappa^{-1/2})\ge0.9(0.16\hat\kappa^{-1/2})^2.
    \ee
    Consequently,
    \bb\label{eq:rough-angles-1}
      \hat\Sigma_{\phi_i}
      &\ge (1-\varepsilon)\Sigma_{\phi_i}\\
      &\ge (1-\varepsilon)\left(b+0.9(0.16)^2(a-b)\hat\kappa^{-1}\right)\\
      &> \frac{399}{400}\left(1+0.9(0.16)^2\frac{19}{21}\right)b
      >1.018b>1.013b\ge\hat\Sigma_{\min}.
    \ee
    Here we used $\varepsilon\le1/400$ and
    \bb
      \frac{a-b}{\hat\kappa b}=\frac{\kappa-1}{\hat\kappa}
      >\frac{(19/20)\kappa}{(21/20)\kappa}=\frac{19}{21},
    \ee
    which follows from $\kappa>20$ and $\hat\kappa\le1.05\kappa$. Therefore, such a $\phi_i$ cannot be $\phi_{\min}$, proving the claim by contradiction.

    We will also use the following consequence of Eq.~\eqref{eq:sigma-bound-s1} conditioned on the success event
    \bb
      \Sigma_{\phi_{\min}}
      \le \frac{\hat\Sigma_{\min}}{1-\varepsilon}
      \le \frac{1.013}{399/400}b<1.016b<1.02b.
    \ee

    \medskip
    \noindent\textbf{Part 3: Precise estimate of $\Sigma$.} Now that we have a rough estimate of $\theta$ via $\phi_{\min}$, we can post-select homodyne angles close to $\phi_{\min}$ to solve for $\Sigma$, as they have the smallest additive errors (see Lemma~\ref{le:homodyne-concen}). While there may exist more systematic methods to approach this problem such as generalized least squares~\cite{wainwright2019high}, here we adopt a more elementary approach.

    In Step~13, we seek to find two angles $\phi_+$ and $\phi_-$ from all sampled angles $\{\phi_i\}_{i=1}^K$ that satisfy
    \bb
      \phi_+ \in [\phi_{\min} + 3\hat\kappa^{-1/2}, \phi_{\min} + 4\hat\kappa^{-1/2}]_{\pi},\quad \phi_- \in [\phi_{\min} - 4\hat\kappa^{-1/2}, \phi_{\min} - 3\hat\kappa^{-1/2}]_{\pi}.
    \ee
    Here $[a,b]_{\pi}\coleq \bigcup_{k\in\mbb Z}[a+k\pi,b+k\pi]\cap[0,\pi)$, i.e., the interval wrapped around $[0,\pi)$. Similar to part 1, such $\phi_\pm$ exist with high probability thanks to our choice of $K$. 
    Also recall that $|\phi_{\min} - \theta|_\pi<0.16\hat\kappa^{-1/2}$.
    The corresponding homodyne estimators satisfy
    \bb
      \left|\hat\Sigma_{\phi_{\pm}} - \Sigma_{\phi_\pm}\right|&\le \varepsilon(b + (a-b)\sin^2(\phi_\pm-\theta)) \leqt{(i)} \varepsilon(b + 18 a \hat\kappa^{-1}) < 20\varepsilon b.\\
      \left|\hat\Sigma_{{\min}} - \Sigma_{\phi_{\min}}\right|&\le \varepsilon\Sigma_{\phi_{\min}} \le 1.02 \varepsilon b.
    \ee
    Here (i) uses $|\sin(\phi_\pm-\theta)|\le |\phi_\pm-\theta|_\pi \le |\phi_\pm-\phi_{\min}|_\pi+|\phi_{\min}-\theta|_\pi < 4.2\hat\kappa^{-1/2}$. The final inequality also uses $\hat\kappa\ge0.95\kappa$, so that $a\hat\kappa^{-1}\le(20/19)b$. Note that they indeed have very small additive errors, which are proportional to $\varepsilon b$.

    \medskip
    \noindent Define the following estimators that are used in Step~14:
    \bb\label{eq:solve-sigma}
      \hat\theta&\coleq f^{-1}\left(\frac{\hat\Sigma_{\phi_+} - \hat\Sigma_{\min}}{\hat\Sigma_{\phi_-} - \hat\Sigma_{\min}}\right),\quad
      \hat\Delta\coleq \frac{\hat\Sigma_{\phi_+}-\hat\Sigma_{\min}}{\sin^2(\phi_+-\hat\theta) - \sin^2(\phi_{\min}-\hat\theta)},\\
      \hat b &\coleq \hat\Sigma_{\min} - \hat\Delta \sin^2(\phi_{\min}-\hat\theta),\quad
      \hat a \coleq \hat b + \hat\Delta,\\
      \hat\Sigma&\coleq R_{\hat\theta}\begin{pmatrix}
        \hat b & 0\\ 0 & \hat a
      \end{pmatrix}R_{\hat\theta}^T,\\
      \check\Sigma&\coleq \frac{\hat\Sigma}{1-\eta},\quad\textrm{where~$\eta<1$~is to be fixed later.}
    \ee
    Here, 
    \bb
    f(\theta)&\coleq \frac{\Sigma_{\phi_+} - \Sigma_{\phi_{\min}}}{\Sigma_{\phi_-} - \Sigma_{\phi_{\min}}} = \frac{\sin^2(\phi_+ -\theta) - \sin^2(\phi_{\min} - \theta)}{\sin^2(\phi_- -\theta) - \sin^2(\phi_{\min} - \theta)}.
    \ee
    \begin{lemma}\label{le:estimators}
      There exists an absolute positive constant $c$ such that the following holds:
      \bb
        |\hat\theta-\theta|_\pi\le c\varepsilon\kappa^{-1/2},\quad |\hat a - a|\le c\varepsilon a, \quad |\hat b - b|\le c\varepsilon b.
      \ee
    \end{lemma}
    \noindent The proof is elementary. We defer it to the end of this section. Given this, we now assess $\hat\Sigma$. We first claim that for some constant $c'$:
    \bb
      (1-c'\varepsilon)\begin{pmatrix}
        \hat b & 0\\ 0 & \hat a
      \end{pmatrix}
      \le 
      R_{\hat\theta - \theta}\begin{pmatrix}
        \hat b & 0\\ 0 & \hat a
      \end{pmatrix}R^T_{\hat\theta - \theta}
      \le
      (1+c'\varepsilon)\begin{pmatrix}
        \hat b & 0\\ 0 & \hat a
      \end{pmatrix}.
    \ee
    This can be verified by brute-force calculation,
    \bb
      &~\left\|   \begin{pmatrix}
        \hat b & 0\\ 0 & \hat a
      \end{pmatrix}^{-\frac12}   R_{\hat\theta - \theta}\begin{pmatrix}
        \hat b & 0\\ 0 & \hat a
      \end{pmatrix}R^T_{\hat\theta - \theta}  \begin{pmatrix}
        \hat b & 0\\ 0 & \hat a
      \end{pmatrix}^{-\frac12} - I \right\|_\mr{op}  \\=&~ \frac{2}{-1+\sqrt{1+\frac{4\hat a\hat b}{(\hat a-\hat b)^2\sin^2(\hat\theta-\theta)}}}
       \le \frac{2}{-1 + \sqrt{1 + \frac{4(1-c\varepsilon)^2}{(1+c\varepsilon)^2c^2\varepsilon^2}}}\le c'\varepsilon.
    \ee
    Then,
    \bb
      &\left| \Sigma^{-\frac12}\hat\Sigma\Sigma^{-\frac12} - I \right| \le  c'\varepsilon R_\theta \begin{pmatrix}
        \frac{\hat b}b & 0\\
        0 & \frac{\hat a}a
      \end{pmatrix}R_\theta^T
      \le 2c'\varepsilon I,\\
      \Longrightarrow~&  (1-2c'\varepsilon)\Sigma \le\hat\Sigma \le(1+2c'\varepsilon)\Sigma.
    \ee
    Now we fix $\eta$ in the definition of $\check\Sigma$ as $\eta\coleq 2c'\varepsilon$. That is, $\check\Sigma\coleq \hat\Sigma/(1-2c'\varepsilon)$. We then have,
    \bb\label{eq:check_Sigma}
      \Sigma \le \check\Sigma \le (1 + 5c'\varepsilon)\Sigma,
    \ee
    as long as $\varepsilon$ is upper bounded by a sufficiently small constant.
    This suffices for our purpose. We will come back to this later.

    \medskip
    \noindent\textbf{Part 4: Precise estimate of $\mu$.} We now analyze the estimator for $\mu$ defined in Steps~15-16:
    \begin{equation}
    \begin{gathered}
    \hat\mu \coleq \hat\nu_{1}\begin{pmatrix}\cos\varphi_1\\\sin\varphi_1\end{pmatrix} 
    + \hat\nu_2\begin{pmatrix}-\sin\varphi_1\\\cos\varphi_1\end{pmatrix},\\
    \text{where}\quad\hat\nu_1\coleq\hat\mu_{\varphi_1}\quad\hat\nu_2\coleq\dfrac{\hat\mu_{\varphi_2} - \hat\mu_{\varphi_1}\cos(\varphi_2-\varphi_1)}{\sin(\varphi_2-\varphi_1)}.
    \end{gathered}
    \end{equation}
    Note that $\mu$ can also be decomposed as
    \begin{equation}
    \begin{gathered}
    \mu \coleq \nu_{1}\begin{pmatrix}\cos\varphi_1\\\sin\varphi_1\end{pmatrix} 
    + \nu_2\begin{pmatrix}-\sin\varphi_1\\\cos\varphi_1\end{pmatrix},\\
    \text{where}\quad \nu_1\coleq\mu_{\varphi_1}\quad\nu_2\coleq\dfrac{\mu_{\varphi_2} - \mu_{\varphi_1}\cos(\varphi_2-\varphi_1)}{\sin(\varphi_2-\varphi_1)}.
    \end{gathered}
    \end{equation}
    One can verify this by checking $\bm e_{\varphi_1}\cdot\mu = \mu_{\varphi_1}$ and $\bm e_{\varphi_2}\cdot\mu = \mu_{{\varphi_2}}$, and note that $\bm e_{\varphi_1}$ and $\bm e_{\varphi_2}$ are linearly independent. By definition of $\varphi_1$ and $\varphi_2$, the fact that $|\hat\theta-\theta|\le c\varepsilon\kappa^{-1/2}$, and our choice of $K=C_1 E$, we have with high probability that
    \bb
      |\varphi_1 - \theta|_\pi \le \kappa^{-1/2},\quad |\varphi_2 - \theta - \pi/2|_\pi \le \kappa^{-1/2}.
    \ee
    Thus, recall the concentration results from Lemma~\ref{le:homodyne-concen}:
    \bb
    |\hat\nu_1-\nu_1| &\le  \varepsilon\sqrt{\Sigma_{\varphi_1}} \le \varepsilon \sqrt{b + a\kappa^{-1}} =\varepsilon\sqrt{2b}.\\
    |\hat\nu_2 - \nu_2| &\leqt{(i)} \frac{\varepsilon\sqrt{a} + 0.43\varepsilon\sqrt{2b}}{0.9} \le 2\varepsilon\sqrt{a}.
    \ee
    Here (i) uses $|\varphi_2-\varphi_1|_\pi \ge \frac{\pi}{2} - 2\kappa^{-1/2}\ge 1.12$ to bound the sine and cosine factors.
    Combining all these, we have that
    \bb
    (\hat\mu -\mu)^T\Sigma^{-1}(\hat\mu -\mu) &= b^{-1} ((\hat\nu_1 -\nu_1)\bm e_{\varphi_1}\cdot\bm e_{\theta} + (\hat\nu_2 -\nu_2)\bm e_{\varphi_1}^\perp\cdot\bm e_{\theta})^2  \\&\quad + a^{-1} ((\hat\nu_1 -\nu_1)\bm e_{\varphi_1}\cdot\bm e_{\theta}^\perp + (\hat\nu_2 -\nu_2)\bm e_{\varphi_1}^\perp\cdot\bm e_{\theta}^\perp)^2 
    \\&\le 2b^{-1}(2\varepsilon^2b + 4\varepsilon^2a\kappa^{-1}) + 2a^{-1}(2\varepsilon^2b\kappa^{-1} + 4\varepsilon^2a)
    \\&\le 24\varepsilon^2.
    \ee

    \noindent\textbf{Part 5: Put everything together.} 
    Combining all the previous parts and use Lemma~\ref{le:perturbation}, we obtain the following with high probability
    \bb
      D_\mr{tr}(\hat\rho(\hat\mu,\check\Sigma),\rho(\mu,\Sigma)) &\le \frac{1}{\sqrt 2}\|\Sigma^{-\frac12}(\hat\mu - \mu)\|_2 + \frac{1+\sqrt3}{8}\Tr\left((\Sigma^{-1} + \check\Sigma^{-1})|\check\Sigma - \Sigma|\right)\\
      &\le 2\sqrt{3}\varepsilon + \frac{1+\sqrt3}{8}\Tr\left(2\Sigma^{-1}\cdot 5c'\varepsilon\Sigma\right)\\  
      &\le c'''\varepsilon,
    \ee
    for some constant $c'''>0$. The second line is by Eq.~\eqref{eq:check_Sigma}. Rescaling $\varepsilon$ by a constant then completes the proof for Theorem~\ref{th:non-ada-upper}.
    \end{proof}

    \noindent We finalize the proof by proving Lemma~\ref{le:estimators} in below:
    \begin{proof}[Proof of Lemma~\ref{le:estimators}]
    We now analyze the error for these estimators. 
    Taking derivative of $f$,
    \bb\label{eq:error-f-deriv}
      f'(\theta) &= -\frac{2\sin(\phi_+ - \phi_{\min})\sin(\phi_+ -\phi_-)}{\sin^2(2\theta - \phi_{\min}-\phi_-)\sin(\phi_{\min}-\phi_-)}
      \\&\leqt{(i)} - \left(\frac2\pi\right)^2 \frac{2|\phi_+-\phi_{\min}|_\pi|\phi_+-\phi_-|_\pi}{|2\theta-\phi_{\min}-\phi_-|_\pi^2|\phi_{\min}-\phi_-|_\pi}
      \\& \leqt{(ii)} - \left(\frac2\pi\right)^2 \frac{2 \cdot 3\hat\kappa^{-1/2} \cdot 6\hat\kappa^{-1/2}}{ (4.32)^2\hat\kappa^{-1} \cdot 4\hat\kappa^{-1/2}}
      \\&\le -0.19 \kappa^{1/2}.
    \ee
    Here (i) uses $\sin x \ge \frac2{\pi} x$ for all $x\in[0,\pi/2]$. (ii) uses $|2\theta-\phi_\mr{min}-\phi_-|_\pi \le |\phi_\mr{min}-\phi_-| + 2|\theta-\phi_\mr{min}|$.
    Meanwhile, denote $\delta_\pm\coleq \Sigma_{\phi_\pm}-\Sigma_{\phi_{\min}}$ and $\hat\delta_\pm\coleq \hat\Sigma_{\phi_\pm} - \hat\Sigma_{\min}$, we have
    \bb
    &\delta_\pm = (a-b)\left(\sin^2(\phi_\pm - \theta) - \sin^2(\phi_{\min} - \theta)\right) \in [2.85 b,25b].\\
    &|\hat\delta_\pm -\delta_\pm| \le |\hat\Sigma_{\phi_\pm} - \Sigma_{\phi_\pm}| + |\hat\Sigma_{\min} - \Sigma_{\phi_{\min}}| \le 27\varepsilon b.
    \ee
    Thus, for sufficiently small $\varepsilon$,
    \bb\label{eq:error-delta-ratio}
      \left|\frac{\hat\delta_+}{\hat\delta_-} - \frac{\delta_+}{\delta_-}\right| &= \left|\frac{\hat\delta_+\delta_- - \delta_+\hat\delta_-}{\delta_-\hat\delta_-}\right| \\
      &= \left|\frac{\hat\delta_+\delta_- - \delta_+\delta_- + \delta_+\delta_- -\delta_+\hat\delta_- }{\delta_-\hat\delta_-}\right| \\
      &\le \frac{\delta_-|\hat\delta_+ - \delta_+|}{|\delta_-\hat\delta_-|} + \frac{\delta_+|\hat\delta_- - \delta_-|}{|\delta_-\hat\delta_-|}
      \\&\le 110 \varepsilon.
    \ee
    Combining Eqs.~\eqref{eq:error-f-deriv} and~\eqref{eq:error-delta-ratio}, we have
    \bb
    \left|\hat\theta-\theta\right|_\pi = \left|f^{-1}\!\left(\frac{\hat\delta_+}{\hat\delta_-}\right) - f^{-1}\!\left(\frac{\delta_+}{\delta_-}\right)\right| \le 600 \varepsilon\kappa^{-1/2}.
    \ee
    Next, we analyze estimate of $\Delta\coleq a - b$. Note that:
    \begin{equation}
      \begin{gathered}
        3\kappa^{-1} \le \sin^2(\phi_+-\theta) - \sin^2(\phi_{\min}-\theta) \le 25\kappa^{-1}\\
        | \sin^2(\phi_+-\theta)-\sin^2(\phi_+-\hat\theta)|\le 63\kappa^{-1/2}|\theta-\hat\theta| \le 37800\varepsilon\kappa^{-1}.\\
        | \sin^2(\phi_{\min}-\theta)-\sin^2(\phi_{\min}-\hat\theta)|\le 6\kappa^{-1/2}|\theta-\hat\theta| \le 3600\varepsilon\kappa^{-1}.
      \end{gathered}
    \end{equation}
    Therefore, using similar steps as in Eq.~\eqref{eq:error-delta-ratio}, we have
    \bb
      |\hat\Delta - \Delta| & \le \frac{27\varepsilon b}{(3-c\varepsilon)\kappa^{-1}} + \frac{25b \cdot c\varepsilon\kappa^{-1}}{(3-c\varepsilon)\kappa^{-1}\cdot 3\kappa^{-1}} \le c' \varepsilon a.
    \ee
    Here $c,c'$ is some positive constant, and we need to assume $\varepsilon<c'$.
  
    \medskip

    For $\hat b$, we have
    \bb
      |\hat b - b| &\le |\hat\Sigma_{\min} - \Sigma_{\phi_{\min}}| + |\hat\Delta - \Delta|\sin^2(\phi_{\min}-\hat\theta) + \Delta\left|\sin^2(\phi_{\min}-\hat\theta) - \sin^2(\phi_{\min}-\theta)\right|\\
      &\le 1.02\varepsilon b + c'\varepsilon a \cdot 4\kappa^{-1} + a \cdot c\varepsilon \kappa^{-1}
      \\&\le c''\varepsilon b.
    \ee
    where $c''$ is some positive constant. Finally, for $\hat a$, 
    \bb
      |\hat a - a| \le |\hat b - b| + |\hat\Delta - \Delta| \le (c'+c'')\varepsilon a.
    \ee
  \end{proof}

\phantomsection
\addcontentsline{toc}{section}{References}
\bibliography{ref.bib}
\bibliographystyle{alpha}

\end{document}